\documentclass[sigconf, nonacm]{acmart}

\newcommand\vldbdoi{XX.XX/XXX.XX}

\newcommand\vldbavailabilityurl{}
\newcommand\vldbpagestyle{empty}

\usepackage{amsmath, bm}
\usepackage{graphicx}
\usepackage[capitalize]{cleveref}
\usepackage{latexsym}
\usepackage{fancyvrb}
\usepackage{subcaption}
\usepackage{paralist}
\usepackage{listings}
\usepackage[skip=0.2pt]{caption}
\usepackage{subcaption}
\usepackage[vlined,ruled,commentsnumbered]{algorithm2e}
\usepackage{hhline,booktabs}
\usepackage{xcolor,colortbl}
\usepackage{multirow}
\usepackage{pgfplots}
\usetikzlibrary{calc,angles,intersections,quotes,arrows,automata,fit,shapes,snakes,trees,topaths,calc,shapes.geometric,arrows.meta,positioning,positioning,chains,fit,shapes,calc,matrix}

\newcommand{\paratitle}[1]{\noindent{\bf #1}.}

\newcommand{\qcount}{\ensuremath{q_1}}
\newcommand{\qsum}{\ensuremath{q_2}}
\newcommand{\qmed}{\ensuremath{q_3}}
\newcommand{\whereClause}{\ensuremath{\varphi}}

\newtheorem{theorem}{Theorem}[section]
\newtheorem{proposition}[theorem]{Proposition}
\newtheorem{observation}[theorem]{Observation}
\newtheorem{example}[theorem]{Example}

\newtheorem{definition}[theorem]{Definition}

\newcommand{\oursys}{\textsc{DP-PQD}}
\newcommand{\countQuery}{\texttt{COUNT}}
\newcommand{\medianQuery}{\texttt{MEDIAN}}
\newcommand{\sumQuery}{\texttt{SUM}}

\newcommand{\alg}{\ensuremath{\mathcal{A}}} % for algos that decide if distance bound is met
\newcommand{\distBd}{\ensuremath{d_q}}
\newcommand{\intvl}{\ensuremath{\mathcal{I}}}

\newcommand{\BasicDecider}{\FuncSty{GenericDecider}} %{\FuncSty{BasicDecider}}
\newcommand{\dpmech}{\FuncSty{DPNoisy}}
\newcommand{\param}{\ensuremath{\phi}}

\newcommand{\lapcount}{\ensuremath{LM_{count}}}
\newcommand{\emcount}{\ensuremath{EM_{count}}}
\newcommand{\emcountNaive}{\ensuremath{EM_{count}^{naive}}}
\newcommand{\emmed}{\ensuremath{EM_{med}}}
\newcommand{\histmed}{\ensuremath{Hist_{med}}}
\newcommand{\GSbound}{\ensuremath{GS_q}}
\newcommand{\DLS}{\ensuremath{DS_{q, D}}}
\newcommand{\lapsum}{\ensuremath{LM_{sum}}}
\newcommand{\rtwotsum}{\ensuremath{R2T_{sum}}}
\newcommand{\svtsum}{\ensuremath{SVT_{sum}}}

\newcommand{\DPbayes}{\ensuremath{\hat{\mathcal{N}}}}

\newcommand{\revm}[1]{{\color{black} #1}} % red
\newcommand{\reva}[1]{{\color{black}  #1}} % purple
\newcommand{\revb}[1]{{\color{black}   #1}} % blue
\newcommand{\revc}[1]{{\color{black}    #1}} % magenta

\begin{document}
% \title{Detecting Per-Query Gaps for \\
% Black-Box Differentially Private Synthetic Data}
\title{\oursys: %A Framework for 
Privately Detecting Per-Query Gaps In \\
Synthetic Data Generated By Black-Box Mechanisms}

%%
%% The "author" command and its associated commands are used to define the authors and their affiliations.
% \author{
%     Shweta Patwa,
%     Danyu Sun,
%     Amir Gilad$^\dagger$,
%     Ashwin Machanavajjhala,
%     Sudeepa Roy
% }
% \affiliation{%
%   \institution{Duke University, 
%   % , Department of Computer Science, 
%   $^\dagger$The Hebrew University}
% }
% \email{
%     {sjpatwa,ds592,ashwin,sudeepa}@cs.duke.edu, amirg@cs.huji.ac.il
% }

\settopmatter{authorsperrow=5}

\author{Shweta Patwa}
\affiliation{%
  \institution{Duke University}
}
\email{sjpatwa@cs.duke.edu}

\author{Danyu Sun}
\affiliation{%
  \institution{Duke University}
}
\email{ds592@cs.duke.edu}

\author{Amir Gilad}
\affiliation{%
  \institution{Hebrew University}
}
\email{amirg@cs.huji.ac.il}

\author{Ashwin Machanavajjhala}
\affiliation{%
  \institution{Duke University}
}
\email{ashwin@cs.duke.edu}

\author{Sudeepa Roy}
\affiliation{%
  \institution{Duke University}
}
\email{sudeepa@cs.duke.edu}

\begin{abstract}
Synthetic data generation methods, and in particular, private synthetic data generation methods, are gaining popularity as a means to make copies of sensitive databases that can be shared widely for research and data analysis. Some of the fundamental operations in data analysis include  analyzing aggregated statistics, e.g., count, sum, or median, on a subset of data satisfying some conditions. 
%, and training AI/ML models.
When synthetic data is generated, users may be interested in knowing if their aggregated queries generating such statistics
%, e.g., count, sum, or median, 
can be reliably answered on the synthetic data, for instance, to decide if the synthetic data is  
%Such information is often useful to detect if the synthetic data is 
suitable for specific tasks. 
However, the standard data generation systems %typically 
do not provide ``per-query'' quality guarantees on the synthetic data, and the users have no way of knowing how much the aggregated statistics on the synthetic data can be trusted. To address this problem,  
we present a novel framework named {\em \oursys\ (differentially-private per-query decider)} to detect if the query answers on the private and synthetic datasets are within a user-specified threshold of each other while guaranteeing differential privacy.
We give a suite of private algorithms for per-query deciders for count, sum, and median queries, analyze their properties, and evaluate them experimentally.
%define and explore the problem of providing an estimate for the discrepancy between the results of a specific query on the private data and synthetic data. 
%These estimates of discrepancy themselves contain information about the sensitive data and hence can not be released without privacy protection. Hence, we develop differentially private algorithms for estimating the discrepancy in count, median and sum queries and analyze their error.}
\end{abstract}

\maketitle

% %%% do not modify the following VLDB block %%
% %%% VLDB block start %%%
\pagestyle{\vldbpagestyle}
% \begingroup\small\noindent\raggedright\textbf{PVLDB Reference Format:}\\
% \vldbauthors. \vldbtitle. PVLDB, \vldbvolume(\vldbissue): \vldbpages, \vldbyear.\\
% \href{https://doi.org/\vldbdoi}{doi:\vldbdoi}
% \endgroup
% \begingroup
% \renewcommand\thefootnote{}\footnote{\noindent
% This work is licensed under the Creative Commons BY-NC-ND 4.0 International License. Visit \url{https://creativecommons.org/licenses/by-nc-nd/4.0/} to view a copy of this license. For any use beyond those covered by this license, obtain permission by emailing \href{mailto:info@vldb.org}{info@vldb.org}. Copyright is held by the owner/author(s). Publication rights licensed to the VLDB Endowment. \\
% \raggedright Proceedings of the VLDB Endowment, Vol. \vldbvolume, No. \vldbissue\ %
% ISSN 2150-8097. \\
% \href{https://doi.org/\vldbdoi}{doi:\vldbdoi} \\
% }\addtocounter{footnote}{-1}\endgroup
% %%% VLDB block end %%%

% %%% do not modify the following VLDB block %%
% %%% VLDB block start %%%
% \ifdefempty{\vldbavailabilityurl}{}{
% \vspace{.3cm}
% \begingroup\small\noindent\raggedright\textbf{PVLDB Artifact Availability:}\\
% The source code, data, and/or other artifacts have been made available at \url{\vldbavailabilityurl}.
% \endgroup
% }
% %%% VLDB block end %%%

\section{Introduction} \label{sec:intro}
For more than a decade, we have witnessed an abundance of data containing private and sensitive information and a growing interest in using this data for decision making and data analytics. 
Formal policies like the GDPR \cite{GDPR} and CCPA \cite{ccpa} require that the privacy of the individuals whose data is being used be maintained.
Differential privacy (DP) \cite{Dwork06} is the gold standard in offering mathematically rigorous bounds on privacy leakage while offering utility through multiple data releases, even in the presence of side information. Intuitively, DP guarantees that the output has a similar distribution whether an individual's data is used. 
It has been widely adopted by many organizations \cite{10.1145/3219819.3226070,dwork2019differential} and leading companies \cite{erlingsson2014rappor,ding2017collecting,tang2017privacy}.

Private data can be queried directly using designated DP mechanisms. However, the accuracy of the results depends heavily on the privacy budget, especially when multiple queries need to be answered. Furthermore, the results may be inconsistent with each other.  
A prominent alternative is using differentially-private {\em synthetic data generators (SDGs)} to produce a synthetic copy of the private data, which can be used repeatedly to answer multiple queries without spending additional privacy budget. Previous works have employed techniques from game theory \cite{NIPS2012_208e43f0, DBLP:conf/icml/GaboardiAHRW14, pmlr-v119-vietri20b}, probabilistic graphical models \cite{10.1145/3134428, DBLP:conf/icml/McKennaSM19, DBLP:journals/pvldb/HuangMBMHM19, DBLP:journals/corr/abs-2108-04978, DBLP:journals/pvldb/CaiLWX21}, and deep learning \cite{DBLP:journals/corr/abs-2011-05537, DBLP:journals/pvldb/GeMHI21}. 
On the other hand, to generate new instances of datasets resembling properties of the given dataset, SDGs not satisfying DP have also become very popular alternatives in applications where it is appropriate to do so. Examples include SDGs using generative modeling \cite{DBLP:conf/dsaa/PatkiWV16} and deep learning techniques \cite{10.5555/3454287.3454946, 10.14778/3231751.3231757}.
%have proven useful in this regard.
Synthetic data offers advantages such as: ($1$) consistency in answering a large number of statistical queries, ($2$) preservation of desired correlations within data, 
% (3) no need to design different systems for particular applications, 
and ($3$) concise representation of the private data that circumvents expending more privacy budget to answer queries. 

While SDGs (DP or not) embody a promising approach for increasing the usability of private data, there may exist discrepancies between query results over the private and synthetic datasets.
This work in particular focuses on analyzing aggregated statistics, e.g., count, sum, or median, on a subset of data satisfying some conditions, which form some of the most fundamental operations in data analysis.
%, and training AI/ML models.
When synthetic data produced by a SDG is used in data analysis, users may be interested in knowing if their aggregated queries generating such statistics can be reliably answered on said data.
%Such information is often useful to detect if the synthetic data is suitable for specific tasks. 
However, the standard data generation systems %typically 
do not provide ``per-query'' quality guarantees on the synthetic data, and the users have no way of knowing how much the aggregated statistics on the synthetic data can be trusted. We illustrate with an example below.
%Today, the user has no way of knowing the size of this gap for arbitrary SDGs. To the best of our knowledge, most existing works do not address this for individual queries. 

% It comes as no surprise that using a synthetic database without any accuracy guarantees poses ethical and practical issues \cite{?} so, it is critical to be able to judge which results can be confidently published and used. 
% There has been extensive research in the machine learning community demonstrating the impact of modeling limitations, bias, etc. on the quality of outcomes that can impact accountability, safety and industrial liability \cite{10.1145/3236009}. 
% A lack of means to investigate and evaluate SDG$_{DP}$s makes mitigation a challenge.

\begin{example}\label{example:problem}
    Let the private database $D$ be a simplified version of the Adult database \cite{Dua:2019} with attributes: {\tt age, education, capital-gain, marital-status, occupation, relationship, and sex}. Let $D_s$ denote the synthetic copy of $D$ from PrivBayes \cite{10.1145/3134428}, which is a DP SDG. %$_{DP}$.
    % Also consider the following query $q$, where \texttt{<str>} equals some value from the domain of the attribute $marital$-$status$:\\
    % \texttt{SELECT COUNT(*) FROM $D_s$ WHERE $marital$-$status$ LIKE <str>}.
    % \begin{figure}[h]
    %     \centering
    %     \includegraphics[width=\columnwidth]{figs/running_eg.pdf}
    %     \caption{Histogram for attribute $marital$-$status$ in $D, D_s$. 
    %     % \amir{The distances here all do not look very significant. Is there a more extreme example?} \shweta{Added a query on $capital$-$gain$}
    %     }
    %     \label{fig:running_eg}
    % \end{figure}
    % 
    % Figure~\ref{fig:running_eg} shows the output (in navy blue) for all values of \texttt{<str>}. This is equivalent to the result of a \texttt{GROUP-BY} query. In practice, the counts from $D$ are private and we include them here only for reference.
    Also consider the following query $q$, where (\texttt{<a>}, \texttt{<b>}) is one of $(0, 200), (200, 400), (400, 600),$ $(600, 800)$ or $(800, 1000)$:\\   
    {\small
    \noindent
    \texttt{SELECT COUNT(*) FROM $D_s$\\
    WHERE capital-gain $\geq$ <a> AND capital-gain $<$ <b>}.
    }
    \begin{figure}[ht]
        \centering
        \includegraphics[width=0.8\columnwidth]{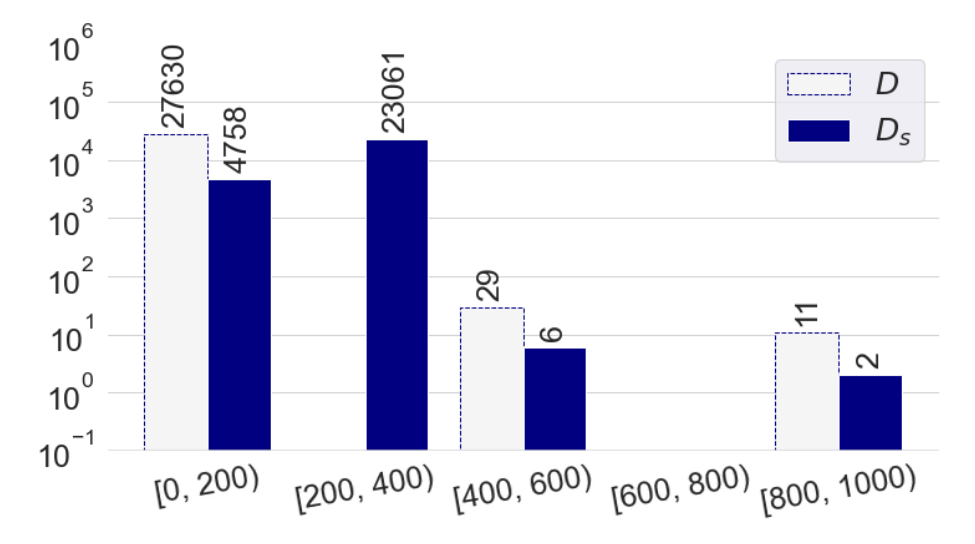}
        \caption{Histogram for attribute $capital$-$gain$ in $D$ (in white, private and not visible to the user) and $D_s$ (in navy blue, visible to the user) for bins given by $(0, 200), (200, 400), (400, 600), (600, 800)$ and $(800, 1000)$.}
        \label{fig:running_eg}
    \end{figure}
    
    Figure~\ref{fig:running_eg} shows the output for the aforementioned values of \texttt{<a>} and \texttt{<b>}. Bars with height $0$ are not shown here.
    Note that the corresponding counts from $D$ are private and we include them here only for reference.

    Suppose the user wants to know if the gap between $q(D)$ and $q(D_s)$ is less than $\tau = 200$. \textbf{How can the user find out if their distance bound for query $q$ is met without access to either the private data or the SDG that was used to produce the synthetic data?} 
    % In Figure~\ref{fig:running_eg}, we can see that when \texttt{<str>} equals `Divorced', the truth is `Yes', whereas when \texttt{<str>} equals `Married-civ-spouse', the truth is `No'. 
    In Figure~\ref{fig:running_eg}, we see that when (\texttt{<a>}, \texttt{<b>}) $= (400, 600)$, the (true) answer is `Yes', whereas when (\texttt{<a>}, \texttt{<b>}) $= (0, 200)$, the (true) answer is `No'. In fact, the distance is $114.36$ times $200$.
\end{example}

In this paper, we aim to build a {\em ``Differentially-Private Per-Query Decider''}, which gives a {\em `Yes' (distance bound is satisfied)} answer if $|q(D) - q(D_s)|$ is smaller than a given distance bound $\tau > 0$, and a {\em `No' (distance bound is unmet)} answer otherwise, while satisfying DP with a given privacy budget $\epsilon > 0$. However, 
there are several challenges that one needs to address.
{\em First}, we assume that we do not have access to the mechanism behind the SDG producing $D_s$. All the per-query decider can see is the synthetic data $D_s$ generated by the SDG. This SDG may be one of the SDGs satisfying DP while outputting $D_s$ using a separate privacy budget \cite{NIPS2012_208e43f0, DBLP:conf/icml/GaboardiAHRW14, pmlr-v119-vietri20b, 10.1145/3134428, DBLP:conf/icml/McKennaSM19, DBLP:journals/pvldb/HuangMBMHM19, DBLP:journals/corr/abs-2108-04978, DBLP:journals/pvldb/CaiLWX21, DBLP:journals/corr/abs-2011-05537, DBLP:journals/pvldb/GeMHI21}, or a SDG that does not use DP \cite{DBLP:conf/dsaa/PatkiWV16, 10.5555/3454287.3454946, 10.14778/3231751.3231757}. The per-query decider should work with $D_s$ generated by any SDG.  {\em Second}, not only can the per-query decider not output the true `Yes' or `No' answers since $q(D)$ is private, but also it cannot give a deterministic answer because DP mechanisms must be randomized algorithms. {\em Third}, the per-query decider should have good accuracy -- answering a random `Yes'/`No' answer is trivially private but is not useful to the user. {\em Finally}, we aim to build a framework of per-query decider that can handle different types of standard aggregates, namely \countQuery, \medianQuery, and \sumQuery, which have different sensitivities on the input (private) data and will need different techniques. %\sr{anything else?}
%and we need to account for any noise injected in the process.} %requires to output a Yes or No answer probabilistically requires a probabilistic analysis 

% \red{Other synthetic data generators may operate differently and there is no guarantee that all of them will provide per-query error bounds. 
% Hence, a better solution is needed so that the user can understand which queries are reliably answered by the synthetic database $D_s$.}\sr{remove}
% \ag{Instead say: since there are no known theoretical guarantees on the quality of the generated synthetic data, the user might wonder if for these queries of interest, are the results obtained by running them on the synthetic data reflect the results that would have been obtained when running the query on the private data, or are they completely different?}

% \shweta{Other use cases like cohort identification in medicine, choosing which ads to show to a user, etc.}

\subsection*{Our Contributions}
In this work, we propose a novel framework called {\em Differentially-Private Per-Query Decider (\oursys) to decide if the distance between $q(D)$ and $q(D_s)$ is less than a user-provided distance bound of $\tau > 0$ for a given query $q$ and privacy budget $\epsilon > 0$.} We investigate the problem for \countQuery, \sumQuery, and \medianQuery\ queries under this framework (with optional predicates selecting a subset of the data), which are three fundamental aggregate operators used in data analysis. We make the following contributions.

%which are important building blocks for more complex data analysis and machine learning tasks, and make the following contributions.

%\begin{enumerate}
    \paragraph{{\bf (1) The \oursys\ framework
    %and {\em effectiveness} measure: 
    (Section~\ref{sec:prob_statement}): }}
    We formally define the differentially-private per-query decider and introduce the notion of {\em effectiveness} of an algorithm to capture the range of input distance thresholds for which a per-query decider algorithm is expected to perform well.

    \paragraph{{\bf (2) \countQuery\ queries (Section~\ref{sec:count}): }}
    For \countQuery\ queries, we present and analyze two approaches, one based on the   Laplace Mechanism (LM) \cite{10.1007/11681878_14} that  uses a DP noisy estimate of $q(D)$ to compare with $q(D_s)$, and the other is direct approach for answering `Yes' or `No' based on the Exponential Mechanism (EM) \cite{10.1109/FOCS.2007.41} using a carefully designed score function. %to reduce error.

    \paragraph{{\bf (3) \sumQuery\ queries (Section~\ref{sec:sum}): }}
    For \sumQuery\ queries, we present and analyze three approaches. Two of them based on the Laplace Mechanism (LM) \cite{10.1007/11681878_14} and the recent {\em Race-to-the-Top} Mechanism (R2T) \cite{10.1145/3514221.3517844} use a DP noisy estimate of $q(D)$ to compare with $q(D_s)$. The third direct approach exploits the {\em Sparse Vector Technique} (SVT) \cite{10.14778/3055330.3055331} originally designed to detect when the first of a sequence of queries % DP noisy query answers
    exceeds a given threshold to implement a per-query decider.  

    \paragraph{{\bf (4) \medianQuery\ queries (Section~\ref{sec:median}): }}
    For \medianQuery\ queries, we present and analyze two approaches: one uses a DP noisy estimate of the median query using EM, and the other is a new histogram-based DP mechanism that directly solves the problem using the LM.
   
    \paragraph{{\bf (5) Experimental evaluation (Section~\ref{sec:experiments}): }}
    We have implemented the \oursys\ framework with all the above algorithms to evaluate our proposed solutions. 
    We analyze the accuracy for \reva{$22$} \countQuery\ queries (with a range of tuple selectivity), \reva{$19$} \sumQuery\ queries (with a range of tuple selectivity and varying downward local sensitivities), and \reva{$19$} \medianQuery\ queries (with different data distribution around the true median). One of the interesting observations is that the error of a DP per-query decider is not always monotonic in the privacy budget $\epsilon$, and we explain why this happens.
%\end{enumerate}

% Our proposed solutions can be categorized into two approaches: ($1$) use a DP algorithm to answer the query and check if the noisy answer is close to the answer on the synthetic data and ($2$) design specialized algorithm for the given problem. 
% We need not be limited by the requirement to first solve for $q(D)$, and we explore this path in our specialized algorithms.

% \shweta{When is one solution preferred over the other?}
% \shweta{Why different approaches for diff aggr? Examples/intuition supporting choice.}

\section{Preliminaries}\label{sec:prelim}
% Problem formulation
% How will you evaluate success?
%     Privacy notion
%     "Utility" meaning
%     Assumptions stated formally
In this section, we review some background concepts and present notations used in the rest of the paper.

\subsection{Data and Queries}\label{subsec:data_and_q}
We are given a private database instance $D$ that comprises a single relation with attributes $A_1, \ldots, A_d$. %and $n$ tuples. 
The domain of attribute $A_i$ is given by $dom(A_i)$, which is categorical or integral. 
% We use $t_i.A_l$ to denote the entry in the $l^{th}$ column of the $i^{th}$ tuple, i.e., a \textit{cell}.

In this paper, we consider three aggregate operators: \countQuery, \sumQuery, and \medianQuery, i.e., the corresponding aggregate queries in SQL  with an optional {\tt WHERE} clause take the following form: 
\begin{itemize}
    \item \texttt{SELECT COUNT(*) FROM $D$ WHERE $\whereClause$},%$<$condition$>$}. 
    \item \texttt{SELECT SUM($A_i$) FROM $D$ WHERE $\whereClause$}, %where $A_i$ is not used in $\whereClause$, 
    and    
    \item \texttt{SELECT MEDIAN($A_i$) FROM $D$ WHERE $\whereClause$}.%, where $A_i$ is not used in $\whereClause$.
\end{itemize}
In these queries, first the predicate $\whereClause$ (if the {\tt WHERE} clause exists)  is applied over all tuples in $D$, and then the aggregate is computed on tuples that satisfy $\whereClause$. 
%When there is no {\tt WHERE} clause, $\whereClause$ is considered to be {\tt true}. 
These queries output single real output value and belong to the class of {\em scalar queries}. 

%Below, we use \texttt{$<$condition$>$} to mean a Boolean conjunctive selection predicate, denoted by \whereClause\ in Sections~\ref{sec:sum} and \ref{sec:median}. \whereClause\ is \texttt{true} when the \texttt{WHERE} clause is empty.
%The queries take the following forms respectively:

\begin{example}\label{example:problem_with_queries}
    Recall the private database $D$ and its synthetic copy $D_s$ from Example~\ref{example:problem}. 
    Consider the following queries:
    \begin{itemize}
        \item \qcount: \texttt{SELECT COUNT(*) FROM $D_s$ WHERE $age > 30$ AND $education$ LIKE 'Masters'}.
        
        \item \qsum: \texttt{SELECT SUM($capital$-$gain$) FROM $D_s$ WHERE $education$ LIKE '$12$th'}.
        
        \item \qmed: \texttt{SELECT MEDIAN($capital$-$gain$) FROM $D_s$}.
    \end{itemize}
    Query \qcount\ asks for the number of people with age above $30$ and Master's degree, \qsum\ asks for the total capital-gain of people with $12$-th grade education, and \qmed\ asks for the median of capital-gain over all people.
\end{example}

\subsection{Differential Privacy}\label{subsec:DP}
We use {\em Differential Privacy (DP)} \cite{Dwork06} as the measure of privacy. 
We say that two databases $D$ and $D'$ are {\em neighbors} if they differ by a single tuple. This is denoted by $D\approx D'$.

% unbounded DP, $\ominus$ - symm diff
\begin{definition}[Differential Privacy \cite{DP_book}]
    A randomized mechanism $\mathcal{M}$ is said to satisfy $\epsilon$-DP if $~\forall S\subseteq Range(\mathcal{M})$ and ~$\forall D, D'$ pair of neighboring databases, i.e., $D\approx D'$,
    \begin{center}
        $Pr[\mathcal{M}(D)\in S] \leq e^\epsilon Pr[\mathcal{M}(D')\in S]$
    \end{center}
\end{definition}

Smaller $\epsilon$ gives stronger privacy guarantee. 

% Advanced composition \cite{5670947} and DP variants \cite{10.1007/978-3-662-53641-4_24} can be used to derive tighter bounds on the privacy budget.

\begin{definition} [Global Sensitivity]\label{def:GS}
    For a scalar query $q$, its global sensitivity is given by $\Delta q = \max\limits_{D\approx D'} |q(D) - q(D')|$.
\end{definition}

\begin{definition} [Downward Local Sensitivity]\label{def:DLS}
    For a scalar query $q$, its downward local sensitivity on database $D$ is given by
    \begin{align*}
        \DLS = \max\limits_{D'\approx D, D'\subseteq D}|q(D) - q(D')|
    \end{align*}
\end{definition}

\begin{example}
    Consider the private database $D$ and sum query \qsum\ from Example~\ref{example:problem_with_queries}.% We know that 
    Suppose $dom(capital$-$gain)$ is $\{0, 1, \ldots,$ $99999\}$, so $\Delta \qsum = 99999$ since the maximum change in the sum over all pairs of neighboring databases is the maximum value in the domain. On the other hand, given a database $D$,  \DLS\ equals the largest value in $capital$-$gain$ from tuples in $D$ with $education$ equal to $12$-th. 
\end{example}

Properties like composition \cite{Dwork06} and post-processing \cite{dwork2006our} give a modular way to build complex DP mechanisms:

\begin{proposition}\label{prop:DP-comp-post} 
    \cite{Dwork06, dwork2006our} give the following:
    \begin{enumerate}
        \item {\bf (Sequential composition)} %states that 
        If $\mathcal{M}_i$ satisfies $\epsilon_i$-DP, then the sequential application of $\mathcal{M}_1$, $\mathcal{M}_2, \cdots$,  satisfies $\sum_{i} \epsilon_i$-DP.
        
        \item {\bf (Parallel composition)} %states that 
        If each $\mathcal{M}_i$ accesses disjoint sets of tuples, then they together satisfy $\max_i \epsilon_i$-DP.
        
        \item {\bf (Post-processing)} %states that a 
        Any function applied to the output of an $\epsilon$-DP mechanism $\mathcal{M}$ also satisfies $\epsilon$-DP.
    \end{enumerate}
\end{proposition}

\paratitle{Laplace mechanism (LM)}
The Laplace mechanism \cite{10.1007/11681878_14} is a common building block in DP mechanisms and is used to get a noisy estimate for scalar queries with numeric answers. The noise injected is calibrated to the global sensitivity of the query.

\begin{definition}[Laplace Mechanism]\label{def:LM}
    % Given the universe $\mathcal{X}$ of tuples, database $x\in\mathbb{N}^{|\mathcal{X}|}$ and $f: \mathbb{N}^{|\mathcal{X}|}\rightarrow \mathbb{R}^k$ with global sensitivity $\Delta f$, the Laplace mechanism computes $\mathcal{M}_L(x, f(\cdot), \epsilon) = f(x) + (Y_1, \ldots, Y_k)$, where $Y_i$ are sampled i.i.d. from $Lap(\Delta f/\epsilon)$.
    Given a database $D$, 
    scalar query $q$ (with output in $\mathbb{R}$), and 
    privacy budget $\epsilon$, 
    the Laplace mechanism $\mathcal{M}_L$ returns $q(D) + \nu_q$, where $\nu_q\sim Lap(\Delta q/\epsilon)$.
\end{definition}

Its accuracy is given by the following theorem %from 
\cite{DP_book}. 

\begin{theorem}[Accuracy of Laplace Mechanism \cite{DP_book}]\label{thm:Lap_accuracy}
    % For database $x$, $f: \mathbb{N}^{|\mathcal{X}|}\rightarrow \mathbb{R}^k$ with global sensitivity $\Delta f$, $y = \mathcal{M}_L(x, f(\cdot),$ $\epsilon)$ and $\forall\delta\in(0, 1]$
    % \begin{align*}
    %     Pr\left[ || f(x) - y ||_\infty\geq \frac{\Delta f}{\epsilon} \cdot \ln\left({\frac{k}{\delta}}\right) \right]\leq \delta
    % \end{align*}
    Given a database $D$ and 
    scalar query $q$ (with output in $\mathbb{R})$. Let $y$ be the output from running the Laplace Mechanism $\mathcal{M}_L$ on $D$ and $q$ with privacy budget $\epsilon$. Then, $\forall \delta\in (0, 1]$, 
    \begin{align*}
        Pr\left[ |q(D) - y|\geq \frac{\Delta q}{\epsilon} \ln{\frac{1}{\delta}} \right] \leq \delta
    \end{align*}
\end{theorem}
% The proof uses Fact $3.7$ \cite{DP_book}, which states that for $Y\sim Lap(b)$, $Pr[|Y| \geq t\cdot b] = e^{-t}$. Also, $Pr[Y \geq t\cdot b] = e^{-t}/2$ (by symmetry).

\begin{example}
    Consider the private database $D$ and query \qcount\ from Example~\ref{example:problem_with_queries}. $\Delta \qcount$ equals $1$ because the count can change by at most one on neighboring databases.
    Say the output of \qcount\ on $D$ is $n_1$, a private quantity, and we want a DP estimate, $\tilde{n}_1$, for it.
    $\mathcal{M}_L$ returns $n_1$ plus noise sampled from $Lap(1/\epsilon)$.
    Also, $Pr\left[|n_1 - \tilde{n}_1|\geq \frac{1}{\epsilon} \ln{\frac{1}{\delta}} \right]\leq \delta$.
\end{example}

% We use the following property of the sum of independent Laplace variables %from 
% \cite{DP_book} in Sections~\ref{sec:sum} and \ref{sec:median}.

% \begin{lemma}\label{lemma:sum-LM-rvs}
%     Given $Y_1, Y_2, \ldots, Y_k$ are independent variables distributed as $Lap(b_i)$. Let $\nu\geq \sqrt{\sum_i (b_i)^2}$ and $0 < \lambda < \frac{2\sqrt{2}\nu^2}{\max_i b_i}$, then
%     \begin{align*}
%         Pr \left[ \sum_i Y_i > \lambda \right] \leq \exp{\left( -\frac{\lambda^2}{8 \nu^2} \right)}
%     \end{align*}
% \end{lemma}

\paratitle{Exponential mechanism (EM)}
For categorical outputs, the Exponential mechanism \cite{10.1109/FOCS.2007.41} is used with an appropriate \textit{score function} that gives the utility of each element in the output space with respect to the given private database $D$.
The likelihood of an element being returned as the output depends on its score.

\begin{definition}[Exponential Mechanism]\label{def:EM}
    % Given the universe $\mathcal{X}$ of tuples, database $x\in\mathbb{N}^{|\mathcal{X}|}$, range of outputs $\mathcal{R}$ and score function $u:\mathbb{N}^{|\mathcal{X}|} \times \mathcal{R}\rightarrow \mathbb{R}$ with global sensitivity $\Delta u$, the Exponential mechanism $\mathcal{M}_E(x, u, \mathcal{R})$ returns $r\in\mathcal{R}$ with probability $\propto e^{\frac{\epsilon u(x, r)}{2\Delta u}}$.
    Given a database $D$, 
    range of outputs $\mathcal{R}$, 
    real-valued score function $u(D, e)$ that gives the utility of $e\in \mathcal{R}$ with respect to $D$, and
    privacy budget $\epsilon$, 
    the Exponential mechanism $\mathcal{M}_E$ returns $e\in \mathcal{R}$ with probability $c\cdot e^{\frac{\epsilon u(D, e)}{2\Delta u}}$, where $c$ is a positive constant and $\Delta u$ is the global sensitivity of $u$.
\end{definition}

% \begin{theorem}[Accuracy of Exponential Mechanism \cite{DP_book}]\label{thm:Exp_accuracy}
%     For database $x$, $\mathcal{R}_{OPT} := \{r\in \mathcal{R}: u(x, r) = OPT_u(x)\}$, where $OPT_u(x) = \max_{r\in\mathcal{R}} u(x, r)$, then
%     \begin{align*}
%         Pr\left[ u(\mathcal{M}_E(x, u, \mathcal{R})) \leq OPT_u(x) - \frac{2\Delta u}{\epsilon} (\ln(|\mathcal{R}|) + t) \right] \leq e^{-t}
%     \end{align*}
% \end{theorem}

\begin{example}
    Recall the private database $D$ and query \qmed\ from Example~\ref{example:problem_with_queries}.  
    We find a DP estimate for the median \cite{DBLP:journals/corr/abs-1103-5170} by applying $\mathcal{M}_E$ with $u(D, e) = -|rank(e) - n/2|$, for $e\in dom(capital$-$gain)$. $\Delta u$ is $1$, and therefore, $\mathcal{M}_E$ returns $e$ with probability $\propto e^{\epsilon u(D, e)/2}$.
\end{example}

\subsection{Synthetic Data Generators}
We treat the {\em synthetic data generator (SDG)}
%referred to as {\em SDG} \sdg, 
used to produce $D_s$ as a black box. 
%and do not require it to satisfy DP.
Our framework \oursys\ can be used in conjunction with any synthetic data generator: a standard SDG (not satisfying DP) like \cite{DBLP:conf/dsaa/PatkiWV16, 10.5555/3454287.3454946, 10.14778/3231751.3231757} that typically takes $D$ and some optional constraints as inputs, or an SDG satisfying DP (SDG$_{DP}$)
%(DP-SDG)  A DP-SDG %\sdg\ 
%satisfying DP 
that takes $D$ and a privacy budget as input. An SDG$_{DP}$ may or may not take a set of queries as input. For example, PrivBayes \cite{10.1145/3134428} does not take queries as input but works such as \cite{NIPS2012_208e43f0, DBLP:conf/icml/GaboardiAHRW14, pmlr-v119-vietri20b, DBLP:journals/pvldb/HuangMBMHM19, DBLP:journals/corr/abs-2108-04978, DBLP:conf/icml/McKennaSM19, 10.14778/3231751.3231769} do.  
For an SDG$_{DP}$, we assume that it has a privacy budget separate from the privacy budget $\epsilon$ for the per-query decider \oursys. \oursys\ takes as input the synthetic data $D_s$ generated from any SDG, the private database $D$ and a privacy budget $\epsilon$, does not need to run the SDG again, and does not assume anything about how the SDG works. 
% \shweta{Questions from PPAI poster session: 
% Some assumed that part of $\epsilon$ gets used by the DP-SDG \dpsdg. 
% One person asked why not simply use the DP estimate and decide? The paper addresses this for particular approaches only as of now. What is a good place to describe the two directions - plug in DP estimate vs design specialized algo - and say why the 1st direction isn't always the best?}

\section{The \oursys\ Framework}\label{sec:prob_statement}
% \begin{definition}[Problem 1]
%     Given a private database $D$,
%     randomized DP-SDG \dpsdg\ with its learned DP model, 
%     synthetic database $D_s$ (with same schema and size $n$ as the private database $D$), 
%     query $q$ 
%     and confidence level of $1 - \alpha$, 
%     return a $(1 - \alpha)$-CI for $q(D)$ with minimum width. 
% \end{definition}
In this section, we present the \oursys\ (Differentially-Private Per-Query Decider) framework that intends to solve the following problem; 
%decides whether the gap between the  output of a given query $q$ on the private database $D$ and its synthetic copy $D_s$ is bounded by a user-provided distance bound of $\tau$, while satisfying $\epsilon$-DP (
the workflow of \oursys\ is given in Figure~\ref{fig:flowchart}.

\begin{definition}[Differentially-Private Per-Query Decider]\label{def:prob_statement}
    Given a private database $D$, 
    synthetic database $D_s$ for $D$ from a black-box SDG, % \sdg, 
    query $q$ (\countQuery, \sumQuery\ or \medianQuery), 
    distance bound $\tau > 0 \in \mathbb{R}$, %\mathbb{R}^+$ 
    and privacy budget $\epsilon > 0$, 
    return whether $|q(D) - q(D_s)| < \tau$ while satisfying $\epsilon$-DP. 
    We call such a mechanism a {\bf differentially-private per-query decider}, or simply a per-query decider, and denote it by $\alg(D, D_s, q, \tau, \epsilon)$, or $\alg(D)$ when $D_s, q, \tau$, and $\epsilon$ are clear from context.
    \par
    To simplify notation, we will write $o$ to denote the outcome of $\alg(D)$ when $\alg, D, D_s, q, \tau, \epsilon$ are clear from context. Here, 
    \begin{eqnarray*}
        o = 1 & \equiv & \text{``Distance bound satisfied''}\\ 
        o = 0 & \equiv & \text{``Distance bound unmet''}    
    \end{eqnarray*}
\end{definition}

\begin{figure}[ht]
    \centering
    \includegraphics[width=0.85\columnwidth]{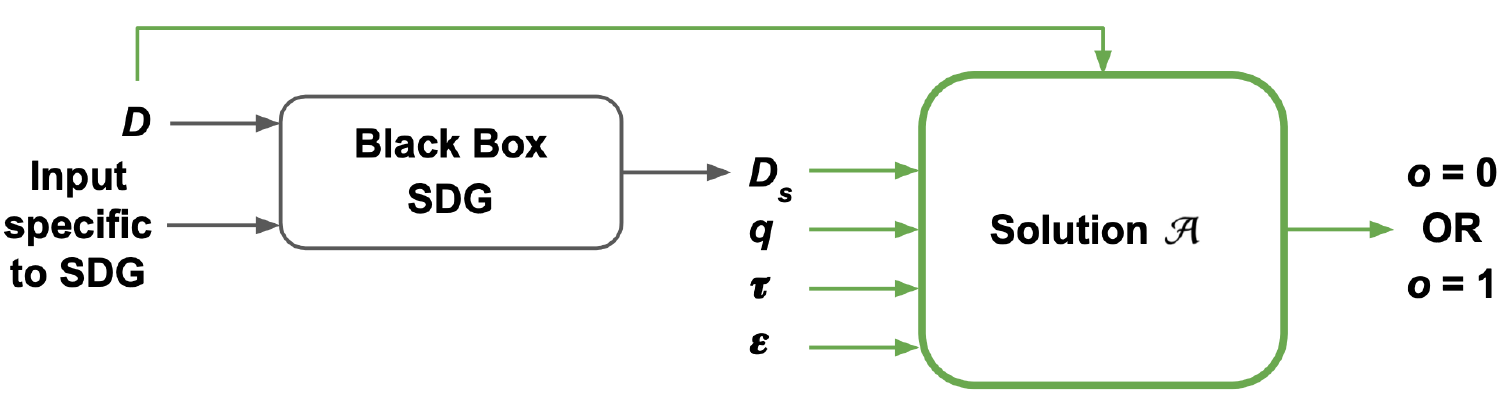}
    \caption{Workflow of \oursys}
    \label{fig:flowchart}
\end{figure}

\revb{
In this paper, we investigate the following approaches for $\alg(D)$: ($1$) spend $\epsilon$ to obtain a noisy DP estimate for $q(D)$, and compare it with $q(D_s)$ to check if the distance bound of $\tau$ is met (by instantiating Algorithm~\ref{algo:plug_in} with a suitable DP mechanism $\dpmech(\cdot)$), and ($2$) design specialized algorithms to solve the problem without first estimating $q(D)$. A summary of our solutions %in this paper 
is given in Table~\ref{tbl:summary}.
}

For convenience, we denote the desired interval for $q(D)$ as $\intvl :=  (l, r)$ (known to the user), and the absolute difference in query values of $q(D)$ and $q(D_s)$ as  $\distBd$ (hidden from the user), i.e., 
\begin{align}
   \intvl &:=  (l, r) = (q(D_s) - \tau, q(D_s) + \tau) \label{eq:interval}\\
   \distBd &:=  |q(D) - q(D_s)| \label{eq:query-diff}
\end{align}
\revb{
Since $q(D_s)$ is known to the user, but $q(D)$ is not, we envision the user choosing $\tau$ as a percentage value of $q(D_s)$ in practical applications. Here, $\intvl = (l, r) = (q(D_s)\cdot (1-\tau), q(D_s)\cdot (1+\tau))$.
The $\tau$ value depends on how much error in $q(D_s)$ can be tolerated by the application using $q(D_s)$. For example, $\tau$ larger than $10\%$ of $q(D_s)$ may introduce too much error in the downstream analysis. To present our techniques in Sections~\ref{sec:count}, \ref{sec:sum}, and \ref{sec:median} for a given query, we use $\tau$ as a constant, whereas in the experiments in Section~\ref{sec:experiments} when we vary the queries, we use and vary $\tau$ as a percentage of $q(D_s)$ since the same value of $\tau$ may not be meaningful to all queries. 
}

\alg\ satisfies $\epsilon$-DP and like any randomized mechanism, incurs error in deciding if the distance bound is met (i.e. if $\distBd < \tau$). 
We quantify error as the expectation of the event that \alg\ returns the wrong outcome on $D$ defined as follows:
%and formally define it in Definition~\ref{def:error}.

\begin{definition}[Error]\label{def:error}
    Let \alg\ be a per-query decider for the given private database $D$, 
    synthetic database $D_s$, 
    query $q$, 
    distance bound $\tau$, and 
    privacy budget $\epsilon$. %We define the error as follows:
    Then the error of \alg\ is given by:
    \begin{align*}
        err_\alg(q, D, D_s, \tau, \epsilon) &= Pr[o = 1 \land \distBd \geq \tau] + Pr[o = 0 \land \distBd < \tau]
    \end{align*}
    where the probability is over the randomness in \alg.  
    % \sr{added \alg\ in the subscript of err - please fix elsewhere, especially in sec 6 use this error by adding false positive and false negative and refer to this definition}
\end{definition}

In Definition~\ref{def:error}, the first term corresponds to the {\em false positive rate} and the second term corresponds to the {\em false negative rate}.
Intuitively, we expect larger error when \distBd\ and $\tau$ are closer because the chance of the random noise injected by \alg\ causing the wrong outcome becomes higher.

\begin{example}
    Consider $D, D_s, q$ from Example~\ref{example:problem} where $(0, 200)$ is used as the range. Given that $\tau = 200$. Here, $\distBd$ (hidden from the user) equals $22,872$ and is greater than $\tau$, so the correct outcome of \alg\ is $o = 0$. Hence, a solution \alg\ makes a mistake if $o = 1$ and the error equals the probability that \alg\ returns $o = 1$.
\end{example}

% \shweta{In the end, user can guess a reasonable proxy for \DLS. $\tau$ -> which algo?}
Next, we introduce a notion we call the \textit{effectiveness} of a per-query decider \alg. The idea is to derive a lower bound for $\tau$ such that \alg\ returns the correct outcome with probability at least $1 - \delta$ (for $0 < \delta < 1$) in two cases: ($1$) $q(D) = q(D_s)$, and ($2$) $q(D) \leq q(D_s) - 2\tau$ or $q(D) \geq q(D_s) + 2\tau$. 
The first condition ensures that \alg\ has high accuracy when the two query outputs match, and the second condition ensures that \alg\ has high accuracy also for very different query outputs. In the absence of the first case, an \alg\ that always returns $o=0$ meets the condition, and in the absence of the second case, an \alg\ that always returns $o=1$ meets the condition, but neither is a useful solution. 
We denote the lower bound by $\tau_{min}^{\alg, \delta}$ and call it the \textit{effectiveness threshold} of \alg\ at $\delta$. 

\begin{definition}[Effectiveness]\label{def:effectiveness_metric}
    Let \alg\ be a per-query decider for a given private database $D$, 
    a synthetic database $D_s$, 
    a query $q$, 
    a distance bound $\tau$, and 
    a privacy budget $\epsilon$.
    %We define $\tau_{min}^{\alg, \delta}$ to be the smallest value that $\tau$ can take, for $0 < \delta < 1$, s.t.
    \alg\ is called {\bf effective} at error probability $0 < \delta < 1$ if the following two conditions hold:
    \begin{enumerate}
        \item if $q(D) = q(D_s), Pr[o = 1]\geq 1 - \delta$, and
        \item if $q(D)\not\in (q(D_s) - 2\tau, q(D_s) + 2\tau) , Pr[o = 0]\geq 1 - \delta$.
    \end{enumerate}

    The smallest value of $\tau$ that achieves the above is called the effectiveness threshold of \alg\ at $\delta$ and is denoted by $\tau_{min}^{\alg, \delta}$.
    %In the above two scenarios, the error is less than $\delta$.
\end{definition}

We give an upper bound on the effectiveness thresholds of each solution for \countQuery\ and \sumQuery\ queries (Section~\ref{sec:count} and \ref{sec:sum}). 
We use effectiveness as proxy for error for queries with high sensitivities, like \sumQuery\ queries.
For \medianQuery\ queries, $q(D)$ shows up in the rank space in the analysis. This is an interesting direction for future work.
Table~\ref{tab:notation} summarizes the notations used throughout the paper.

\begin{table}[t] %[!ht]
    \centering
    \footnotesize
    \begin{tabular}{| c | c |}
        \hline
        \cellcolor[HTML]{C0C0C0} Notation & \cellcolor[HTML]{C0C0C0} Description\\
        \hline $D, D_s$& Private database and its synthetic copy\\
        \hline $q$& \countQuery, \sumQuery, or \medianQuery\ query\\
        \hline $\Delta q$& Global sensitivity of query $q$\\
        \hline \DLS& Downward local sensitivity of $q$ on $D$ (Sec.~\ref{sec:sum})\\
        \hline $\tau$& Given upper bound on $|q(D) - q(D_s)|$\\
        \hline $\epsilon$& Privacy budget\\
        % \hline \distBd& $|q(D) - q(D_s)|$\\
        \hline $\intvl = (l, r)$& $(q(D_s) - \tau, q(D_s) + \tau)$\\
        \hline $\alg(D, D_s, q, \tau, \epsilon)$& $\epsilon$-DP per-query decider\\
        \hline $o\in\{0, 1\}$& \alg's outcome on the given problem\\
        \hline $\tau_{min}^{\alg, \delta}$& Effectiveness threshold of \alg\ for $0 < \delta \ll 1$\\
        \hline \GSbound& Upper limit on $\Delta q$ for \sumQuery\ query $q$ (Sec.~\ref{sec:sum})\\
        \hline
    \end{tabular}
    \caption{Table of Notations}
    \label{tab:notation}
\end{table}

% Ashwin's comments (Spring '22):

% Given a private database $D$ and a randomized DP mechanism $M$ that can take as input $D$ and output a synthetic database $D'$ in the same schema as $D$. Given a query $q$ for which we want to determine its error (i.e. an estimate of $|q(D) - q(D)'|$ maybe in expectation or high probability).

% Caveat: We know not all properties in $D$ will be preserved by $M$ (by fundamental law of information reconstruction). Which means, not all queries will have low error. 

% Problem: Given a query q, synthetic database $D'$ output one of the following: 
% \begin{itemize}
% \item Either say $Error(q) = \bot$. There is some heuristic that depends on $q, M$ and $D'$ that says we can't estimate the error in q. 
% \item Provide confidence interval on Error(q)
% \end{itemize}
\section{Solutions for \countQuery\ query}\label{sec:count}
We propose two approaches for \countQuery\ query $q$: 
($1$) \lapcount\ (that instantiates Algorithm~\ref{algo:plug_in} with the Laplace Mechanism (LM)) in Section~\ref{subsec:lapcount}, and 
($2$) \emcount\ (that directly solves the problem using the Exponential Mechanism (EM)) in Section~\ref{subsec:emcount}, and analyze their errors. 
We also derive upper bounds for their respective effectiveness thresholds.
\revm{We give an error comparison in Section~\ref{subsec:emcount}}.

\IncMargin{1em}
\begin{algorithm}[!ht]
    \caption{Basic approach using DP estimate of $q(D)$}
    \label{algo:plug_in}
    
    \SetKwInOut{Input}{Input}\SetKwInOut{Output}{Output}
    \LinesNumbered
    \Input{$q$ - count/sum/median query, $D$ - private database, $D_s$ - synthetic database, $\tau$ - distance bound, $\epsilon$ - privacy budget, \dpmech\ - any $\epsilon$-DP mechanism to get a noisy estimate for $q(D)$, \param\ - any additional parameter(s) that \dpmech\ takes.\\
    {\em /* If \dpmech\ = LM (Defn.~\ref{def:LM}), then $\param = \emptyset$\\ ~~~~~~~If \dpmech\ = EM (Defn.~\ref{def:EM}), then $\param = \{\mathcal{R}, u\}$\\ ~~~~~~~If \dpmech\ = R2T (Sec.~\ref{subsec:rtwotsum}), then $\param = \{\GSbound, \beta\}$  */}}
    \Output{$o = 1$ if the desired distance bound from $q(D_s)$ is satisfied for $q(D)$, else $o = 0$.}
    
    \BlankLine
    
    \SetKwFunction{FMain}{\BasicDecider}
    \SetKwProg{Fn}{Function}{:}{}
    \Fn{\FMain{$q, D, D_s, \tau, \epsilon, \dpmech, \param$}}
    {
        \If{$-\tau < \dpmech(D, q, \epsilon, \param) - q(D_s) < \tau$}
        {   \label{l:basic_if}
            \Return $o = 1$ (``Distance bound satisfied'')\;
        }
        % \Else
        % {   \label{l:basic_else}
        \Return $o = 0$ (``Distance bound unmet'')\; \label{l:basic_else}
        % }
    }
\end{algorithm}
\DecMargin{1em}

\subsection{Laplace Mechanism-Based Approach}\label{subsec:lapcount}
In our first algorithm, \lapcount, we use the LM (Definition~\ref{def:LM}) to obtain a DP estimate for $q(D)$ and check if the noisy answer is less than $\tau$ away from $q(D_s)$. 
This is achieved by running \BasicDecider\ (Algorithm~\ref{algo:plug_in}) with the LM that adds noise from $Lap(\frac{1}{\epsilon})$ as \dpmech\ (since $\Delta q$ = 1 for \countQuery\ queries). 
%, with no additional parameters.
% \begin{enumerate}
%     \item Compute $q(D) + \nu_q$, where $\nu_q\sim Lap(1/\epsilon)$.
%     \item If $-\tau < q(D) + \nu_q - q(D_s) < \tau$, then return ``Distance bound satisfied''. Else, return ``Distance bound unmet''.
% \end{enumerate}
Since we post-process a DP estimate (Proposition~\ref{prop:DP-comp-post}), the following holds:

\begin{observation}
    \lapcount\ satisfies $\epsilon$-DP. 
\end{observation}

We denote the noise injected by the LM (Definition~\ref{def:LM}) as $\nu_q$ and analyze \lapcount's error next. We will frequently use the following properties of the Laplace distribution (with mean $0$) \cite{DP_book}. For a Laplace random variable $\nu_q \sim Lap(\frac{1}{\epsilon})$ and for $t \geq 0$, 
% \vspace{-1mm}
\begin{eqnarray}
    Pr \left[ \nu_q\geq t\cdot \frac{1}{\epsilon} \right] & = & \frac{1}{2}e^{-t} \label{eq:lap-geq}\\
    Pr \left[ \nu_q\leq - t\cdot \frac{1}{\epsilon} \right] & = & \frac{1}{2}e^{-t} \label{eq:lap-leq}\\
    Pr \left[|\nu_q| \geq t\cdot \frac{1}{\epsilon} \right] & = & e^{-t}\label{eq:lap-mod}
\end{eqnarray}

We next employ these equations to bound the error of \lapcount.
We give the full proof of Proposition~\ref{prop:lapcount_err} in Appendix~\ref{subsec:appendix_lapcount_err}.

\begin{proposition}\label{prop:lapcount_err}
    Given a private database $D$, 
    synthetic database $D_s$, 
    \countQuery\ query $q$, 
    distance bound $\tau$, and 
    privacy budget $\epsilon$. 
    Interval $\intvl = (l, r) = (q(D_s) - \tau, q(D_s) + \tau)$ (\ref{eq:interval}). \lapcount\ satisfies the following:
    % Recall from (\ref{eq:interval}) that $\intvl = (l, r) = (q(D_s) - \tau, q(D_s) + \tau)$. For \lapcount,
    \begin{enumerate}
        \item If $q(D) \leq l$ but $o = 1$, then $err(\cdot) \leq \frac{1}{2} e^{-(l - q(D))\epsilon} - \frac{1}{2} e^{-(r - q(D))\epsilon}$.
        
        \item If $q(D) \geq r$ but $o = 1$, then $err(\cdot) \leq \frac{1}{2} e^{-(q(D) - r)\epsilon} - \frac{1}{2} e^{-(q(D) - l)\epsilon}$.
        
        \item If $l < q(D) < r$ but $o = 0$, then $err(\cdot) = \frac{1}{2} e^{-(q(D) - l)\epsilon} + \frac{1}{2} e^{-(r - q(D))\epsilon}$. 
    \end{enumerate}
\end{proposition}

We now give an upper bound for the effectiveness threshold.

\begin{proposition}\label{prop:tau_min_lapcount}
    Given a private database $D$, 
    synthetic database $D_s$, 
    \countQuery\ query $q$, 
    privacy budget $\epsilon$, and 
    error probability $0 < \delta < 1$, 
    the effectiveness threshold of the per-query decider \lapcount\ at $\delta$ (Definition~\ref{def:effectiveness_metric}) has the following upper bound: 
    $\tau_{min}^{\lapcount, \delta} \leq \frac{1}{\epsilon} \ln{\frac{1}{2\delta}}$.
\end{proposition}

We give the full proof of Proposition~\ref{prop:tau_min_lapcount} in Appendix~\ref{subsec:appendix_tau_min_lapcount}.

\subsection{A Direct Solution Using the EM}\label{subsec:emcount}
In our second algorithm, \emcount, instead of plugging in a DP estimate for $q(D)$ to reach a decision about the distance bound, directly returns whether $o = 1$ (``Distance bound satisfied'') or not.

\subsubsection{A straightforward EM-based per-query decider.~} We begin by describing a straightforward approach we call \emcountNaive\ that directly instantiates the EM (Definition~\ref{def:EM}) with the output range $\mathcal{R} = \{0, 1\}$ and the score function $u(D, D_s, q, \tau, o)$ given by:
% \begin{align*}
%     u(D, D_s, q, \tau, o) = 
%     \left\{
%         \begin{array}{lr}
%             1, & q(D) \in \intvl, o = 1\\
%             0, & q(D) \in \intvl, o = 0\\
%             1, & q(D) \not\in \intvl, o = 0\\
%             0, & q(D) \not\in \intvl, o = 1
%         \end{array}
%     \right.
% \end{align*}
\begin{align*}
    \footnotesize
    \text{If $q(D)\in \intvl$}
    \left\{
        \begin{array}{lr}
            u(D, D_s, q, \tau, o = 0) = 0\\
            u(D, D_s, q, \tau, o = 1) = 1
        \end{array}
    \right.
\end{align*}
\begin{align*}
    \footnotesize
    \text{If $q(D)\not\in \intvl$}
    \left\{
        \begin{array}{lr}
            u(D, D_s, q, \tau, o = 0) = 1\\
            u(D, D_s, q, \tau, o = 1) = 0
        \end{array}
    \right.
\end{align*}
where $\intvl = (l, r) = (q(D_s) - \tau, q(D_s) + \tau)$ (from (\ref{eq:interval})).
Intuitively, $u(\cdot)$ gives a non-zero score only to the correct outcome. 
% For example, when $q(D)\in \intvl$, $u(D, D_s, q, \tau, o = 1) = 1$ and $u(D, D_s, q, \tau, o = 0) = 0$.

We give the full proof of Proposition~\ref{prop:bin_u_sens} in Appendix~\ref{subsec:appendix_bin_u_sens}.

\begin{proposition}\label{prop:bin_u_sens}
    The global sensitivity of $u$ is $\Delta u = 1$.
\end{proposition}

The problem with \emcountNaive\ is that it suffers from high error as sensitivity $1$ is too high. For instance, when $q(D)\in\intvl$ but $o=0$, $Pr[o = 0] = c\cdot e^{\epsilon \times 0/2}$ and $Pr[o = 1] = c\cdot e^{\epsilon \times 1/2}$ (Definition~\ref{def:EM}), where $c$ is a positive constant. Since there are only two outcomes %$o = 0, 1$
\begin{align}
    Pr[o = 0] & = \frac{1}{1 + e^{\epsilon/2}}\label{eq:LM-naive-0}\\
    Pr[o = 1] & = \frac{e^{\epsilon/2}}{1 + e^{\epsilon/2}}\label{eq:LM-naive-1}
\end{align}
When $q(D) \in \intvl$, $o = 0$ is the wrong outcome and the error is given by (\ref{eq:LM-naive-0}).
Similarly for $q(D)\not\in \intvl$, error equals $1/(1 + e^{\epsilon/2})$.

In contrast, \lapcount\ shows that lower error may be achieved. 
% For example, when $q(D)\in\intvl$ but $o=0$ (case (3) in Proposition~\ref{prop:lapcount_err}), \lapcount\ incurs error between $e^{-\tau\epsilon}$ (when $q(D) = q(D_s)$) and $\frac{1}{2} + \frac{1}{2}e^{-2\tau\epsilon}$ (when $q(D)\rightarrow l$).
For example, when $q(D) = q(D_s)$ but $o=0$ (case (3) in Proposition~\ref{prop:lapcount_err}), \lapcount's error, $e^{-\tau\epsilon}$, is smaller than \emcountNaive's error for values of $\epsilon$ and $\tau$ such as $0.1$ and $8$, respectively. 
We improve upon \emcountNaive\ by engineering a new score function with a smaller global sensitivity.

\subsubsection{An improved EM-based per-query decider.~} 
 To use the EM more reliably, we propose a new score function $u'(D, D_s, q, \tau, o)$: 
% \begin{align*}
%     \footnotesize
%     u'(D, D_s, q, \tau, o) = 
%     \left\{
%         \begin{array}{lr}
%             1, & q(D)\not\in (l - \tau, r + \tau), o = 0\\
%             0, & q(D)\not\in (l - \tau, r + \tau), o = 1\\
%             1 - \frac{q(D) - (l-\tau)}{2\tau}, & q(D)\in (l - \tau, q(D_s)], o = 0\\
%             \frac{q(D) - (l - \tau)}{2\tau}, & q(D)\in (l - \tau, q(D_s)], o = 1\\
%             \frac{q(D) - q(D_s)}{2\tau}, & q(D)\in (q(D_s), r + \tau), o = 0\\
%             1 - \frac{q(D) - q(D_s)}{2\tau}, & q(D)\in (q(D_s), r + \tau), o = 1
%         \end{array}
%     \right.
% \end{align*}
% \sr{this can be done later, but better to write the score functions as if-else form.. e.g., case (1)  $q(D)\not\in (l - \tau, r + \tau)$, then if $o = 0$ then $u'= 1$, if $o = 1$ then $u' = 0$. for both naive and improved -- this will explicitly say that which two outcomes form the exhaustive cases -- the above form gives the impression that there are 6 possible outcomes.}
\begin{align}
    \footnotesize
    \text{If $q(D)\not\in (l - \tau, r + \tau)$}
    \left\{
        \begin{array}{lr}
            u'(D, D_s, q, \tau, o = 0) = 1\\
            u'(D, D_s, q, \tau, o = 1) = 0
        \end{array}
    \right. \label{eq:uprime_1}
\end{align}
\begin{align}
    \footnotesize
    \text{If $q(D)\in (l - \tau, q(D_s)]$}
    \left\{
        \begin{array}{lr}
            u'(D, D_s, q, \tau, o = 0) = 1 - \frac{q(D) - (l-\tau)}{2\tau}\\
            u'(D, D_s, q, \tau, o = 1) = \frac{q(D) - (l - \tau)}{2\tau}
        \end{array}
    \right. \label{eq:uprime_2}
\end{align}
\begin{align}
    \footnotesize
    \text{If $q(D)\in (q(D_s), r + \tau)$}
    \left\{
        \begin{array}{lr}
            u'(D, D_s, q, \tau, o = 0) = \frac{q(D) - q(D_s)}{2\tau}\\
            u'(D, D_s, q, \tau, o = 1) = 1 - \frac{q(D) - q(D_s)}{2\tau}
        \end{array}
    \right. \label{eq:uprime_3}
\end{align}
%where $\intvl = (l, r) = (q(D_s) - \tau, q(D_s) + \tau)$ (from (\ref{eq:interval})).
% \blue{Note that when $q(D) = q(D_s)$, (\ref{eq:uprime_2}) and (\ref{eq:uprime_3}) are equal.}
We illustrate $u'$ in Figure~\ref{fig:smooth_score_lin_q}. 
Note that $u'$ outputs a value in the range $[0, 1]$ by definition, and not just in the set $\{0, 1\}$. It allows for a more gradual transition between scores of $0$ and $1$.

% red, green, blue, cyan, magenta, yellow, black, gray, darkgray, lightgray, brown, lime, olive, orange, pink, purple, teal, violet and white
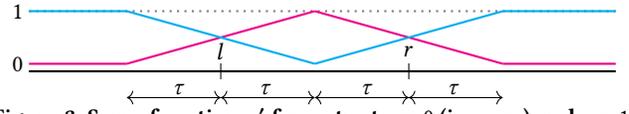
\begin{figure}[t] %[!ht]
    \centering
    \begin{tikzpicture}
        % l
        \coordinate (A5) at ($(2.75cm,0.5)$) {};
        \draw ($(A5)+(0,3pt)$) -- ($(A5)-(0,3pt)$);
        \node at ($(A5)+(0,2ex)$) {$l$};
        
        % r
        \coordinate (A11) at ($(5.25cm,0.5)$) {};
        \draw ($(A11)+(0,3pt)$) -- ($(A11)-(0,3pt)$);
        \node at ($(A11)+(0,2ex)$) {$r$};
        
        % s/2 on the left
        \coordinate (L1) at ($(2.2cm,0.2)$) {};
        \node at ($(L1)+(0,0.5ex)$) {$\tau$};
        \coordinate (L2) at ($(3.35cm,0.2)$) {};
        \node at ($(L2)+(0,0.5ex)$) {$\tau$};
        
        % s/2 on the right
        \coordinate (R1) at ($(4.7cm,0.2)$) {};
        \node at ($(R1)+(0,0.5ex)$) {$\tau$};
        \coordinate (R2) at ($(5.85cm,0.2)$) {};
        \node at ($(R2)+(0,0.5ex)$) {$\tau$};
        
        \draw[thick, black, -] (0.2, 0.5) -- (8, 0.5);
        \coordinate (L0) at ($(0.1*0.5cm,0.6)$) {};
        \node at ($(L0)$) {$0$};
        
        \draw[dotted, thick, gray, -] (0.2, 1.3) -- (8, 1.3);
        \coordinate (L1) at ($(0.1*0.5cm,1.3)$) {};
        \node at ($(L1)$) {$1$};
        
        \draw[thick, magenta, -] (0.2, 0.6) -- (1.5, 0.6);
        \draw[thick, magenta, -] (1.5, 0.6) -- (4, 1.3);
        % \draw[thick, magenta, -] (3.5, 1.3) -- (4.5, 1.3);
        \draw[thick, magenta, -] (4, 1.3) -- (6.5, 0.6);
        \draw[thick, magenta, -] (6.5, 0.6) -- (8, 0.6);
        
        \draw[thick, cyan, -] (0.2, 1.3) -- (1.5, 1.3);
        \draw[thick, cyan, -] (1.5, 1.3) -- (4, 0.6);
        % \draw[thick, cyan, -] (3.5, 0.6) -- (4.5, 0.6);
        \draw[thick, cyan, -] (4, 0.6) -- (6.5, 1.3);
        \draw[thick, cyan, -] (6.5, 1.3) -- (8, 1.3);
        
        \draw[thin, <->] (1.5, 0.15) -- (2.75, 0.15);
        \draw[thin, <->] (2.75, 0.15) -- (4, 0.15);
        
        \draw[thin, <->] (4, 0.15) -- (5.25, 0.15);
        \draw[thin, <->] (5.25, 0.15) -- (6.5, 0.15);
    \end{tikzpicture}
    \caption{Score function $u'$ for output $o=0$ (in cyan) and $o=1$ (in pink). Recall that $r - l = 2\cdot \tau$.}
    \label{fig:smooth_score_lin_q}
\end{figure}

We give the full proof of Proposition~\ref{prop:uprime_sens} in Appendix~\ref{subsec:appendix_uprime_sens}.

\begin{proposition}\label{prop:uprime_sens}
    The global sensitivity of $u'$ is $1/2\tau$.
\end{proposition}

We refer to the algorithm that directly instantiates the EM (Definition~\ref{def:EM}) with $\mathcal{R} = \{0, 1\}$ and score function $u'$ as \emcount. Since we post-process a DP estimate (Proposition~\ref{prop:DP-comp-post}), the following holds:

\begin{observation}
    \emcount\ satisfies $\epsilon$-DP.
\end{observation}

We now give an upper bound for the effectiveness threshold (proof in Appendix~\ref{subsec:appendix_tau_min_emcount}).

\begin{proposition}\label{prop:tau_min_emcount}
    Given a private database $D$, 
    synthetic database $D_s$, 
    \countQuery\ query $q$, 
    privacy budget $\epsilon$, and 
    error probability $0 < \delta < 1$, 
    the effectiveness threshold of the per-query decider \emcount\ at $\delta$ (Definition~\ref{def:effectiveness_metric}) has the following upper bound:
    $\tau_{min}^{\emcount, \delta} \leq \frac{1}{\epsilon} \ln{\frac{1 - \delta}{\delta}}$.
\end{proposition}

\paratitle{Comparison of \lapcount\ and \emcount}
Figure~\ref{fig:LM_Algo1_comparison} depicts the error (along the Y axis) for fixed $D_s, q$, and $\epsilon$ but varying $q(D)$ values (along the X axis) when $\tau = 10$ (on the left) and $\tau = 30$ (on the right).
We plot the error profile for \lapcount\ following Proposition~\ref{prop:lapcount_err}. Note that it gives upper bounds when $q(D)\leq l$ or $q(D)\geq r$. We plot the error profile for \emcount\ based on Definition~\ref{def:EM} for the EM with score function $u'$ ((\ref{eq:uprime_1}) to (\ref{eq:uprime_3})).
% The dotted lines are drawn at $x = l, q(D_s)$, and $r$.
For $\tau = 10$ and $q(D)\in \intvl$, \emcount's error is smaller. At $q(D) = q(D_s)$, \lapcount's error $e^{-\epsilon\tau}$ is larger than \emcount's error $\frac{1}{1 + e^{\epsilon\tau}}$ ($\epsilon\tau > 0$).

\begin{figure}[t]
    \centering
    \includegraphics[width=0.88\columnwidth]{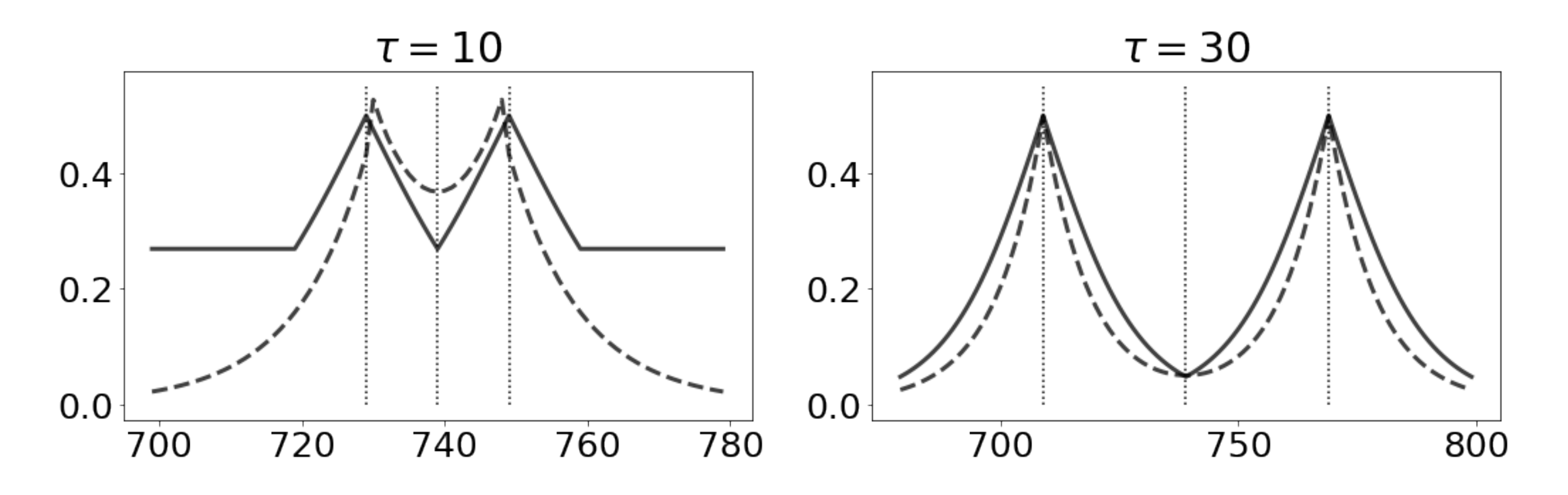}
    \caption{Error from \lapcount\ (dashed) and \emcount\ (solid) as $q(D)$ is varied along the X axis, where $D$ is the Adult dataset. The selection predicate is \texttt{$relationship$ LIKE `Unmarried' AND $sex$ LIKE `Male'}, $q(D_s) = 749, \epsilon = 0.1$ and $\tau = 10, 30$. \revm{The vertical lines mark the positions of $l, q(D_s)$, and $r$, respectively.}}
    \label{fig:LM_Algo1_comparison}
\end{figure}

\section{Solutions for \sumQuery\ query}\label{sec:sum}
We now discuss our solutions for \sumQuery\ query on attribute $A_i$: 
($1$) \lapsum\ (that instantiates Algorithm~\ref{algo:plug_in} with the Laplace Mechanism (LM)) in Section~\ref{subsec:lapsum}, 
($2$) \rtwotsum\ (that instantiates Algorithm~\ref{algo:plug_in} with the {\em Race-to-the-Top} (R2T) mechanism \cite{10.1145/3514221.3517844}) in Section~\ref{subsec:rtwotsum}, and 
($3$) \svtsum\  (that directly solves the problem using the {\em Sparse Vector Technique} (SVT) \cite{10.14778/3055330.3055331}) in Section~\ref{subsec:svtsum}. 
We derive upper bounds for their effectiveness thresholds and then compare them. %in Section~\ref{subsec:compareSum}.

Consider the \sumQuery\ query $q:$ \texttt{SELECT SUM($A_i$) FROM $D$} WHERE \whereClause, where \whereClause\ denotes the predicate in the WHERE clause and is empty. 
Its global sensitivity $\Delta q$ equals $\max dom(A_i)$ and is unbounded if the domain is unbounded. As done in previous work \cite{pmlr-v97-amin19a, DBLP:journals/corr/abs-1905-03871, NEURIPS2021_da54dd5a, 10.1145/3514221.3517844}, we assume a bound of \GSbound\ on $\Delta q$ and use it in the per-query deciders. Thus, we use \GSbound\ as the global sensitivity in the analysis.

\subsection{Laplace Mechanism-based Approach}\label{subsec:lapsum}
Our algorithm \lapsum\ works similarly to \lapcount. The only difference is that the scale of the Laplace distribution for noise $\nu_q$ is now $\GSbound/\epsilon$ instead of $1/\epsilon$. 
\lapsum\ works by running \BasicDecider\ (Algorithm~\ref{algo:plug_in}) with the LM (Definition~\ref{def:LM}) as \dpmech.
Since we post-process a DP estimate (Proposition~\ref{prop:DP-comp-post}):
% \begin{enumerate}
%     \item Compute $q(D) + \nu_q$, where $\nu_q\sim Lap(\GSbound/\epsilon)$.
%     \item If $-\tau < q(D) + \nu_q - q(D_s) < \tau$, then return ``Distance bound satisfied''. Else, return ``Distance bound unmet''.
% \end{enumerate}

\begin{observation}
    \lapsum\ satisfies $\epsilon$-DP.
\end{observation}

\lapsum\ suffers from the drawback that $GS_q$ is often large in practice, resulting in high variance in $\nu_q$.
% \sr{can we give evidence that it is large in practice? how about saying when the domain is unbounded and comprise large numbers?} 
For example, $A_i$ can represent incomes, distances, etc. and may contain large numbers.
This can cause \lapsum\ to make mistakes with higher probability. We derive an upper bound for the effectiveness threshold below.

% For a true positive, $Pr[\nu_q\geq \tau - \distBd] < e^{-2}$ gives $\tau > \frac{(2 - \ln{2})\GSbound}{\epsilon} + \distBd$. 
% For a true negative, $Pr[\nu_q < \tau - \distBd] < e^{-2}$ gives $\tau < \distBd - \frac{(2 - \ln{2})\GSbound}{\epsilon}$ (assuming RHS of the inequality is positive). 
% Using the same reasoning as that for bounding $\epsilon$ for \lapcount\ and rearranging the terms, we get that $\tau$ has a linear dependence on \GSbound\ here. 

\begin{proposition}\label{prop:tau_min_lapsum}
    Given a private database $D$, 
    synthetic database $D_s$, 
    \sumQuery\ query $q$ with global sensitivity \GSbound, 
    privacy budget $\epsilon$, and 
    error probability $0 < \delta < 1$, 
    the effectiveness threshold of the per-query decider \lapsum\ at $\delta$ (Definition~\ref{def:effectiveness_metric}) has the following upper bound: 
    $\tau_{min}^{\lapsum, \delta} \leq \frac{\GSbound}{\epsilon} \ln{\frac{1}{2\delta}}$.
\end{proposition}
\begin{proof}[Proof Sketch]
    The proof is the same as that of Proposition~\ref{prop:tau_min_lapcount} (full proof given in Appendix~\ref{subsec:appendix_tau_min_lapsum}).
    Since $\nu_q\sim Lap(\GSbound/\epsilon)$, bounds (\ref{eq:lap-geq})-(\ref{eq:lap-mod}) now use $\frac{\GSbound}{\epsilon}$ than $\frac{1}{\epsilon}$.
\end{proof}

We now give an example to illustrate how $\tau_{min}^{\lapsum, \delta}$ is computed. We will refer to this example again in Sections~\ref{subsec:rtwotsum} and \ref{subsec:svtsum}.

\begin{example}\label{example:tau_bound_sum_LM}
    Let $D$ be the database derived from IPUMS-CPS \cite{IPUMS} with attributes $age,$ $sex, education$ and $income$-$total$, and $D_s$ be its synthetic copy from an SDG.  %\sdg.  
    Consider the query $q$: \texttt{SELECT SUM($income$-$total$) FROM $D_s$ WHERE $age \leq 18$}. Suppose $\GSbound = 2M, \DLS = 9K$ and $\epsilon = 0.1$. 
    For $\delta = 0.05$, we get $\tau_{min}^{\lapsum, \delta} = \frac{\GSbound}{\epsilon} \ln{\frac{1}{2\delta}} = 4.605\times 10^7$.
\end{example}
% https://www.wolframalpha.com/input?i=log%281%2F0.05%292000000%2F0.8

Observe that $\tau_{min}^{\lapsum, \delta}$ is proportional to \GSbound, so it is large in part due to \GSbound\ being large. If \GSbound\ was $2$, then $\tau_{min}^{\lapsum, 0.05} = 46.052$. Thus, $\tau_{min}^{\lapsum, \delta}$ highly depends on $dom(A_i)$ for sum, as expected.

% \textit{What about \lapsum's error for the given problem instance?}
% Recall Proposition~\ref{prop:lapcount_err}, where we calculated the probability of \lapcount\ making a mistake when: ($1$) $q(D)\leq l$, ($2$) $q(D)\geq r$, and ($3$) $q(D)\in \intvl$. Observe that the noise scale for $\nu_q$ in \lapsum\ is \GSbound\ times larger than that in \lapcount, and we end up with the same expressions except the exponents get divided by \GSbound. 
% For the given instance, we evaluate these expressions and use the $\min$ to bound the probability of \lapsum\ making a mistake, i.e., \lapsum's error.\\
% % https://www.wolframalpha.com/input?i=substitute+a+%3D+1419469142%2C+b+%3D+1422461794+in+%28e%5E%28-%28b+-+1000000+-+a%29*0.8%2F2000000%29+-+e%5E%28-%28b+%2B+1000000+-+a%29*0.8%2F2000000%29%29%2F2
% \shweta{Make prop.}

\subsection{R2T-based Approach}\label{subsec:rtwotsum}
Our next algorithm, \rtwotsum, uses a state-of-the-art DP mechanism for sum queries called {\em Race-to-the-Top} (R2T) \cite{10.1145/3514221.3517844} to obtain a DP estimate for $q(D)$. 
R2T constructs $\log(\GSbound)$ number of queries (described below) with global sensitivities bounded by $2, 4, 8, \ldots, \GSbound$. Then it computes a noisy estimate for each query using the LM (Definition~\ref{def:LM}) and returns the max of the largest noisy estimate and $0$ (for non-negative integer-valued domain).

In \rtwotsum, the $j$-th query (for $j = 1$ to $\log(\GSbound)$) works by running the original \sumQuery\ query $q$ on a truncated database constructed by removing tuples in $D$ with $A_i$ value greater than $t_j = 2^j$. The output is denoted by $q(D, t_j)$.
In Appendix~\ref{subsec:appendix_R2T_truncSatProperties_Thm5pt1},
we show that this choice of {\em truncation function} for constructing queries satisfies the required properties as stated in \cite{10.1145/3514221.3517844}. 

R2T \cite{10.1145/3514221.3517844} computes the noisy estimates for the constructed queries as follows (for a given parameter $0 < \beta < 1$, which corresponds to the confidence bound as equation (\ref{eq:R2T_thm5-1}) below will show):
\begin{footnotesize}
    \begin{align}
        \tilde{q}(D, t_j) = q(D, t_j) + Lap\left( \frac{t_j}{\epsilon/\log(\GSbound)} \right) - \frac{t_j}{\epsilon/\log(\GSbound)} \ln{\frac{\log(\GSbound)}{\beta}}\label{eq:R2T-queries}
    \end{align}
\end{footnotesize}
where $t_j = 2^j, j = 1$ to $\log(\GSbound)$. The final estimate for $q(D)$ is:
\begin{footnotesize}
    \begin{align}
        \tilde{q}(D) = \max \{\max_j \tilde{q}(D, t_j), q(D, 0)\}\label{eq:R2T-output}
    \end{align}
\end{footnotesize}

In our setting, the downward local sensitivity \DLS\ (Definition~\ref{def:DLS}) is the largest $A_i$ value from tuples in $D$ that satisfy the WHERE condition \whereClause\ in $q$, because removing the matching tuple gives the worst-case change for $D\approx D', D'\subseteq D$.
% When \GSbound\ is much larger than \DLS, the error in the estimate $\tilde{q}(D)$ may be much smaller than what we get in \lapsum.
The following result from \cite{10.1145/3514221.3517844} holds (discussed further in Appendix~\ref{subsec:appendix_R2T_truncSatProperties_Thm5pt1})
\begin{footnotesize}
    \begin{align}
        Pr\left[ q(D) \geq \tilde{q}(D) \geq q(D) - 4\log(\GSbound) \ln{\left( \frac{\log(\GSbound)}{\beta} \right)} \frac{\DLS}{\epsilon} \right]\geq 1 - \beta\label{eq:R2T_thm5-1}
    \end{align}
\end{footnotesize}
% i.e., the error in R2T's final estimate is $O(\log(\GSbound)\log{\log(\GSbound)})$ with probability $\geq 1 - \beta$.
In other words, the probability that the noisy estimate $\tilde{q}(D)$ for $q(D)$ in \rtwotsum\ lies in the range $[q(D) - 4\log(\GSbound) \ln{\left( \frac{\log(\GSbound)}{\beta} \right)}\cdot$ $\frac{\DLS}{\epsilon},$ $q(D)]$ is at least $1 - \beta$.

To summarize, \rtwotsum\ works by running \BasicDecider\ (Algorithm~\ref{algo:plug_in}) with R2T \cite{10.1145/3514221.3517844} (and the aforementioned truncation function) as \dpmech, with parameters $\GSbound$ and $0 < \beta < 1$ for confidence bound in (\ref{eq:R2T_thm5-1}). Since we post-process a DP estimate (Proposition~\ref{prop:DP-comp-post}), the following holds:

\begin{observation}
    \rtwotsum\ satisfies $\epsilon$-DP.
\end{observation}

We give an upper bound on the effectiveness threshold of \rtwotsum\ using (\ref{eq:R2T_thm5-1}) (proof in Appendix~\ref{subsec:appendix_tau_min_rtwotsum}).

\begin{proposition}\label{prop:tau_min_rtwotsum}
    Given a private database $D$, 
    synthetic database $D_s$, 
    \sumQuery\ query $q$ with global sensitivity \GSbound and downward local sensitivity \DLS, 
    privacy budget $\epsilon$, and 
    error probability $0 < \delta < 1$, 
    the effectiveness threshold of the per-query decider \rtwotsum\ at $\delta$ (Definition~\ref{def:effectiveness_metric}) has the following upper bound:  
    $\tau_{min}^{\rtwotsum, \delta} \leq 4\log(\GSbound) \ln{\left( \frac{\log(\GSbound)}{\delta} \right)} \frac{\DLS}{\epsilon}$.
\end{proposition}

We now illustrate how to compute $\tau_{min}^{\rtwotsum, \delta}$ for Example~\ref{example:tau_bound_sum_LM}.

\begin{example}\label{example:tau_bound_sum_R2T}
    Recall the setting in Example~\ref{example:tau_bound_sum_LM}.
    As a result, 
    \begin{align*}
        \tau_{min}^{\rtwotsum, \delta} = 4\log(\GSbound) \ln{\left( \frac{\log(\GSbound)}{\delta} \right)} \frac{\DLS}{\epsilon} = 4.549\times 10^7
    \end{align*}
    which is less than $\tau_{min}^{\lapsum, \delta} = 4.605 \times 10^7$.
\end{example}
% https://www.wolframalpha.com/input?i=4*log_2%282000000%29*log%28log_2%282000000%29%2F0.05%29*10000%2F0.8

\revb{
Observe that $\tau_{min}^{\rtwotsum, \delta}$ is proportional to $\log{(\GSbound)} \log{\log{(\GSbound)}}\cdot$ $\DLS$, so it is large in part due to \GSbound\ being large. In Example~\ref{example:tau_bound_sum_R2T}, $\tau_{min}^{\rtwotsum, \delta} < \tau_{min}^{\lapsum, \delta}$ due to the large gap in \GSbound\ and \DLS. In our experiments (Section~\ref{sec:experiments}), we empirically show how the errors for \sumQuery\ queries vary in datasets with larger and smaller \GSbound\ values.
}

\subsection{SVT-based Approach}\label{subsec:svtsum}
The {\em Sparse Vector Technique} (SVT) \cite{DP_book, 10.14778/3055330.3055331} is a DP mechanism to report whether the output of a query $q$ on $D$ exceeds its specified threshold. 
It compares noisy versions of $q(D)$ and the threshold, and gives a yes ($\top$ = noisy query answer exceeded noisy threshold) or no ($\bot$ = noisy query answer did not exceed noisy threshold) answer. An advantage of using SVT is that it consumes privacy budget only if output = $\top$. 
The SVT algorithm is applied on a sequence of queries, each with its own threshold, resulting in an output given by $\{\top, \bot\}^l$, where $l$ is the number of queries answered. The balance on the privacy budget degrades with the number of queries with output = $\top$ until the privacy budget runs out \cite{10.14778/3055330.3055331} and SVT stops. 

We denote the per-query decider for \sumQuery\ query $q$ that runs SVT on a sequence of sum queries (described below) as \svtsum\ (Algorithm~\ref{algo:err_svtsum}). 
Given privacy budget $\epsilon$, \svtsum\ begins by sampling a noise term $\rho$ distributed as $Lap(\frac{1}{\epsilon/2})$ (line~\ref{l:div_budget}) and uses it to get the noisy thresholds for all queries. 
The remaining $\epsilon/2$ privacy budget is used towards the first $\top$ in the output, after which \svtsum\ stops.
We modify the queries from \rtwotsum\ (Section~\ref{subsec:rtwotsum}): $q_j(D)$ $= q(D, t_j)/t_j$, for $t_j = 2^j, j = 1$ to $\log(\GSbound)$. We divide by $t_j$ so that $\forall j, \Delta q_j = 1$ 
\revc{
in the worst-case (i.e. the max over the sensitivities of the input queries now does not exceed $1$), allowing the noise scales to be proportional to $1$ than \GSbound. The relevant threshold values are $l$ and $r$, which must also be divided by $t_j$.}
First, SVT is run on queries $q_j(D)$ with threshold $T_j = r/t_j$ (loop in line~\ref{l:1st_SVT}), adding independent Laplace noise distributed as $Lap(\frac{1}{\epsilon/2})$ to each $q_j(D)$. If any noisy $q_j(D)$ value results in $\top$, then $o = 0$ is returned (line~\ref{l:exceed_r}). 
If not, only $\epsilon/2$ privacy budget has been consumed by $\rho$ and SVT is run on queries $q_j(D)$ with threshold $T_j = (l + 1)/t_j$ (loop in line~\ref{l:2nd_SVT}), again adding independent Laplace noise distributed as $Lap(\frac{1}{\epsilon/2})$ to each $q_j(D)$. 
If any noisy $q_j(D)$ value results in $\top$, then $o = 1$ is returned (line~\ref{l:exceed_lplus1}). 
Otherwise, $o = 0$ is returned at the end (line~\ref{l:not_in_interval}). 

\begin{definition}\label{def:SVT_monotonic}
    In the context of SVT \cite{10.14778/3055330.3055331}, the sequence of queries used is called monotonic if, in going from $D$ to a neighboring database $D'$, all query answers that are different change in the same direction, i.e., they all increase or they all decrease. 
\end{definition}

Observe that $q_j(D)$s are monotonic because $\forall j, q(D, t_j)\geq q(D',$ $t_j)$ or $\forall j, q(D, t_j)\leq q(D', t_j)$. This helps save a factor of $2$ in the noise scale for $\nu_j$ (discussed further in Appendix~\ref{subsec:appendix_MonoQueries_isEpsDP}).
Unlike the previous per-query deciders that post-process a DP estimate, here we need to show that \svtsum\ preserves DP.

\IncMargin{1em}
\begin{algorithm}[!ht]
    \caption{Per-query decider \svtsum\ for \sumQuery}
    \label{algo:err_svtsum}
    
    \SetKwInOut{Input}{Input}\SetKwInOut{Output}{Output}
    \LinesNumbered
    \Input{$q$ - \sumQuery\ query, $D$ - private database, $D_s$ - synthetic database, $\tau$ - distance bound, $\epsilon$ - privacy budget}
    \Output{$o = 1$ if the desired distance bound from $q(D_s)$ is satisfied for $q(D)$, else $o = 0$.}
    
    \BlankLine
    
    \SetKwFunction{FMain}{\svtsum}
    \SetKwProg{Fn}{Function}{:}{}
    \Fn{\FMain{$q, D, D_s, \tau, \epsilon$}}{
        $A_i\gets$ aggregate attribute in $q$\;
        $l\gets q(D_s) - \tau, r\gets q(D_s) + \tau$\;
        %$\epsilon_1, \epsilon_2 \gets \epsilon/2$\; \label{l:div_budget}
        $\rho \gets Lap(\frac{1}{\epsilon/2})$\; \label{l:div_budget}
        \For{$j\in \{1, 2, 3, \ldots, \lceil \log(\GSbound) \rceil \}$}
        {
            \label{l:1st_SVT}
            $t_j \gets 2^j, q(D, t_j) \gets$ \texttt{SELECT SUM($A_i$) FROM $D$ WHERE \whereClause\ AND $A_i \leq t_j$}\;
            $q_j(D) \gets q(D, t_j)/t_j, \nu_{j} \gets Lap(\frac{1}{\epsilon/2})$\;
            \If{$q_j(D) + \nu_j \geq r/t_j + \rho$}
            {\label{l:check_r}
                \Return $o = 0$ (``Distance bound unmet'')\; \label{l:exceed_r}
            }
        }
        \For{$j\in \{1, 2, 3, \ldots, \lceil \log(\GSbound) \rceil \}$}
        {
            \label{l:2nd_SVT}
            $t_j \gets 2^j, q(D, t_j) \gets$ \texttt{SELECT SUM($A_i$) FROM $D$ WHERE \whereClause\ AND $A_i \leq t_j$}\;
            $q_j(D) \gets q(D, t_j)/t_j, \nu_{j} \gets Lap(\frac{1}{\epsilon/2})$\;
            \If{$q_j(D) + \nu_j \geq (l + 1)/t_j + \rho$}
            {\label{l:check_l}
                \Return $o = 1$ (``Distance bound satisfied'')\; \label{l:exceed_lplus1}
            }
        }
        \Return $o = 0$ (``Distance bound unmet'')\; \label{l:not_in_interval}
    }
\end{algorithm}
\DecMargin{1em}

\begin{theorem} \label{thm:SVTsum_is_DP}
    \svtsum\ is $\epsilon$-DP.
\end{theorem}
\begin{proof}[Proof Sketch]
    Let $\top_1$ and $\top_2$ denote the `yes' answers from SVT, i.e., when the noisy query answer exceeds the noisy threshold in lines~\ref{l:exceed_r} and \ref{l:exceed_lplus1}, respectively. Let $\bot_1$ and $\bot_2$, respectively, denote when these checks fail. Since the algorithm stops as soon as a $\top_1$ or $\top_2$ is returned in lines~\ref{l:exceed_r} or \ref{l:exceed_lplus1}, the output string $a$ is either of the form $\bot_1, \ldots, \bot_1, \top_1$ or $\bot_1, \ldots, \bot_1, \bot_2, \ldots, \bot_2, \top_2$. For both forms, we show in the full version that  $Pr[\svtsum(q, D, D_s, \tau, \epsilon) = a] \leq e^{\epsilon} Pr[\svtsum(q, D', D_s, \tau, \epsilon) = a]$ for $D\approx D'$ holds adapting ideas from \cite{10.14778/3055330.3055331}. Intuitively, \svtsum\ is equivalent to running SVT once, with the first half mapped to $o = 0$ and the remaining half mapped to $o = 1$. A complete proof is given in Appendix~\ref{subsec:appendix_MonoQueries_isEpsDP}.
\end{proof}

\svtsum\ incurs high error if $q(D)$ and $r$ (or $l$) are close because in the later iterations where $t_j$ values are large, the check is easily influenced by noise.
Note that $\forall t_k\geq \DLS, q(D, t_k) = q(D)$, where \DLS\ is the downward local sensitivity (Definition~\ref{def:DLS}).
We present an optimization 
in Algorithm~\ref{algo:upBound_DS_fromSVT}
in Appendix~\ref{subsec:appendix_SVT_optimization}
aimed at improving the accuracy of \svtsum\ by obtaining a private bound for \DLS\ to be used as the largest truncation threshold.
In the rest of the paper, by \svtsum\ we will refer to the improved algorithm using the bound from 
Algorithm~\ref{algo:upBound_DS_fromSVT}.
We analyze the error of \svtsum\ in Appendix~\ref{subsec:appendix_tau_min_svtsum}.\\

\paratitle{Comparison of \lapsum, \rtwotsum\ and \svtsum}
As demonstrated by Examples~\ref{example:tau_bound_sum_LM} and \ref{example:tau_bound_sum_R2T}, when the difference between \DLS\ and \GSbound\ is large, the effectiveness threshold for \rtwotsum\ is likely to be smaller than that of \lapsum. We show that \svtsum\ can achieve smaller error than \rtwotsum\ in the experiments (Section~\ref{sec:experiments}).

%\sr{do you have any advice or insight on which algo to use in what cases}

\section{Solutions for \medianQuery\ query}\label{sec:median}
We present two solutions and analyze their errors for \medianQuery\ query $q$ on attribute $A_i$: \texttt{SELECT MEDIAN($A_i$) FROM $D$ WHERE \whereClause}. ($1$) \emmed\ (that instantiates Algorithm~\ref{algo:plug_in} with the Exponential Mechanism (EM)) in Section~\ref{subsec:emmed}, and ($2$) \histmed\ (that directly solves the problem using a noisy histogram) in Section~\ref{subsec:histmed}.
The true output of the median query $q(D)$ is the $\lceil \frac{n'}{2} \rceil$-th element in the sorted list of $A_i$ values among tuples that satisfy the WHERE clause, where $n'$ is the (private) number of tuples in $D$ satisfying $\varphi$.

\subsection{Exponential Mechanism-based approach}\label{subsec:emmed}
Let $rank_{\whereClause}(D, e)$ be the output of the query: \texttt{SELECT COUNT(*) FROM $D$ WHERE \whereClause\ AND $A_i < e$}.
% Let $n'$ be the output of the query: \texttt{SELECT COUNT(*) FROM $D$ WHERE \whereClause}.
%Let $n'$ be the number of tuples in $D$ that satisfy \whereClause.
Our approach \emmed\ uses the algorithm from \cite{DBLP:journals/corr/abs-1103-5170} that computes a noisy estimate for $q(D)$.
\emmed\ runs \BasicDecider\ with the EM as \dpmech, with additional parameters $\mathcal{R} = dom(A_i)$ and score function $u(D, e) = -|rank_{\whereClause}(D, e) - \frac{n'}{2}|$, $\forall e\in \mathcal{R}$. The sensitivity of the score function equals $1$ because rank of any $e$ either stays the same or changes in the same direction as $n'$ between databases $D\approx D'$.
% \begin{enumerate}
%     \item Compute an estimate of $q(D)$ using the EM with $u(D, e) = -|rank_{\whereClause}(e) - \frac{n'}{2}|$, $\forall e\in dom(A_i)$. %that satisfy \whereClause. 
%     % http://dimacs.rutgers.edu/~graham/pubs/slides/privdb-tutorial.pdf
%     \item Plug in estimate and return $1$ if $\distBd < \tau$. Else, return $0$.
% \end{enumerate}
Since we post-process a DP estimate (Proposition~\ref{prop:DP-comp-post}), the following observation holds:

\begin{observation}
    \emmed\ satisfies $\epsilon$-DP.
\end{observation}

We analyze the error of \emmed\ in Appendix~\ref{subsec:appendix_emmed_err}.

\subsection{Histogram-Based Algorithm}\label{subsec:histmed}
%Let $n'$ be the number of tuples in $D$ that satisfy \whereClause.
We next propose a histogram-based approach called \histmed\ (Algorithm~\ref{algo:err_1_q_med}), which uses the intuition that if at least half the values in $A_i$ (from tuples satisfying \whereClause) either are less than or equal to $l$, or are greater than or equal to $r$, then $q(D)\not\in \intvl$. 
These bounds $l = q(D_s) - \tau$ and $r = q(D_s) + \tau$ are compared with a DP estimate, say $m$, for $\lceil \frac{n'}{2} \rceil$ obtained using the LM (line~\ref{l:n_private}, where $n'$ is the number of tuples in $D$ satisfying \whereClause ). 
If neither count exceeds $m$, then $o = 1$ (line \ref{l:satisfy_return}). Otherwise, $o = 0$ (lines~\ref{l:q1_return} and \ref{l:q2_return}).

% \sr{Minor but do not understand why this is called histogram-based algo.. there is no histogram?} \shweta{bins are $[0, l], (l, r), [r, \infty]$. We aren't using the 2nd bin}

% Useful if $\hat{\theta}$ for $D$ given?
\IncMargin{1em}
\begin{algorithm}[!ht]
    \caption{Per-query decider \histmed\ for \medianQuery}
    \label{algo:err_1_q_med}
    
    \SetKwInOut{Input}{Input}\SetKwInOut{Output}{Output}
    \LinesNumbered
    \Input{$q$ - \medianQuery\ query, $D$ - private database, $D_s$ - synthetic database, $\tau$ - distance bound, $\epsilon$ - privacy budget for error analysis}
    % $\epsilon$ - privacy budget for error analysis}
    \Output{Whether the distance bound is met for $q$}
    
    \BlankLine
    
    \SetKwFunction{FMain}{\histmed}
    \SetKwProg{Fn}{Function}{:}{}
    \Fn{\FMain{$q, D, D_s, \tau, \epsilon$}}{
        $n'\gets$ number of tuples in $D$ that satisfy \whereClause\ in $q$\;
        $\nu_q\gets Lap(\frac{1}{\epsilon/2}), \tilde{n}\gets n' + \nu_q$\; \label{l:n_private}
        $A_i\gets$ attribute for median used in $q$\;
        $q_1(D)\gets$ \texttt{SELECT COUNT(*) FROM $D$ WHERE \whereClause\ AND $A_i\leq q(D_s) - \tau$}\; \label{l:q1_def}
        $q_2(D)\gets$ \texttt{SELECT COUNT(*) FROM $D$ WHERE \whereClause\ AND $A_i\geq q(D_s) + \tau$}\; \label{l:q2_def}
        $\nu_{q_1}, \nu_{q_2}\gets Lap(\frac{1}{\epsilon/2})$\;
        \If{$q_1(D) + \nu_{q_1} \geq \lceil \tilde{n}/2 \rceil $}
        {   \label{l:q1_return}
            \Return $o = 0$ (``Distance bound unmet'')\;
        }
        \ElseIf{$q_2(D) + \nu_{q_2} \geq \lceil \tilde{n}/2 \rceil $}
        {
            \label{l:q2_return}
            \Return $o = 0$ (``Distance bound unmet'')\;
        }
        \Else
        {
            \Return $o = 1$ (``Distance bound satisfied'')\; \label{l:satisfy_return}
        }
    }
\end{algorithm}
\DecMargin{1em}

We next show that \histmed\ is $\epsilon$-DP.

\begin{proposition}\label{prop:properties_histmed}
    \histmed\ satisfies $\epsilon$-DP.
\end{proposition}
\begin{proof}[Proof Sketch]
    We spend $\epsilon/2$ to obtain an estimate for $n'$, and the remaining $\epsilon/2$ on $q_1(D)$ and $q_2(D)$ (lines~\ref{l:q1_def}-\ref{l:q2_def}) that use disjoint sets of tuples ($\tau > 0$). Hence \histmed\ satisfies $\epsilon$-DP by sequential and parallel composition, and post-processing (Proposition~\ref{prop:DP-comp-post}).
\end{proof}

\paratitle{Comparison of \emmed\ and \histmed}
Suppose $q$ is a query that computes the median on attribute $age, q(D) = 37, \epsilon = 0.1, \tau = 5$, and $\intvl = (28, 38)$ (as defined in (\ref{eq:interval})). \emmed\ returns $e = 38$ with probability equal to $0.9995$. $38\not\in \intvl$ but $q(D)\in \intvl$. 
\histmed\ can be the better choice (discussed further in Appendix~\ref{subsec:appendix_histmed_err}).\\

\revm{
\paratitle{Extending the framework to other aggregates} Our solutions can be used to support some other aggregates, e.g., 
%smore complex aggregates that can be decomposed into count, sum, and median queries. For example, 
average can be expressed as the output of a \sumQuery\ query divided by the output of a \countQuery\ query, each consuming some $\epsilon$. 
% Similarly for standard deviation, variance, Pearson's correlation coefficient, etc. 
%However, this requires further analysis to establish bounds. 
The solutions for \medianQuery\ can be generalized to work for other quantiles. For example, to compute the first quartile, we change the score function in \emmed\ to use $n'/4$ instead of $n'/2$, and change $\lceil \tilde{n}/2 \rceil$ to $\lceil \tilde{n}/4 \rceil$ and $\lceil 3\tilde{n}/4 \rceil$ in lines~\ref{l:q1_return} and \ref{l:q2_return} (Algorithm~\ref{algo:err_1_q_med}), respectively.
However, supporting more complex 
%SQL queries (with joins, subqueries) and other 
aggregates needs careful analysis to establish bounds.
}
\section{Experiments}\label{sec:experiments}
% We analyze the performance of our approach in these terms:
In this section, we analyze the accuracy and efficiency of our proposed per-query deciders for \countQuery, \sumQuery\ and \medianQuery\ queries with the following questions:
\begin{enumerate}
    \item How is the accuracy of each proposed solution affected when $\tau$ and $\epsilon$ are varied separately?
    
    \item For each proposed solution, what type of queries benefit most in terms of accuracy?
    % What type of queries benefit most from which solution in terms of accuracy?
    
    \item How does the performance of the specialized solutions compare with that of the solutions bases on Algorithm~\ref{algo:plug_in}?
\end{enumerate}
% \shweta{Fix $\tau$, vary $\epsilon$ -- $\epsilon$ vs error figure; also fix $\epsilon$ ---- Boxplot (type of query for each algo) or line graph}
% \shweta{Effectiveness ---- metric for which algo is more effective for this query}
% \shweta{Run-time}

We have implemented the per-query deciders in Python $3.8.8$ using Pandas \cite{reback2020pandas} and NumPy \cite{harris2020array} libraries. All experiments were run on Apple M$1$ CPU $@3.2$ GHz with $16$ GB of RAM.

\subsection{Experimental Setup}\label{subsec:setup}
We describe the datasets, queries, error measures, and parameters.% settings.

\paratitle{Dataset} \reva{We consider two datasets as the \textbf{private database $D$}.}
\begin{sloppypar}
\reva{(1) The first dataset is}
% We use a database 
derived from the {\bf IPUMS-CPS survey data } \cite{IPUMS}, %the largest 
an individual-level population database, for the years $2011$-$2019$ with $1,340,703$ tuples and $10$ attributes: $relate$, $age$, $sex$, $race$, $marst$, $citizen$, $workly$, $classwkr$, $educ$ and $inctot$. 
% as the private database $D$. 
The only 
% two 
numerical attributes are $age$ and $inctot$ with domains $\{0, 1, \ldots, 80, 85\}$ and $\{0, 1, \ldots,$ $999999999\}$, respectively. We only include tuples with $inctot$ value less than or equal to $500$K.
The domain sizes of the categorical attributes vary from $3$ to $36$. 
\reva{(2) The second dataset is derived from the {\bf NYC Yellow Taxi Trip data} \cite{TLC} for January $2022$ with $2,177,719$ tuples and $10$ attributes: $vendorID$, $passenger\_count$, $trip\_distance$, $rateCodeID$, $store\_and\_fwd\_flag$, $payment\_type$, $fare\_amount$, $tip\_amount$, $total\_amount$ and $congestion\_surcharge$ (with some pre-processing as discussed in Appendix~\ref{sec:appendix_allQueryErrorPlots}).
%For numerical attributes, we round to nearest integer and set domains as $[1, 6]$ for $passenger\_count$, $[1, 100]$ for $trip\_distance$, $[4, 250]$ for $fare\_amount$, $[0, 125]$ for $tip\_amount$, and $[8, 377]$ for $total\_amount$. 
The domain sizes of the categorical attributes vary from $2$ to $6$.}

We generate a \textbf{synthetic database $D_s$ for $D$} using PrivBayes \cite{10.1145/3134428}, a Bayesian network based DP-SDG. Nodes and edges in the network represent attributes in $D$ and conditional independence relations between attributes in $D$. 
PrivBayes first learns a differentially private Bayesian network \DPbayes\ and then uses it to derive a factored form of the joint tuple probabilities based on the noisy conditional probabilities. 
Note that \DPbayes\ can make incorrect conditional independence assumptions between attributes in $D$.
\end{sloppypar}

\smallskip
\paratitle{Queries} We refer to the Summary File $1$ (SF-$1$) \cite{2010SF1} released by the U.S. Census Bureau to construct 
% sets of 
queries \reva{for the IPUMS-CPS dataset}. 
% For \sumQuery\ queries, we compute sum on attribute $inctot$. For \medianQuery\ queries, we compute median on attribute $age$.
We analyze $12$ \countQuery, $9$ \sumQuery\ \reva{(on $inctot$)} and $9$ \medianQuery\ \reva{(on $age$)} queries. 
% in our experiments. 
\reva{For the second dataset, we analyze $10$ \countQuery, $10$ \sumQuery\ (on $total\_amount$), and $10$ \medianQuery\ (on $trip\_distance$) queries. 
Due to space constraints, here we present results on $4$ representative queries of each aggregate on the IPUMS-CPS data, and $2$ \sumQuery\ queries on the NYC Taxi Trip data with \DLS\ much smaller than \GSbound\ (Table~\ref{tbl:exp_rep_queries}), while the full list of queries and results for the other queries on both datasets are shown in Appendix~\ref{sec:appendix_allQueryErrorPlots}.}

\begin{table}[t] %[!ht]
    \centering \scriptsize
    \begin{tabular}{| c | c | c |}
        \rowcolor[HTML]{C0C0C0} \hline {\bf Query} & {\bf WHERE clause} & $q(D), q(D_s)$\\
        \hline \cellcolor[HTML]{D7D7D7} & \begin{tabular}{@{}c@{}} $q_1$: \texttt{$sex$ LIKE `Female' AND $race$ LIKE `White-}\\ \texttt{American Indian-Asian' AND $workly$ LIKE `Yes'} \end{tabular} & \begin{tabular}{@{}c@{}} $34$ \\ $34$ \end{tabular}\\
        \cline{2-3} \cellcolor[HTML]{D7D7D7} & \begin{tabular}{@{}c@{}} $q_3$: \texttt{$sex$ LIKE `Male' AND $educ$ LIKE `Doctorate}\\ \texttt{degree' AND $marst$ LIKE `Separated'} \end{tabular} & \begin{tabular}{@{}c@{}} $87$ \\ $91$ \end{tabular}\\
        \cline{2-3} \cellcolor[HTML]{D7D7D7} & \begin{tabular}{@{}c@{}} $q_5$: \texttt{$sex$ LIKE `Female' AND $educ$ LIKE `Grades $5$}\\ \texttt{or $6$' AND $marst$ LIKE `Never married/single'} \end{tabular} & \begin{tabular}{@{}c@{}} $1560$ \\ $1606$ \end{tabular}\\
        \cline{2-3} \multirow{-8}{*}{\cellcolor[HTML]{D7D7D7} 1. \countQuery} &  \begin{tabular}{@{}c@{}} $q_{12}$: \texttt{$race$ LIKE `White' AND $marst$ LIKE `Married,}\\ \texttt{spouse present' AND $citizen$ LIKE `Born in U.S'} \end{tabular} & \begin{tabular}{@{}c@{}} $471994$ \\ $470483$ \end{tabular}\\
        
        \hline \cellcolor[HTML]{D7D7D7} & \begin{tabular}{@{}c@{}} $q_{13}$: \texttt{$sex$ LIKE `Female' AND $race$ LIKE}\\ \texttt{`White-Black' AND $workly$ LIKE `No'} \end{tabular} & \begin{tabular}{@{}c@{}} $6915340$ \\ $6942866$ \end{tabular}\\
        \cline{2-3} \cellcolor[HTML]{D7D7D7} & \begin{tabular}{@{}c@{}} $q_{16}$: \texttt{$race$ LIKE `White' AND $marst$ LIKE `Divorced'}\\ \texttt{AND $citizen$ LIKE `Not a citizen'} \end{tabular} & \begin{tabular}{@{}c@{}} $123543040$ \\ $128497757$ \end{tabular}\\
        \cline{2-3} \cellcolor[HTML]{D7D7D7} & \begin{tabular}{@{}c@{}} $q_{20}$: \texttt{$sex$ LIKE `Female' AND $educ$ LIKE `High}\\ \texttt{school diploma or equivalent' AND $marst$}\\ \texttt{LIKE `Never married/single'} \end{tabular} &  \begin{tabular}{@{}c@{}} $685635093$ \\ $690711885$ \end{tabular}\\
        \cline{2-3} \multirow{-8}{*}{\cellcolor[HTML]{D7D7D7} {\begin{tabular}{@{}c@{}} 2. \sumQuery\\ on\\ $inctot$ \end{tabular}}} & \begin{tabular}{@{}c@{}} $q_{21}$: \texttt{$race$ LIKE `White' AND $marst$ LIKE `Married,}\\ \texttt{spouse present' AND $citizen$ LIKE `Born in U.S'} \end{tabular} &  \begin{tabular}{@{}c@{}} $23542765109$ \\ $23434676868$ \end{tabular}\\
        
        \hline \cellcolor[HTML]{D7D7D7} & \begin{tabular}{@{}c@{}} $q_{22}$: \texttt{$workly$ LIKE `Yes' AND $classwkr$ LIKE}\\ \texttt{`Wage/salary, private' AND $educ$ LIKE}\\ \texttt{`Bachelor's degree'} \end{tabular} & \begin{tabular}{@{}c@{}} $41$ \\ $41$ \end{tabular}\\
        \cline{2-3} \cellcolor[HTML]{D7D7D7} & \begin{tabular}{@{}c@{}} $q_{23}$: \texttt{$sex$ LIKE `Male' AND $race$ LIKE `White-}\\ \texttt{Black' AND $relate$ LIKE `Spouse'} \end{tabular} & \begin{tabular}{@{}c@{}} $40$ \\ $40$ \end{tabular}\\
        \cline{2-3} \cellcolor[HTML]{D7D7D7} & \begin{tabular}{@{}c@{}} $q_{28}$: \texttt{$race$ LIKE `Asian only' AND $marst$ LIKE}\\ \texttt{`Separated' AND $citizen$ LIKE `Born in U.S'} \end{tabular} & \begin{tabular}{@{}c@{}} $39$ \\ $39$ \end{tabular}\\
        \cline{2-3} \multirow{-8}{*}{\cellcolor[HTML]{D7D7D7} {\begin{tabular}{@{}c@{}} 3. \medianQuery\\ on\\ $age$ \end{tabular}}} & \begin{tabular}{@{}c@{}} $q_{29}$: \texttt{$sex$ LIKE `Male' AND $race$ LIKE `White'}\\ \texttt{AND $classwkr$ LIKE `Wage/salary, private'} \end{tabular} & \begin{tabular}{@{}c@{}} $40$ \\ $40$ \end{tabular}\\

        \hline \cellcolor[HTML]{D7D7D7}& \begin{tabular}{@{}c@{}} \reva{$q_{45}$:} \reva{\texttt{$passenger\_count < 2$ AND $trip\_distance = 8$}}\\\reva{\texttt{AND $tip\_amount \leq 2$}} \end{tabular} & \begin{tabular}{@{}c@{}} \reva{$218047$}\\\reva{$213685$} \end{tabular}\\
        \cline{2-3} \multirow{-3}{*}{\cellcolor[HTML]{D7D7D7} {\begin{tabular}{@{}c@{}} \reva{4. \sumQuery\ on}\\ \reva{$total\_$}\\\reva{$amount$} \end{tabular}}}& \begin{tabular}{@{}c@{}} \reva{$q_{50}$:} \reva{\texttt{$trip\_distance \leq 1$ AND $fare\_amount \leq 6$}}\\\reva{\texttt{AND $congestion\_surcharge = 2$}} \end{tabular} & \begin{tabular}{@{}c@{}} \reva{$5680776$}\\\reva{$5624898$} \end{tabular}\\
        
        \hline
    \end{tabular}
    \caption{Queries used in experiments. Blocks 1, 2, 3 are for the IPUMS-CPS data, block 4 is for the NYC Taxi Trip data. }\label{tbl:exp_rep_queries}
    \vspace{-4mm}
\end{table}

\paratitle{Error measure} We measure the error of per-query deciders as follows. We run each of our per-query deciders 100 times to decide whether $q(D)$ lies in $(q(D_s)\cdot (1-\tau), q(D_s)\cdot (1+\tau))$, where $\tau$ is a percentage of the query answer on the synthetic data, $q(D_s)$.
We measure error as the fraction of times the algorithm makes an error in determining whether $q(D)$ lies in $(q(D_s)\cdot (1-\tau), q(D_s)\cdot (1+\tau))$.
% A false positive occurs when the ground truth is that the distance bound is unmet but the output of the per-query decider \alg\ is $o = 1$. 
% A false negative occurs when the ground truth is that the distance bound is satisfied but \alg\ returns $o = 0$. 
% We measure error (Definition~\ref{def:error}) as the fraction of times \alg\ returns the wrong outcome over $100$ repeated runs.

\paratitle{Parameter settings} We set $\beta = 0.05$ in \rtwotsum\ (Section~\ref{subsec:rtwotsum}) and $\theta = 0.95$ in 
% Algorithm~\ref{algo:upBound_DS_fromSVT} 
Algorithm $4$
(Section~\ref{subsec:svtsum}).
Default $\epsilon = 0.25$ and $\tau = 3.2\%$ of $q(D_s)$. 
In our experiments, we vary $\tau = 0.2\%, 0.8\%, 3.2\%,$ $12.8\%, 51.2\%$, and vary $\epsilon = 0.0625, 0.125, 0.25, 0.5, 1$.

\subsection{Accuracy and Performance Analysis}\label{subsec:findings}
% \ashwin{For which of the settings is the algorithm supposed to be effective? Can you for instance underline the values when the algorithm is supposed to ve effective for that setting of eps and tau?}
% \ashwin{Would it make sense to compute for each query, the smallest epsilon at which an algorithm would be effective for $\tau =5\%$? }

\subsubsection{Accuracy analysis}\label{subsec:accuracy}
We present our analysis of the impact on accuracy as $\epsilon$ and $\tau$ are varied individually.
We also investigate which queries benefit the most for each per-query decider.
In the following discussion, we use $\intvl = (q(D_s)\cdot (1-\tau), q(D_s)\cdot (1+\tau))$, where $\tau$ is a percentage of $q(D_s)$.

\begin{figure*}[ht]
  \includegraphics[width=0.9\textwidth]{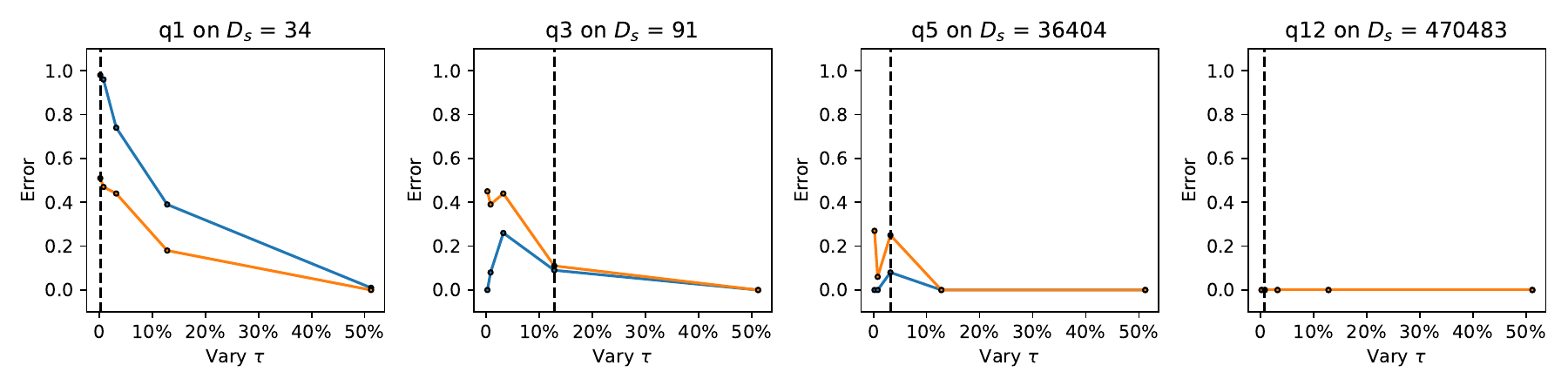}
  \caption{\reva{IPUMS-CPS data:} Error for \countQuery\ queries from \lapcount\ (in blue) and \emcount\ (in orange) as $\tau$ varies. 
  % \lapcount\ and \emcount\ are effective for $q_1$ at $\tau = 51.2\%$ and $q_8$ at $\tau = 0.2\%$. Cannot say based on effectiveness if these approaches give low error for $q_3$ (at $\tau\geq 3.2\%$) and $q_{12}$.
  \revc{The dotted line marks the smallest $\tau$ value considered such that the query answer on the private data belongs in the interval \intvl.}
  }
  \label{fig:count_pick4_vary_tau}
\end{figure*}

\begin{figure*}[ht]
  \includegraphics[width=0.9\textwidth]{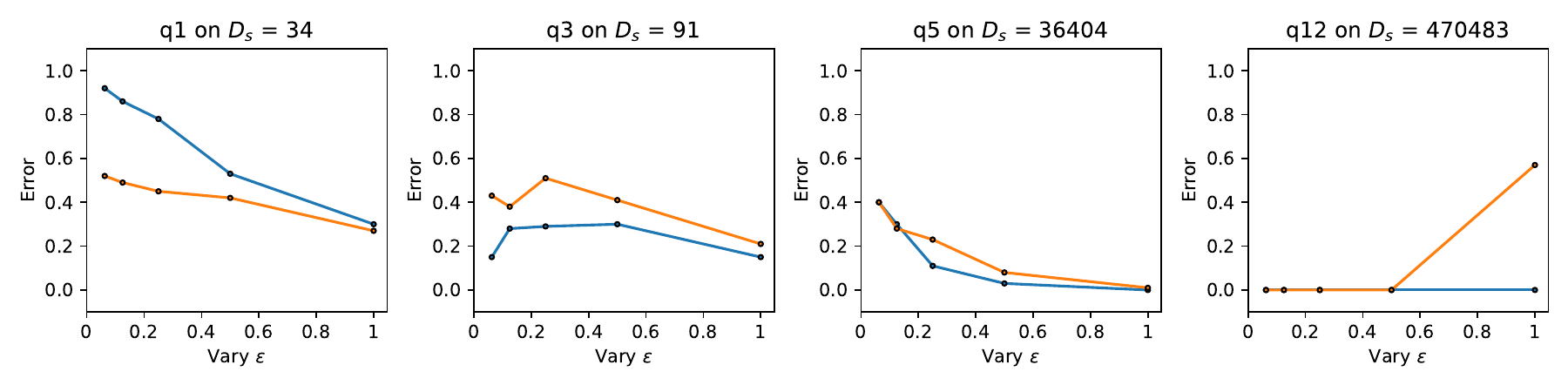}
  \caption{\reva{IPUMS-CPS data:} Error for \countQuery\ queries from \lapcount\ (in blue) and \emcount\ (in orange) as $\epsilon$ varies. 
  % \lapcount\ and \emcount\ are not effective for $q_1$. Cannot say based on effectiveness if these approaches give low error for $q_3, q_8$ and $q_{12}$.
  }
  \label{fig:count_pick4_vary_eps}
\end{figure*}

\begin{figure*}[ht]
  \includegraphics[width=0.9\textwidth]{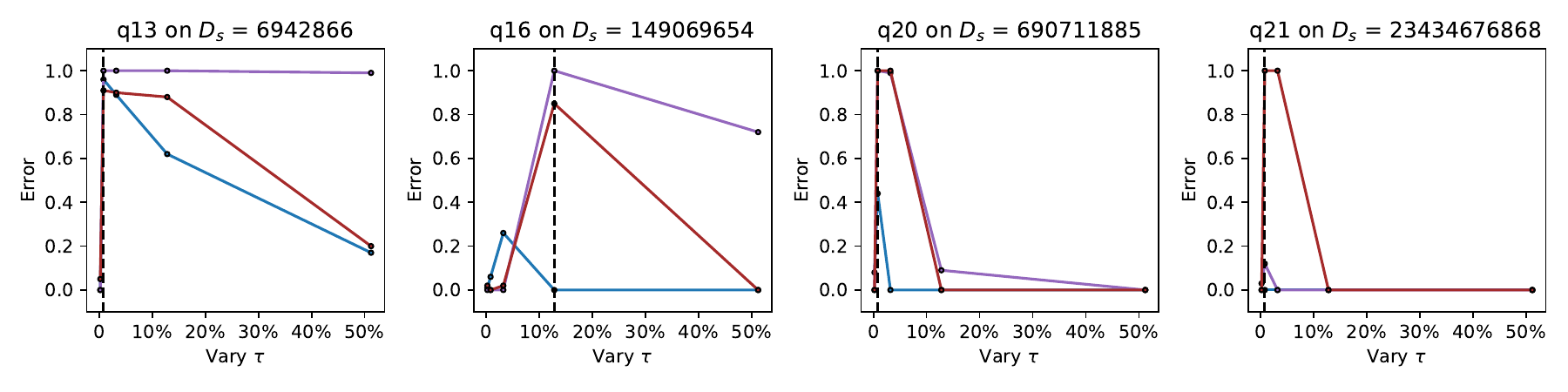}
  \caption{\reva{IPUMS-CPS data:} Error for \sumQuery\ queries from \lapsum\ (in blue), \rtwotsum\ (in purple) and \svtsum\ (in brown) as $\tau$ varies. 
  % \lapsum\ is effective for $q_{21}$ at $\tau = 0.2\%$. 
  % \ashwin{I am not able to distinguish between purple and brown}
  \revc{The dotted line marks the smallest $\tau$ value considered such that the query answer on the private data belongs in the interval \intvl.}
  }
  \label{fig:sum_pick4_vary_tau}
\end{figure*}

\begin{figure*}[ht]
  \includegraphics[width=0.9\textwidth]{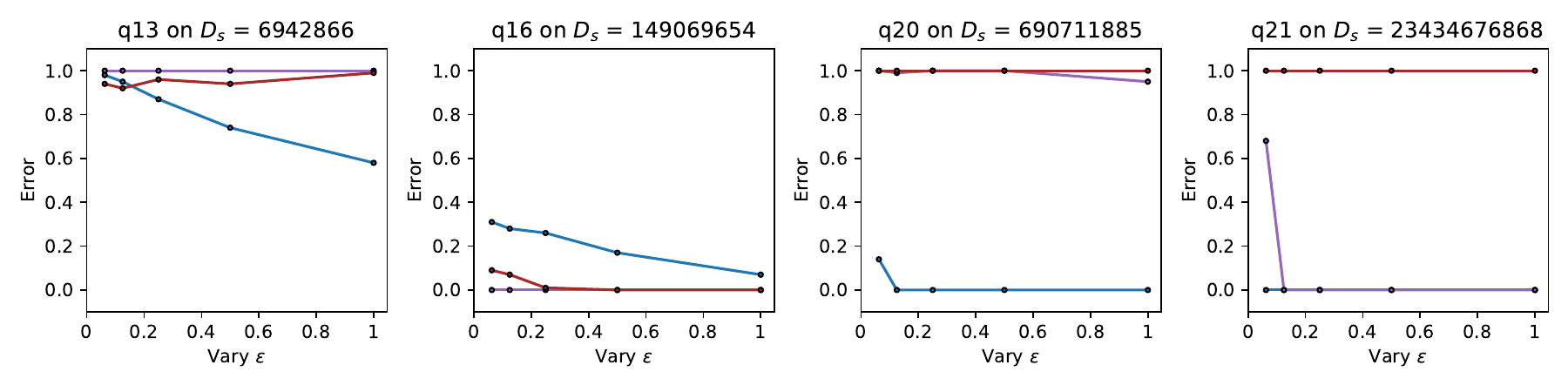}
  \caption{\reva{IPUMS-CPS data:} Error for \sumQuery\ queries from \lapsum\ (in blue), \rtwotsum\ (in purple) and \svtsum\ (in brown) as $\epsilon$ varies. 
  % \rtwotsum\ and \svtsum\ are not effective for chosen values of $\epsilon$. Cannot say based on effectiveness if \lapsum\ gives low error for $q_{13}$.
  }
  \label{fig:sum_pick4_vary_eps}
\end{figure*}

\begin{figure*}[ht]
  \includegraphics[width=0.9\textwidth]{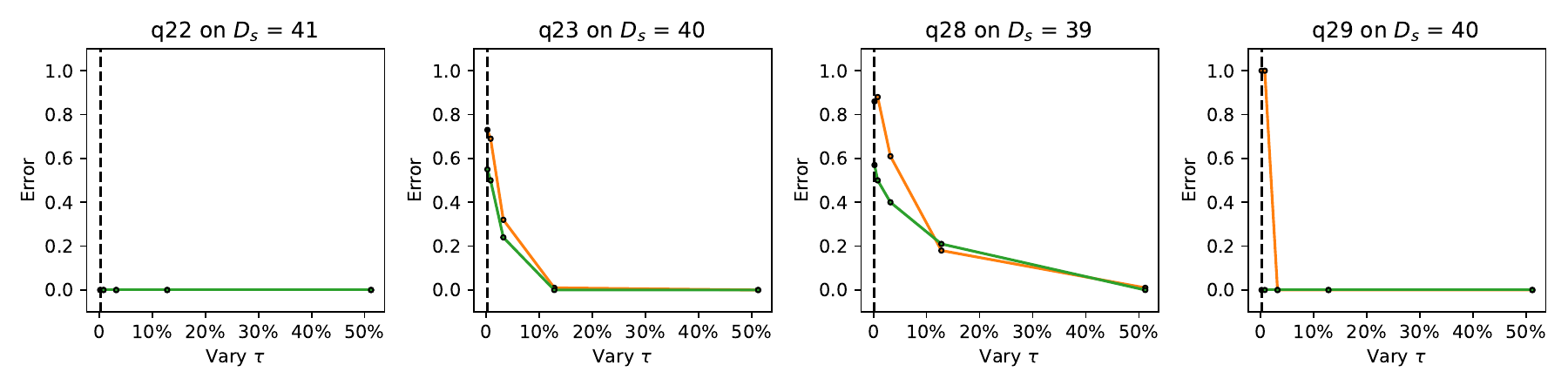}
  \caption{\reva{IPUMS-CPS data:} Error for \medianQuery\ queries from \emmed\ (in orange) and \histmed\ (in green) as $\tau$ varies.
  \revc{The dotted line marks the smallest $\tau$ value considered such that the query answer on the private data belongs in the interval \intvl.}
  }
  \label{fig:median_pick4_vary_tau}
\end{figure*}

\begin{figure*}[ht]
  \includegraphics[width=0.9\textwidth]{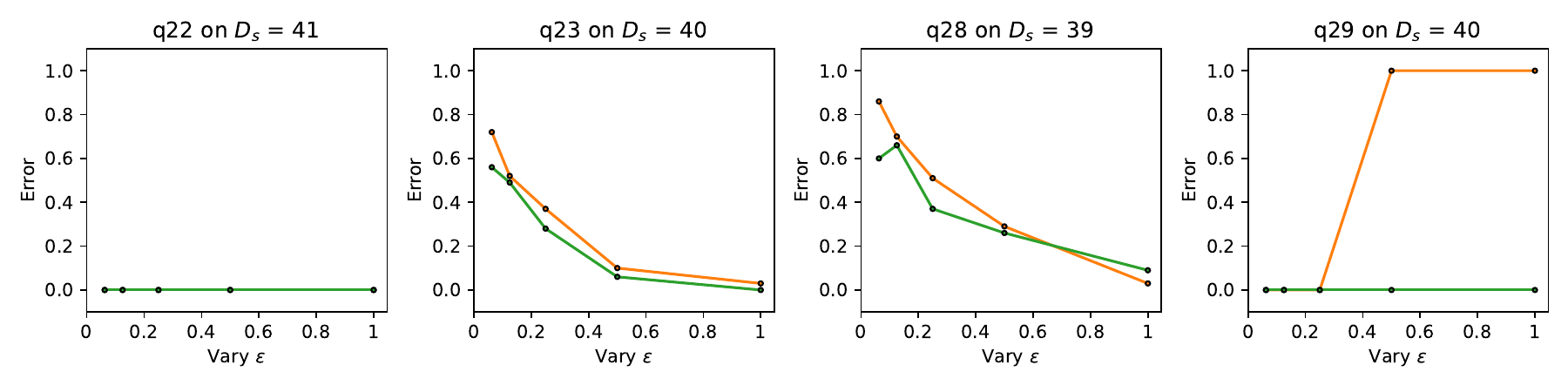}
  \caption{\reva{IPUMS-CPS data:} Error for \medianQuery\ queries from \emmed\ (in orange) and \histmed\ (in green) as $\epsilon$ varies.}
  \label{fig:median_pick4_vary_eps}
\end{figure*}

\begin{figure*}[ht]
  \includegraphics[width=0.9\textwidth]{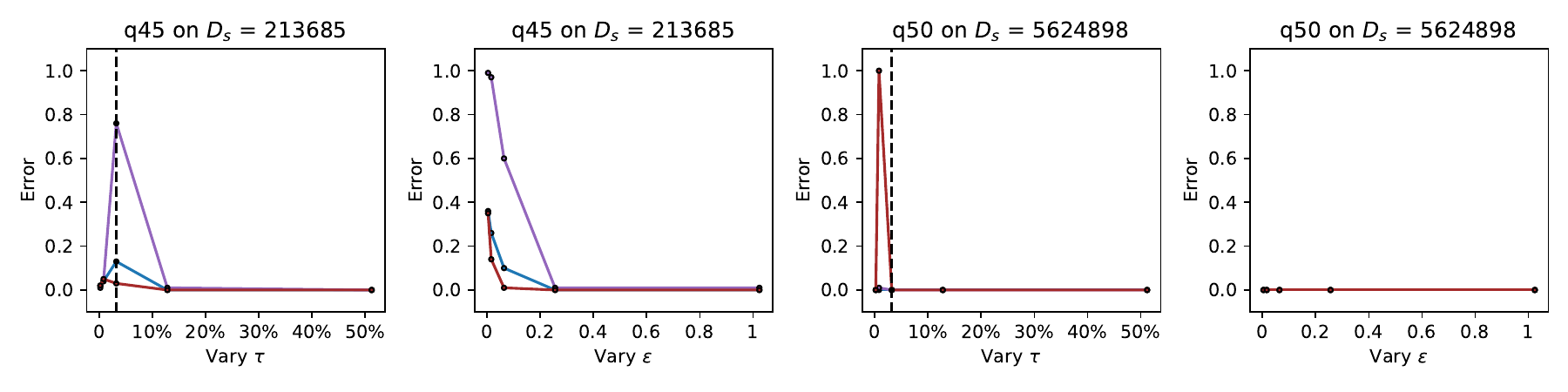}
  \caption{\reva{NYC Taxi data: Error for \sumQuery\ queries from \lapsum\ (in blue), \rtwotsum\ (in purple) and \svtsum\ (in brown) as $\epsilon$ and $\tau$ vary separately.} \revc{The dotted line marks the smallest $\tau$ value considered such that the query answer on the private data belongs in \intvl.}}
  \label{fig:sum_large_gap}
\end{figure*}

\smallskip
\paratitle{\countQuery\ queries} We present our analysis for $4$ queries: $q_1, q_3, q_5$ and $q_{12}$ (Figures~\ref{fig:count_pick4_vary_tau}-\ref{fig:count_pick4_vary_eps}).
% \texttt{WHERE} clauses are:
% ($1$) \texttt{$sex$ LIKE `Female' AND $race$ LIKE `White-American Indian-Asian' AND $workly$ LIKE `Yes'}, 
% ($2$) \texttt{$sex$ LIKE `Male' AND $educ$ LIKE `Doctorate degree' AND $marst$ LIKE `Separated'}, 
% ($3$) \texttt{$sex$ LIKE `Female' AND $educ$ LIKE `Master's degree' AND $marst$ LIKE `Married, spouse present'}, 
% ($4$) \texttt{$race$ LIKE `White' AND $marst$ LIKE `Married, spouse present' AND $citizen$ LIKE `Born in U.S'}.
% \ashwin{I think it is better to have a smaller tables in main paper with the queries you are focusing on so you can also show what q(D) and $q(D_s)$ are.}
% \ashwin{Graphs for q8 and q12 look the same. Why did you choose q8? Could you have chosen a medium size query like q5 instead?}
Consider the setting where $\tau$ varies. 
$q_1(D)$ equals $q_1(D_s)$ and error from both \lapcount\ and \emcount\ decreases when $\tau$ increases, as expected. \emcount\ gives smaller error. %, which follows from Section~\ref{subsec:compareCount}.
For $\tau \leq 3.2\%, q_3(D)\not\in \intvl$. At $\tau = 3.2\%$, $q_3(D)$ is closest to one of \intvl's endpoints, so we see the error from \lapcount\ increase (before it goes to $0$) because there is a higher probability of $q_3(D)$'s noisy estimate being in \intvl. Similarly for \emcount\ on $q_3$, where error starts decreasing for $\tau > 3.2\%$. \emcount's error on $q_3$ for $\tau \leq 3.2\%$ ranges from $0.39$ to $0.45$.
For $\tau \leq 0.8\%, q_5(D)\not\in\intvl$. At $\tau = 3.2\%$, $q_5(D)\in\intvl$ and is closest to one of \intvl's endpoints, so \lapcount's error has the same trend as that for $q_3$.
% so we see the error from \lapcount\ increase (before it goes to $0$) because there is a higher probability of $q_5(D)$'s noisy estimate not being in \intvl. 
In \emcount\ for $q_5$ and $\tau = 0.2\%$,  $c\cdot \exp(\epsilon u'(\cdot) \tau)$ (Definition~\ref{def:EM}) for $o = 0$ and $o = 1$ are closest, so $Pr[o = 1] = 0.27$ (after normalization). \emcount's error increases at $\tau = 3.2\%$ (before it goes to $0$) because $q_5(D)$ is close to an endpoint of \intvl.
\lapcount\ and \emcount\ give $0$ error for $q_{12}$ because the resulting $\tau$ values are such that $q(D)$ stays far from \intvl's endpoints. 

When $\epsilon$ is varied instead, the errors from both \lapcount\ and \emcount\ have a decreasing trend except for $q_3$ and $q_{12}$. As $\epsilon$ increases, \revb{the DP computations become} less noisy. 
\lapcount's error on $q_3$ increases initially because probability of $q_3(D)$'s noisy estimate being in \intvl\ increases despite the noise scale decreasing (before error goes to $0$). Similarly for \emcount\ on $q_3$.
\emcount's error on $q_{12}$ increases for $\epsilon > 0.5$ because $c\cdot \exp(\epsilon u'(\cdot) \tau)$ (Definition~\ref{def:EM}) for $o = 0$ and $o = 1$ is greater than $c\cdot \exp(7527 \cdot u'(\cdot))$, represented as infinity in Python. Either outcome is equally likely to be returned. 
% \ashwin{Why did you drop q3? This was one query that had an interesting trend that as you increase epsilon error does not decrease. This is an important finding that can talk about more in a discussion section and call out the potential for future work here. (The same trend is seen in q28 as well.}

\smallskip
\paratitle{\sumQuery\ queries} We present our analysis for $4$ queries: $q_{13}, q_{16}, q_{20}$ and $q_{21}$ (Figures~\ref{fig:sum_pick4_vary_tau}-\ref{fig:sum_pick4_vary_eps}). \DLS\ values are: $127764,$ $278011, 403353, 500000$.
% \texttt{WHERE} clauses are: 
% ($1$) \texttt{$sex$ LIKE `Female' AND $race$ LIKE `White-Black' AND $workly$ LIKE `No'}, 
% ($2$) \texttt{$race$ LIKE `White' AND $marst$ LIKE `Divorced' AND $citizen$ LIKE `Not a citizen'}, 
% ($3$) \texttt{$sex$ LIKE `Female' AND $educ$ LIKE `High school diploma or equivalent' AND $marst$ LIKE `Never married/single'}, 
% ($4$) \texttt{$race$ LIKE `White' AND $marst$ LIKE `Married, spouse present' AND $citizen$ LIKE `Born in U.S'}.
% \ashwin{I think the choice of queries is a bit unfortunate since in most cases LM seems to be better. I would have liked to see at least one query where the DS(q) is much much smaller than the GS(q) so that R2T or SVT perform better. But we may not have time to find such queries in this dataset for this deadline.}
Consider the setting where $\tau$ varies.
$q_{13}(D)\not\in \intvl$ (as defined in (\ref{eq:interval})) at $\tau = 0.2\%$ and the probability of \lapsum's noisy estimate being outside \intvl\ is high, so the error is low. \lapsum's error shoots up at $\tau = 0.8\%$ because $q_{13}(D)$ is now in \intvl, but the probability that $q_{13}(D)$'s noisy estimate is outside \intvl\ is high. Error decreases as $\tau$ increases further. Similarly for $q_{16}$ and $q_{20}$. \lapsum\ incurs an error of $0$ for $q_{21}$ in all cases because $q_{21}(D)$ is large and for the chosen $\tau$ values, it is far from \intvl's endpoints in comparison to the noise scale. 
\rtwotsum\ gives error close to $1$ whenever $q_{13}(D)\in \intvl$ because the noise in its estimate for $q_{13}(D)$ is large and this estimate falls outside \intvl. Similarly for $q_{16}$. \rtwotsum's error on $q_{20}$ is high when $q_{20}(D)\in \intvl$ but $\tau$ is small compared to the noise in \rtwotsum's estimate for $q_{20}(D)$. As $\tau$ increases, the error decreases. \rtwotsum's error on $q_{21}$ follows the same trend as that of $q_{20}$, except the error stays low because $q_{21}(D_s)$ is large and consequently $\tau$ values are large.
\svtsum's error on $q_{13}$ is low at $\tau = 0.2\%$ because $q_{13}(D) < l$ and the chance of any noisy truncated sum query exceeding its noisy threshold is low (lines~\ref{l:check_r} and \ref{l:check_l} in Algorithm~\ref{algo:err_svtsum}). For larger $\tau$, $q_{13}(D)\in \intvl$ and moves away from \intvl's endpoints, which increases the chance of the threshold in line~\ref{l:check_r} not being exceeded and in line~\ref{l:check_l} being exceeded. Similarly for $q_{16}$, $q_{20}$ and $q_{21}$, where the last two have large outputs on $D_s$ and so error decreases sooner.

When $\epsilon$ is varied, \lapsum's error on the $4$ queries shows a decreasing trend because the noise scale decreases as $\epsilon$ increases.
$q_{13}(D), q_{20}(D)$ and $q_{21}(D)$ are in the associated \intvl\ at (default) $\tau = 3.2\%$. For the first two, increasing $\epsilon$ does not change \rtwotsum's error 
\revb{
because the noise is $O(\log{(\GSbound)} \log{\log{(\GSbound)}}\cdot \DLS)$ (from (\ref{eq:R2T_thm5-1})), keeping $\tilde{q}(D)$ outside \intvl\ with high probability}. Since $q_{21}(D_s)$ is the largest, \rtwotsum's error decreases \revb{because the noise in $\tilde{q}(D)$ decreases as $\epsilon$ increases, while \intvl's width stays the same}.
For queries with answers on $D$ in \intvl, \svtsum's error does not change much as $\epsilon$ increases. 
\revb{
$q_{13}(D)$ and $q_{20}(D)$ are close to the left-endpoints of the respective intervals, and so \svtsum's error is high (line~\ref{l:check_l} in Algorithm~\ref{algo:err_svtsum}).}
In contrast, \svtsum's error on $q_{16}$ decreases as $\epsilon$ increases because the noise scale decreases and $q_{16}(D) < l$.

\revb{
Next, we analyze the impact of larger gaps in \DLS\ and \GSbound, for small \GSbound\ of $377$ (Figure~\ref{fig:sum_large_gap}). 
% The ratio is $0.22$ and $0.36$ for $q_{45}$ and $q_{50}$, respectively.
When $\tau$ varies, \lapsum\ and \rtwotsum\ give highest error for $q_{45}$ at $\tau = 3.2\%$ as $q_{45}(D)\in \intvl$ (closer to $r$), but the noise in their estimates for $q_{45}(D)$ is large while interval width is small, causing the estimates to not be in \intvl. 
At $\tau = 0.8\%$, $q_{50}(D)\not\in\intvl$ (closer to $r$). \lapsum\ and \rtwotsum\ give low error as \intvl's width is large and their estimates for $q_{50}(D)$ are in \intvl\ with high probability.
\svtsum\ gives low error for both queries, except for $q_{50}$ at $\tau = 0.8\%$ as the truncation threshold from 
% Algorithm~\ref{algo:upBound_DS_fromSVT} 
Algorithm $4$
is often smaller than \DLS, causing the check in line~\ref{l:check_r} to incorrectly fail.

At default $\tau = 3.2\%$, $q_{45}(D)$ and $q_{50}(D)$ are in their associated intervals (closer to $r$). For $q_{45}$, the interval is small, so all $3$ solutions give error at least $0.35$ when $\epsilon$ is small, with \rtwotsum\ giving the highest error due to large noise in its estimate. For $q_{50}$, the interval width is larger, so all $3$ solutions give $0$ error.
}

\smallskip
\paratitle{\medianQuery\ queries} We present our analysis for $4$ queries: $q_{22}, q_{23}, q_{28}$ and $q_{29}$ (Figures~\ref{fig:median_pick4_vary_tau}-\ref{fig:median_pick4_vary_eps}).
% \texttt{WHERE} clauses are: 
% ($1$) \texttt{$workly$ LIKE `Yes' AND $classwkr$ LIKE `Wage/salary, private' AND $educ$ LIKE `Bachelor's degree'}, 
% ($2$) \texttt{$sex$ LIKE `Male' AND $race$ LIKE `White-Black' AND $relate$ LIKE `Spouse'}, 
% ($3$) \texttt{$race$ LIKE `Asian only' AND $marst$ LIKE `Separated' AND $citizen$ LIKE `Born in U.S'}, 
% ($4$) \texttt{$sex$ LIKE `Male' AND $race$ LIKE `White' AND $classwkr$ LIKE `Wage/salary, private'}.
Consider the setting where $\tau$ varies.
% For $q_{22}$, \emmed\ picks $41$ as the noisy estimate with probability $1$ (Definition~\ref{def:EM}) and $q_{22}(D) = 41$. Thus, \emmed's error is $0$. 
% For $q_{23}$, \emmed\ picks $40$ with probability $0.324$ and $q_{23}(D) = 40$. As $\tau$ increases, the probability of picking a value inside \intvl\ (as defined in (\ref{eq:interval})) increases and thus the error decreases.
% We observe the same trend in error for $q_{28}$ as that of $q_{23}$ from \emmed. At $\tau = 12.8\%$, \emmed's error is $0.18$ whereas \histmed's error is $0.21$.
% \emmed\ gives error equal to $1$ on $q_{29}$ for small values of $\tau$ because the probability of picking $40$ is 0 but that of $41$ is $1$, so as $\tau$ increases and \intvl\ includes both $q_{29}(D) = 40$ and $41$, error decreases.
% \histmed's error has a decreasing trend in all cases because the noisy bin counts are far from the query's noisy estimate for majority, i.e., $\lceil n'/2 \rceil$ (lines~\ref{l:q1_return} and \ref{l:q2_return} in Algorithm~\ref{algo:err_1_q_med}).
Let us first look at \emmed. For $q_{22}$ and $q_{29}$, the support, i.e., tuples in $D$ that satisfy the predicates in the \texttt{WHERE} clause, is large. As a result, the probability distribution used to sample the noisy estimate for the median is more concentrated around the correct value. The opposite is true for $q_{23}$ and $q_{28}$. We see error decrease as $\tau$ increases.
\histmed's error has a decreasing trend in all cases because the noisy bin counts are far from the query's noisy estimate for majority, i.e., $\lceil n'/2 \rceil$ (lines~\ref{l:q1_return} and \ref{l:q2_return} in Algorithm~\ref{algo:err_1_q_med}).

When $\epsilon$ increases, the variance of the probability distribution used to sample the noisy estimate in \emmed\ decreases. \emmed's error has a decreasing trend except for $q_{29}$ because the scores are negatives numbers with large magnitudes because the support is large (Definition~\ref{def:EM}). 
As discussed above for \histmed, the noisy bin counts are far from the query's noisy estimate for majority. As $\epsilon$ increases, the error decreases because the noise scale decreases. 
% \ashwin{There is a lot of detail going into the actual numbers here. I think what is more important is discussing the general trends and any excpetions to the trend. For instance, you can explain in general error decreases with increase in tau and eps and then call out the exception to the trend in q28 when increase eps actually increases the error.}

\subsubsection{Performance analysis}
We compare average runtimes of per-query deciders for \countQuery, \sumQuery\ and \medianQuery\ queries in Table~\ref{tbl:runtime}. The reported average is per query and per $1$ run (out of the $100$ trials).
% \ashwin{Does LM count take 31 seconds for each of the 12 queries or for all 12 queries in total? If it is the latter, I would report the average per-query runtime and divide the values you have in the tables by the number of queries.}
% \ashwin{I thought SVT might run faster than R2T, but looks like it is slower? } \shweta{It spends some time to first find a bound for \DLS}
% \ashwin{You do not need a lot of text here. Just a table with average wall clock time for each alg across all queries would be sufficient. You can just say "We report the average run time of each of our algorithm in TableXX". If there is something interesting then you can add a line about that.}

\subsubsection{Discussion}\label{subsec:discussion}

We summarize our findings in this section.
The experiments suggest the following comparative trends between algorithms \revm{(Table~\ref{tbl:summary})}. 
In \countQuery\ queries, when $q(D) = q(D_s)$, \emcount's error is smaller than \lapcount's error for different $\epsilon$ and $\tau$ values. When $q(D)\in \intvl$ and $q(D)\not\in \intvl$, \lapcount\ is superior.
In \sumQuery\ queries, when $q(D)\in \intvl$ and $\tau$ varies, \lapsum\ was the best choice followed by \svtsum. Note that \DLS\ plays an important role here. When $q(D)\not\in \intvl$, \rtwotsum\ and \svtsum\ were the better choices.
In \medianQuery\ queries, \histmed\ generally gives smaller error than \emmed\ no matter the relationship between $q(D)$ and \intvl.

There is not a clear pattern for when the error is high except for the condition when $q(D)$ is close to one of the endpoints of the interval, or when the downward local sensitivity is not far from the global sensitivity for the given \sumQuery\ query. %\ashwin{is this right?}

Typically, we expect error of DP algorithms to decrease with higher $\epsilon$. But, we see the reverse in some experiments (see \lapcount\ for $q_3$, \emcount\ for $q_{12}$, \svtsum\ and \rtwotsum\ for queries with answer on $D$ in the interval, and \emmed\ for $q_{29}$). This is because the error function for the per-query deciders is not monotonic in $\epsilon$. This may suggest that there exist smarter ways to design algorithms that only use a portion of the overall budget available to get better accuracy.

% \ashwin{Maybe something about effectiveness?}
% \ashwin{I think you should drop the effectiveness discussion in the figure captions -- it is hard to understand there. You may consider adding a subsection on effectiveness or just bring it up in the discussion for count and sum queries and say something short like: whenever LM count was effective, EM count was also effective, and a summary for sum queries}
\lapcount\ and \emcount\ were not effective for queries at smaller $\tau$, except when $q(D_s)$ and $\tau$ are large, in which case we do not know if either approach is effective. They were effective for the same settings. 
Similarly for \lapsum. \rtwotsum\ and \svtsum\ were not effective for queries with large answer on $D$ at small $\tau$.
\revb{
In general for all solutions, the accuracy can be improved with a larger privacy budget $\epsilon$, or a larger distance bound $\tau$ between $q(D)$ and $q(D_s)$, which may not always be feasible. 
}

\section{related work}\label{sec:related}
% Literature review - show problem is open (need not be detailed)
% \sr{it should be a bit longer -- you can mention the standard DP-mechanisms that you have used in the paper but no need to go over them, or say that you have used several existing DP mechanisms for count, sum, median from the literature (cite) for \oursys. }
We have used several existing DP mechanisms for count, sum, and median from the literature \cite{DP_book, 10.1145/3514221.3517844, 10.14778/3055330.3055331, DBLP:journals/corr/abs-1103-5170} for \oursys.
To the best of our knowledge, most existing works on SDGs do not give per-query error bounds to the user.
\reva{ 
% AIM \cite{DBLP:journals/corr/abs-2201-12677} is a novel differentially private SDG that takes a workload of weighted marginal queries as input and gives high probability per-query error bounds for them. 
AIM \cite{DBLP:journals/corr/abs-2201-12677} is a novel differentially private SDG that generates $D_s$ while minimizing average error over all the input marginal queries, and only gives probabilistic upper-bounds on the error for marginal queries in the \textit{downward closure} of the input workload. It}
% A marginal query is a histogram over the full domain of a set of attributes. 
follows the {\em select-measure-generate} paradigm \cite{10.1145/1807085.1807104, 10.14778/3231751.3231769, DBLP:journals/corr/abs-2106-12118}: ($1$) selects a set of queries, ($2$) privately measures chosen queries, and ($3$) generates synthetic data from the noisy measurements. 
% Prev version
% It achieves low error from innovations in the {\em select} stage, where it chooses the next best marginal from the {\em downward closure} of the input marginals. It is these queries for which AIM gives error guarantees.
% In contrast, we study the problem for \countQuery, \sumQuery\ and \medianQuery\ queries that may not be represented well in the synthetic data. We do not make assumptions about the SDG and do not require that the given query be well approximated by the SDG.
\reva{
However, it differs from our model wherein we are given $D_s$ from some black-box SDG and the goal is to decide if the distance between the given (count, sum, or median) query $q$'s output on $D$ and $D_s$ is less than the user-provided threshold $\tau$.
% We study the problem for queries that may not be represented well in the synthetic data. 
We do not make assumptions about the SDG and do not require that the given query should be well approximated by the SDG (or $D_s$).

We also found that the \textit{parametric bootstrap approach} \cite{param_boot} does not work well because the underlying assumption about the difference between bootstrap estimates and $q(D_s)$ being representative of the difference between $q(D_s)$ and $q(D)$ does not hold when the SDG (e.g. PrivBayes \cite{10.1145/3134428}) uses techniques like post-processing.
}

% \shweta{Question from PPAI Q/A: If our approach tells the user that the error is larger than $\tau$, what should they do? Is the synthetic data bad?}
% \shweta{Can address this in the paper by saying that every SDG method gives low error on a subset of queries. Otherwise you can reconstruct the underlying data. This approach will help the user figure out which query answers can be trusted and which cannot.}

% \shweta{Add (?) bootstrap may not work; example. If budget for data gen is small, it (intuitively) impacts the width of CI making it harder to optimize for the width}

\begin{table}[t] %[!ht]
    \centering \scriptsize
    \begin{tabular}{| c | c | c |}
        \rowcolor[HTML]{C0C0C0} 
        \hline {\bf } & {\bf \lapcount} & {\bf \emcount}\\
        \hline \cellcolor[HTML]{D7D7D7} $Time(s)$ & $0.312$ & $0.321$\\
        \hline
    \end{tabular}
%    \caption{Average runtimes for 12 \countQuery\ queries.}\label{tbl:exp_count}
%\end{table}

%\begin{table}[t]
%    \centering \scriptsize
    \begin{tabular}{| c | c | c | c |}
        \rowcolor[HTML]{C0C0C0} 
        \hline {\bf }  & {\bf \lapsum} & {\bf \rtwotsum} & {\bf \svtsum}\\
        \hline \cellcolor[HTML]{D7D7D7} $Time(s)$ & $0.450$ & $4.022$ & $9.131$ \\
        \hline
    \end{tabular}
%    \caption{Average runtimes for 9 \sumQuery\ queries.}\label{tbl:exp_sum}
%\end{table}

%\begin{table}[t]
%    \centering \scriptsize
    \begin{tabular}{| c | c | c | }
        \rowcolor[HTML]{C0C0C0} 
        \hline {\bf } & {\bf \emmed} & {\bf \histmed}\\
        \hline \cellcolor[HTML]{D7D7D7} $Time(s)$ & $ 0.533$ & $1.043$\\
        \hline
    \end{tabular}
%    \caption{Average runtimes for 9 \medianQuery\ queries.}\label{tbl:exp_med}
    % \caption{Avg. runtimes for 12 \countQuery, 9 \sumQuery, 9 \medianQuery\ queries.}
    \caption{\reva{IPUMS-CPS data:} Average runtimes.}
    \label{tbl:runtime}
\end{table}

\begin{table}[t] %[!ht]
    \centering \scriptsize
    \revm{
        \begin{tabular}{| c | c | c | c |}
            \rowcolor[HTML]{C0C0C0} 
            \hline {\bf Query} & {\bf Solution} & {\bf $\tau_{min}^{\alg, \delta}$ (upper bound)} & {\bf Conclusions}\\
    
            \hline \cellcolor[HTML]{D7D7D7} & \begin{tabular}{@{}c@{}} \lapcount\\(plug-in) \end{tabular} & $\frac{1}{\epsilon} \ln{\frac{1}{2\delta}}$ & \\
            \cline{2-3} \multirow{-3}{*}{\cellcolor[HTML]{D7D7D7} \countQuery} & \begin{tabular}{@{}c@{}} \emcount\\(direct) \end{tabular} & $\frac{1}{\epsilon} \ln{\frac{1 - \delta}{\delta}}$ & \multirow{-3}{*}{\begin{tabular}{@{}c@{}} \lapcount\ is the better choice,\\for $\delta < \frac{1}{2}$ unless $q(D)=q(D_s)$. \end{tabular}}\\
    
            \hline \cellcolor[HTML]{D7D7D7} & \begin{tabular}{@{}c@{}} \lapsum\\(plug-in) \end{tabular} & $\frac{\GSbound}{\epsilon} \ln{\frac{1}{2\delta}}$ & \\
            \cline{2-3} \cellcolor[HTML]{D7D7D7} & \begin{tabular}{@{}c@{}} \rtwotsum\\(plug-in) \end{tabular} & \begin{tabular}{@{}c@{}} $4\log(\GSbound) \cdot$\\$\ln{\left( \frac{\log(\GSbound)}{\delta} \right)} \frac{\DLS}{\epsilon}$ \end{tabular} & \\
            \cline{2-3} \multirow{-6}{*}{\cellcolor[HTML]{D7D7D7} \sumQuery} & \begin{tabular}{@{}c@{}} \svtsum\\(direct) \end{tabular} & -- & \multirow{-6}{*}{\begin{tabular}{@{}c@{}} If approx. \DLS\ value is known\\to the user, and \DLS\ is much\\smaller than \GSbound, then choose\\\svtsum\ or \rtwotsum. Otherwise,\\choose \lapsum\ unless $\epsilon$ is small. \end{tabular}}\\
    
            \hline \cellcolor[HTML]{D7D7D7} & \begin{tabular}{@{}c@{}} \emmed\\(plug-in) \end{tabular} & -- & \\
            \cline{2-3} \multirow{-3}{*}{\cellcolor[HTML]{D7D7D7} \medianQuery} & \begin{tabular}{@{}c@{}} \histmed\\(direct) \end{tabular} & -- & \multirow{-3}{*}{\begin{tabular}{@{}c@{}} \histmed\ empirically gives\\smaller error than \emmed. \end{tabular}}\\
            
            \hline
        \end{tabular}
    }
    \caption{\revm{Summary of proposed solutions and our recommendations based on theoretical and empirical results. Those marked by (plug-in) are based on Algorithm~\ref{algo:plug_in}, whereas the rest solve the problem without plugging-in an estimate for $q(D)$. Some theoretical upper bounds remain open problems. }
    \label{tbl:summary}}
\end{table}
% Note. $\tau_{min}$ gives the smallest interval width at which conditions in Def. 3.4 are met: ($1$) if $q(D) = q(D_s), Pr[o = 1] \geq 1 - \delta$, ($2$) if $q(D) \not\in (q(D_s) - 2\tau, q(D_s) + 2\tau), Pr[o = 0] \geq 1 - \delta$.
\section{Conclusions and Future Work}\label{sec:conclusions}
We have presented the problem of measuring the per-query distance between the output on private data and synthetic data, and detailed our error analysis for \countQuery, \medianQuery\ and \sumQuery\ queries.
Our proposed solutions fall in two classes: ($1$) use a DP algorithm to answer the query and check if the noisy answer is close to the answer on the synthetic data, and ($2$) design specialized algorithm. %for the problem.
In addition to analyzing the error, we also introduce the notion of effectiveness of a per-query decider and derive upper bounds on the effectiveness thresholds of solutions for \countQuery\ and \sumQuery\ queries (except \svtsum). 
Deriving such bounds for \medianQuery\ query is %an intriguing subject for 
future work. 
We find that some mechanisms work better for smaller $\tau$. %than others. 
\revm{Extending our work to other useful queries involving other aggregate functions, joins, subqueries, and group-by is %an interesting direction for 
future work.}
\reva{Designing baselines and benchmarks for the problem in this work is future work, and may be of interest to the synthetic data %generation 
and data privacy communities.}

\begin{acks}
We thank the anonymous reviewers for their helpful comments. 
This work was supported in part by NSF awards IIS-2147061,  IIS-2008107, IIS-1703431, and IIS-1552538. 

%We also thank the anonymous reviewe for contributing [...].
\end{acks}

\clearpage

\bibliographystyle{ACM-Reference-Format}
\bibliography{bibtex}

\clearpage
\appendix

\section{Proofs from Section~\ref{sec:count}}

\subsection{Proof of Proposition~\ref{prop:lapcount_err}}\label{subsec:appendix_lapcount_err}
\begin{proof}
    For ($1$): When we get an error due to $q(D)\leq l$ and $o = 1$, the if condition (line~\ref{l:basic_if}) in \BasicDecider\ (Algorithm \ref{algo:plug_in}) is satisfied, i.e., $l < q(D) + \nu_q < r$. 
    Re-arranging the terms gives $0\leq l - q(D) < \nu_q < r - q(D)$. Now, instantiating (\ref{eq:lap-geq}) with $t = (l - q(D))\epsilon$ and then with $t = (r - q(D))\epsilon$ gives:
    \begin{align}
        Pr[\nu_q > l - q(D)] &\leq Pr[\nu_q \geq l - q(D)] = \frac{1}{2}e^{-(l - q(D))\epsilon}\label{eq:4pt2a_1}\\
        Pr[\nu_q \geq r - q(D)] &= \frac{1}{2}e^{-(r - q(D))\epsilon}\label{eq:4pt2a_2}
    \end{align}

    Where the first inequality in (\ref{eq:4pt2a_1}) stems from the property of DP that every output has a positive probability.
    Subtracting (\ref{eq:4pt2a_2}) from (\ref{eq:4pt2a_1}), $Pr[l - q(D) < \nu_q < r - q(D)] \leq \frac{1}{2} \left( e^{-(l - q(D))\epsilon} - e^{-(r - q(D))\epsilon} \right)$.
    
    ($2$): When we get an error due to $q(D)\geq r$ and $o = 1$, the if condition (line~\ref{l:basic_if}) in \BasicDecider\ is satisfied, i.e., $l - q(D) < \nu_q < r - q(D)\leq 0$ (after re-arranging the terms). Now,
    \begin{align}
        Pr[\nu_q > l - q(D)] &= 1 - Pr[\nu_q \leq -(q(D) - l)]\nonumber\\
        &= 1 - \frac{1}{2} e^{-(q(D) - l)\epsilon}\label{eq:4pt2b_1}
    \end{align}
    which follows from (\ref{eq:lap-leq}) with $t = (q(D) - l)\epsilon$. 
    
    \begin{align}
        Pr[\nu_q \geq r - q(D)] &= 1 - Pr[\nu_q < r - q(D)] &\nonumber\\
        &\geq 1 - Pr[\nu_q \leq r - q(D)]\nonumber\\
        &= 1 - \frac{1}{2} e^{-(q(D) - r)\epsilon}\label{eq:4pt2b_2}
    \end{align}
    which follows from (\ref{eq:lap-leq}) with $t = (q(D) - r)\epsilon$.
    
    Note that $Pr[l - q(D) < \nu_q < r - q(D)] = Pr[\nu_q > l - q(D)] - Pr[\nu_q \geq r - q(D)]$, and from (\ref{eq:4pt2b_1}) and (\ref{eq:4pt2b_2}), this expression is less than or equal to $\frac{1}{2} \left( e^{-(q(D) - r)\epsilon} - e^{-(q(D) - l)\epsilon} \right)$.
    
    ($3$): When we get an error due to $l < q(D) < r$ and $o = 0$, the else condition (line~\ref{l:basic_else}) in \BasicDecider\ is satisfied.
    Here, error is given by the probability of the union of two disjoint events: (i) $\nu_q \leq l - q(D)$, which is calculated by instantiating (\ref{eq:lap-leq}) with $t = (q(D) - l)\epsilon$, and (ii) $\nu_q \geq r - q(D)$, which is calculated by instantiating (\ref{eq:lap-geq}) with $t = (r - q(D))\epsilon$. 
    Error equals $\frac{1}{2} e^{-(q(D) - l)\epsilon} + \frac{1}{2} e^{-(r - q(D))\epsilon}$.
\end{proof}

\subsection{Proof of Proposition~\ref{prop:tau_min_lapcount}} \label{subsec:appendix_tau_min_lapcount}
\begin{proof}
    We analyze the two cases in Definition~\ref{def:effectiveness_metric}.
    \begin{enumerate}
        \item Let $q(D) = q(D_s)$. \lapcount\ returns $o = 1$ when the if condition (line~\ref{l:basic_if}) in \BasicDecider\ is satisfied. 
        Therefore, 
        \begin{align}
            Pr[o = 1] &= Pr[|q(D) + \nu_q - q(D_s)| < \tau]\nonumber\\
            &= Pr[|\nu_q| < \tau] ~~\quad\text{(since $q(D) = q(D_s)$)}\nonumber\\
            &= 1 - Pr[|\nu_q| \geq \tau]\nonumber\\
            &= 1 - e^{-\tau \epsilon} ~~\quad\text{(from (\ref{eq:lap-mod}) with $t = \tau\epsilon$)}\nonumber\\
            &\geq 1 - \delta ~~\quad\text{(when $\tau \geq \frac{1}{\epsilon} \ln{\frac{1}{\delta}} > \frac{1}{\epsilon} \ln{\frac{1}{2\delta}}$)} \label{eq:lmcount-th-1}
        \end{align}
        
        \item Let $q(D)\not\in (q(D_s) - 2\tau, q(D_s) + 2\tau)$. Without loss of generality, suppose $q(D)\leq q(D_s) - 2\tau$, so
        \begin{align}
            \tau \leq q(D_s) - \tau - q(D)\label{eq:4pt3b_tau}
        \end{align}
        
       Recall that \lapcount\ returns $o = 0$ when the else condition (line~\ref{l:basic_else}) in \BasicDecider\ is satisfied. 
       Therefore, 
        \begin{align}
            Pr[o = 0] &= Pr[|q(D) + \nu_q - q(D_s)| \geq \tau] \nonumber\\
            &\geq Pr[q(D) + \nu_q - q(D_s) \leq - \tau]\nonumber\\
            &= Pr[\nu_q \leq q(D_s) - \tau - q(D)]\nonumber\\
            &= 1 - Pr[\nu_q > q(D_s) - \tau - q(D)]\nonumber\\
            &> 1 - Pr[\nu_q \geq q(D_s) - \tau - q(D)]\nonumber\\
            &= 1 - \frac{1}{2}e^{-(q(D_s) - \tau - q(D))\epsilon} ~~\quad\text{(from (\ref{eq:lap-geq}) with}\nonumber\\ 
            &\hspace{25mm}\text{$t = (q(D_s) -\tau - q(D))\epsilon$)}\nonumber\\
            &\geq 1 - \frac{1}{2}e^{-\tau \epsilon} ~~\quad\text{(from (\ref{eq:4pt3b_tau}))}\nonumber\\
            &\geq 1 - \delta ~~\quad\text{(when $\tau\geq \frac{1}{\epsilon} \ln{\frac{1}{2\delta}}$)} \label{eq:lmcount-th-2}
        \end{align}
    \end{enumerate}
    From (\ref{eq:lmcount-th-1}) and (\ref{eq:lmcount-th-2}), the upper bound of $\tau_{min}^{\lapcount, \delta}$ follows.
\end{proof}

\subsection{Proof of Proposition~\ref{prop:bin_u_sens}}\label{subsec:appendix_bin_u_sens}
\begin{proof}
    Note that since $u$ is binary, $\Delta u \leq 1$. We show that $\Delta u = 1$. 
    The global sensitivity equals the worst case change in the output of $u$ on neighboring databases $D\approx D'$ for $o\in \mathcal{R} = \{0, 1\}$.
    Consider an interval $\intvl$ such that $q(D) = p\in \intvl$ and $q(D') = p + 1\not\in \intvl$ (as $q$ is a count query and $\Delta q = 1$). $u(D, D_s, q, \tau, o = 1) = 1$, whereas $u(D', D_s, q, \tau, o = 1) = 0$. Instead if $o = 0$, the scores differ by $1$ again. Similarly for $q(D)\not\in \intvl$ and $q(D')\in \intvl$.
    Thus, $\Delta u = 1$.
\end{proof}

\subsection{Proof of Proposition~\ref{prop:uprime_sens}}\label{subsec:appendix_uprime_sens}
\begin{proof}
    The global sensitivity equals the worst case change in the output of $u'$ on neighboring databases $D\approx D'$ for $o\in \mathcal{R} = \{0, 1\}$. 
    Let $q(D) = p$ and $q(D') = p + 1$, for an integer $p$ (as $\Delta q = 1$ for count query $q$).
    We break down the analysis for $o = 1$ as follows:%, where $l$ and $r$ are as defined in (\ref{eq:interval}):
    \begin{enumerate}
        \item If $q(D)$ and $q(D')$ both satisfy (\ref{eq:uprime_1}), then the scores from $u'$ are equal.

        \item If $q(D)$ and $q(D')$ both satisfy (\ref{eq:uprime_2}), then the scores from $u'$ differ by $1/2\tau$ because:
        \begin{align*}
            |\frac{q(D) - (l - \tau)}{2\tau} - \frac{q(D') - (l - \tau)}{2\tau}|
            = |\frac{q(D) - q(D')}{2\tau}|
        \end{align*}
        which equals $1/2\tau$ because $|q(D) - q(D')| = |p - (p + 1)| = 1$.

        \item If $q(D)$ and $q(D')$ both satisfy (\ref{eq:uprime_3}), then the scores from $u'$ differ by $1/2\tau$ because:
        \begin{align*}
            |1 - \frac{q(D) - q(D_s)}{2\tau} - \left( 1 - \frac{q(D') - q(D_s)}{2\tau} \right)|
            = |\frac{q(D) - q(D')}{2\tau}|
        \end{align*}
        which equals $1/2\tau$ because $|q(D) - q(D')| = |p - (p + 1)| = 1$.

        \item If $q(D)$ satisfies (\ref{eq:uprime_1}) and $q(D')$ satisfies (\ref{eq:uprime_2}), then $q(D)\leq l - \tau$ because $q(D)$ and $q(D')$ differ by $1$ (given), and so the scores from $u'$ differ by less than $1/2\tau$. 
        Similarly for when: ($i$) $q(D)$ satisfies (\ref{eq:uprime_2}) and $q(D')$ satisfies (\ref{eq:uprime_3}), and ($ii$) $q(D)$ satisfies (\ref{eq:uprime_3}) and $q(D')$ satisfies (\ref{eq:uprime_1}). 
    \end{enumerate}
    The other case, i.e., $o = 0$, is proven similarly. Hence, the worst case change is $1/2\tau$ and the global sensitivity of $u'$ equals $1/2\tau$.
\end{proof}

\subsection{Proof of Proposition~\ref{prop:tau_min_emcount}} \label{subsec:appendix_tau_min_emcount}
\begin{proof}
    We analyze the two cases in Definition~\ref{def:effectiveness_metric}.
    $c$ is a positive constant (Definition~\ref{def:EM}).
    Proposition~\ref{prop:uprime_sens} gives $\Delta u' = 1/2\tau$.
    \begin{enumerate}
        \item Let $q(D) = q(D_s)$. Here, the score assigned by $u'$ to $o = 0$ is $0$ and to $o = 1$ is $1$ (by (\ref{eq:uprime_2})). From Definition~\ref{def:EM}, we get
        \begin{align}
            &Pr[o = 0] = c\cdot e^{\epsilon \times 0 \times \tau} \text{ and } Pr[o = 1] = c\cdot e^{\epsilon \times 1 \times \tau}\nonumber\\
            \implies &Pr[o = 1] = \frac{e^{\epsilon \tau}}{1 + e^{\epsilon \tau}} \geq 1 - \delta ~~\quad\text{(when $\tau\geq \frac{1}{\epsilon} \ln{\frac{1 - \delta}{\delta}}$)}\label{eq:tau-em-1}
        \end{align}
        
        \item Let $q(D)\not\in (q(D_s) - 2\tau, q(D_s) + 2\tau)$. Without loss of generality, let $q(D)\leq q(D_s) - 2\tau = l - \tau$. 
        Here, the score assigned by $u'$ to $o = 0$ is $1$ and to $o = 1$ is $0$ (by (\ref{eq:uprime_1})). From Definition~\ref{def:EM}, we get
        \begin{align}
            &Pr[o = o] = c\cdot e^{\epsilon \times 1 \times \tau} \text{ and } Pr[o = 1] = c\cdot e^{\epsilon \times 0 \times \tau}\nonumber\\
            \implies &Pr[o = 0] = \frac{e^{\epsilon \tau}}{1 + e^{\epsilon \tau}} \geq 1 - \delta ~~\quad\text{(when $\tau\geq \frac{1}{\epsilon} \ln{\frac{1 - \delta}{\delta}}$)}\label{eq:tau-em-0}
        \end{align}
    \end{enumerate}
    From (\ref{eq:tau-em-1}) and (\ref{eq:tau-em-0}), the upper bound of $\tau_{min}^{\emcount, \delta}$ follows.
\end{proof}

\section{Proofs from Section~\ref{sec:sum}}

\subsection{Proof of Proposition~\ref{prop:tau_min_lapsum}} \label{subsec:appendix_tau_min_lapsum}
\begin{proof}
    Since $\nu_q\sim Lap(\GSbound/\epsilon)$, bounds (\ref{eq:lap-geq})-(\ref{eq:lap-mod}) become 
    \begin{align}
        Pr \left[ \nu_q \geq t\cdot \frac{\GSbound}{\epsilon} \right] = \frac{e^{-t}}{2}\label{eq:lap-sum-geq}\\
        Pr \left[ \nu_q \leq -t\cdot \frac{\GSbound}{\epsilon} \right] = \frac{e^{-t}}{2}\label{eq:lap-sum-leq}\\
        Pr \left[ |\nu_q| \geq t\cdot \frac{\GSbound}{\epsilon} \right] = e^{-t}\label{eq:lap-sum-mod}
    \end{align}
    
    We analyze the two cases in Definition~\ref{def:effectiveness_metric}.
    \begin{enumerate}
        \item Let $q(D) = q(D_s)$. \lapsum\ returns $o = 1$ when the if condition (line~\ref{l:basic_if}) in \BasicDecider\ is satisfied. 
        Therefore, 
        \begin{align}
            Pr[o = 1] &= Pr[|q(D) + \nu_q - q(D_s)| < \tau]\nonumber\\
            &= Pr[|\nu_q| < \tau] ~~\quad\text{(since $q(D) = q(D_s)$)}\nonumber\\
            &= 1 - Pr[|\nu_q| \geq \tau]\nonumber\\
            &= 1 - e^{-\tau \epsilon/\GSbound} ~~\quad\text{(from (\ref{eq:lap-sum-mod}) with $t = \tau\epsilon/\GSbound$)}\nonumber\\
            &\geq 1 - \delta ~~\quad\text{(when $\tau \geq \frac{\GSbound}{\epsilon} \ln{\frac{1}{\delta}} > \frac{\GSbound}{\epsilon} \ln{\frac{1}{2\delta}}$)} \label{eq:lmsum-th-1}
        \end{align}
        
        \item Let $q(D)\not\in (q(D_s) - 2\tau, q(D_s) + 2\tau)$. Without loss of generality, suppose $q(D)\leq q(D_s) - 2\tau$, so
        \begin{align}
            \tau \leq q(D_s) - \tau - q(D)\label{eq:5pt2b_tau}
        \end{align}
        
       Recall that \lapsum\ returns $o = 0$ when the else condition (line~\ref{l:basic_else}) in \BasicDecider\ is satisfied. 
       Therefore, 
        \begin{align}
            Pr[o = 0] &= Pr[|q(D) + \nu_q - q(D_s)| \geq \tau] \nonumber\\
            &\geq Pr[q(D) + \nu_q - q(D_s) \leq - \tau]\nonumber\\
            &= Pr[\nu_q \leq q(D_s) - \tau - q(D)]\nonumber\\
            &= 1 - Pr[\nu_q > q(D_s) - \tau - q(D)]\nonumber\\
            &> 1 - Pr[\nu_q \geq q(D_s) - \tau - q(D)]\nonumber\\
            &= 1 - \frac{1}{2}e^{-(q(D_s) - \tau - q(D))\cdot\frac{\epsilon}{\GSbound}} ~~\quad\text{(from (\ref{eq:lap-sum-geq}) with}\nonumber\\ 
            &\hspace{25mm}\text{$t = (q(D_s) -\tau - q(D))\cdot \frac{\epsilon}{\GSbound}$)}\nonumber\\
            &\geq 1 - \frac{1}{2}e^{-\tau \epsilon/\GSbound} ~~\quad\text{(from (\ref{eq:5pt2b_tau}))}\nonumber\\
            &\geq 1 - \delta ~~\quad\text{(when $\tau\geq \frac{\GSbound}{\epsilon} \ln{\frac{1}{2\delta}}$)} \label{eq:lmsum-th-2}
        \end{align}
    \end{enumerate}
    From (\ref{eq:lmsum-th-1}) and (\ref{eq:lmsum-th-2}), the upper bound of $\tau_{min}^{\lapsum, \delta}$ follows.
\end{proof}

\subsection{R2T: Validity of Truncation Function and Known Confidence Bound} \label{subsec:appendix_R2T_truncSatProperties_Thm5pt1}
We are given a \sumQuery\ query $q$ on (aggregate) attribute $A_i$.
Recall from Section~\ref{subsec:rtwotsum} that we use a simple truncation function which removes tuples in $D$ with $A_i$ value greater than $t_j = 2^j$, for $j = 1$ to $\log(\GSbound)$.
$q(D, t_j)$ denotes the output of running $q$ on the corresponding truncated database.
\cite{10.1145/3514221.3517844} lays down three properties that a {\em valid} truncation function must satisfy.
% The validity of the truncation function is determined based on whether or not it satisfies certain properties as stated in \cite{10.1145/3514221.3517844}. 
We show next that the simple truncation function meets the criteria.

\begin{proposition}\label{prop:R2T-trunc-properties}
    The simple truncation function satisfies the properties required by R2T (as stated below) and is valid:
    \begin{enumerate}
        % \item For any truncation threshold $t_j$, the global sensitivity of $q$ on the truncated database is less than or equal to $t_j$.
        \item For any $t_j$, the global sensitivity of $q(\cdot, t_j)\leq t_j$.
        
        % \item For any truncation threshold $t_j$, the output of $q$ on the truncated database cannot exceed $q(D)$.
        \item For any $t_j$, $q(D, t_j)\leq q(D)$.
        
        % \item For any $D$, there exists an integral truncation threshold $t^*\leq \GSbound$ such that for any $t_j\geq t^*$, the output of $q$ on the truncated database equals $q(D)$.
        \item For any $D$, $\exists$ non-negative integer $t^*(D)\leq \GSbound$ such that $\forall t_j\geq t^*(D), q(D, t_j) = q(D)$.
    \end{enumerate}
\end{proposition}
\begin{proof}
    The simple truncation function is valid because:
    \begin{enumerate}
        \item For any $t_j$, the worst-case change in the output of the truncated sum over databases $D\approx D'$ occurs when we add or remove a tuple with $A_i = t_j$. Larger values get discarded.
    
        \item For any $t_j$, the result of applying $q$ on the truncated database cannot exceed $q(D)$ because we remove at least $0$ tuples in the truncated database.
    
        \item $t^*(D) = \GSbound$ satisfies the condition. \DLS\ (Definition~\ref{def:DLS}) is another such value and is a private quantity.
    \end{enumerate}
\end{proof}

Theorem $5.1$ in \cite{10.1145/3514221.3517844} states that:
\begin{footnotesize}
    \begin{align*}
        Pr\left[ q(D) \geq \tilde{q}(D) \geq q(D) - 4\log(\GSbound) \ln{\left( \frac{\log(\GSbound)}{\beta} \right)} \frac{t^*}{\epsilon} \right]\geq 1 - \beta
    \end{align*}
\end{footnotesize}
where $\tilde{q}(D)$ is R2T's final estimate for $q(D)$, $t^*$ is as given in Proposition~\ref{prop:R2T-trunc-properties} and $0 < \beta < 1$. We apply this with $t^* = \DLS$ in (\ref{eq:R2T_thm5-1}).

\subsection{Proof of Proposition~\ref{prop:tau_min_rtwotsum}} \label{subsec:appendix_tau_min_rtwotsum}
\begin{proof}
    We analyze the two cases in Definition~\ref{def:effectiveness_metric}.
    \begin{enumerate}
        \item Let $q(D) = q(D_s)$. \rtwotsum\ returns $o = 1$ when the if condition (line~\ref{l:basic_if}) in \BasicDecider\ is satisfied. 
        Therefore,
        \begin{align}
            Pr[o = 1] &= Pr[|\tilde{q}(D) - q(D_s)| < \tau]\nonumber\\
            &= Pr[|\tilde{q}(D) - q(D)| < \tau] ~~\quad\text{(since $q(D) = q(D_s)$)}\nonumber\\
            &\geq 1 - \delta \label{eq:rtwotsum-th-1}
        \end{align}
        from (\ref{eq:R2T_thm5-1}) for $\beta = \delta$ and $\tau\geq 4\log(\GSbound) \ln{\left( \frac{\log(\GSbound)}{\delta} \right)} \frac{\DLS}{\epsilon}$.
        
        \item Let $q(D)\not\in (q(D_s) - 2\tau, q(D_s) + 2\tau)$. Recall that \rtwotsum\ returns $o = 0$ when the else condition (line~\ref{l:basic_else}) in \BasicDecider\ is satisfied.
        If $q(D)\leq q(D_s) - 2\tau = l - \tau$, then 
        \begin{align}
            Pr[o = 0] &= Pr[|\tilde{q}(D) - q(D_s)|\geq \tau]\nonumber\\
            &\geq Pr[-\tau \geq \tilde{q}(D) - q(D_s)]\nonumber\\
            &= Pr[\tilde{q}(D)\leq l] ~~\quad\text{(since $l = q(D_s) - \tau$)}\nonumber\\
            &\geq Pr[\tilde{q}(D)\leq q(D)] ~~\quad\text{(since $q(D)\leq l - \tau < l$)}\nonumber\\
            &\geq 1 - \delta ~~\quad\text{(from (\ref{eq:R2T_thm5-1}) with $\beta = \delta$)} \label{eq:rtwotsum-th-2}
        \end{align}
        
        Otherwise, if $q(D)\geq q(D_s) + 2\tau = r + \tau$, then
        \begin{align}
            Pr[o = 0] &= Pr[|\tilde{q}(D) - q(D_s)|\geq \tau]\nonumber\\
            &\geq Pr[\tilde{q}(D) - q(D_s)\geq \tau]\nonumber\\
            &= Pr[\tilde{q}(D)\geq r] ~~\quad\text{(since $r = q(D_s) + \tau$)}\nonumber\\
            &\geq Pr[\tilde{q}(D)\geq q(D) - \tau] ~~\quad\text{(given $r\leq q(D) - \tau$)}\nonumber\\
            &\geq 1 - \delta \label{eq:rtwotsum-th-3}
        \end{align}
        from (\ref{eq:R2T_thm5-1}) for $\beta = \delta$ and $\tau\geq 4\log(\GSbound) \ln{\left( \frac{\log(\GSbound)}{\delta} \right)} \frac{\DLS}{\epsilon}$.
    \end{enumerate}
    From (\ref{eq:rtwotsum-th-1}), (\ref{eq:rtwotsum-th-2}) and (\ref{eq:rtwotsum-th-3}), the upper bound of $\tau_{min}^{\rtwotsum, \delta}$ follows.
\end{proof}

\subsection{SVT: Noise Scale in the Laplace Distribution for Monotonic Queries and Proof of Theorem~\ref{thm:SVTsum_is_DP}} \label{subsec:appendix_MonoQueries_isEpsDP} 
% \sr{what does ``Noise Scale in the Laplace Distribution for Monotonic Queries '' mean? remove?} \shweta{$\nu_j \sim Lap(\frac{2}{\epsilon/2})$ if queries are not monotonic. $\nu_j \sim Lap(\frac{1}{\epsilon/2})$ if queries are monotonic -- which we use. This section gives privacy proof for the monotonic query case. Discussion applies to Algo 2 and 3}
For monotonic queries (Definition~\ref{def:SVT_monotonic}), Theorem $3$ \cite{10.14778/3055330.3055331} shows that a factor of $2$ can be saved in the scale of the Laplace distribution used to sample independent noise for each noisy query answer. 
It analyzes Algorithm $7$ \cite{10.14778/3055330.3055331} that returns a DP estimate for each query with a `yes' answer (noisy query answer exceeded noisy threshold), which consumes additional privacy budget as simply returning the already computed noisy query answer violates $\epsilon$-DP. 
Algorithm $7$ \cite{10.14778/3055330.3055331} divides the given privacy budget into $\epsilon_1$ for the noisy thresholds, $\epsilon_2$ for the noisy query answers, and $\epsilon_3$ for any DP estimates to be returned.
However, we note that this analysis does not directly apply to our application of SVT in Algorithm~\ref{algo:err_svtsum} (or Algorithm~\ref{algo:upBound_DS_fromSVT}) because we do not return any DP estimates.
We set $\epsilon_1 = \epsilon_2 = \epsilon/2$.
Below, we prove Theorem~\ref{thm:SVTsum_is_DP}. 
{\bf We note that this proof is adapted from the proof of Theorem 3 \cite{10.14778/3055330.3055331}, but we repeat the steps here for the sake of completeness.}
\begin{proof}
    Let $\top_1$ and $\top_2$ denote the `yes' answers from SVT \cite{10.14778/3055330.3055331} in Algorithm~\ref{algo:err_svtsum}, i.e., when the noisy query answer exceeds the noisy threshold in lines~\ref{l:exceed_r} and \ref{l:exceed_lplus1}, respectively. 
    Let $\bot_1$ and $\bot_2$, respectively, denote when these checks fail. 
    Algorithm~\ref{algo:err_svtsum} stops as soon as either: ($i$) check in line~\ref{l:exceed_r} passes, so the sequence of `yes'/`no' answers is $\bot_1, \ldots, \bot_1, \top_1$, or ($ii$) check in line~\ref{l:exceed_lplus1} passes, so the sequence of `yes'/`no' answers is $\bot_1, \ldots, \bot_1, \bot_2, \ldots, \bot_2, \top_2$. We denote this sequence of `yes'/`no' answers by a string $a$, with length $|a|$.
    
    Let $\mathbb{I}_{\top_1} = \{i: a_i = \top_1\}$, where $a_i$ is the $i$-th character in $a$. Similarly $\mathbb{I}_{\top_2}, \mathbb{I}_{\bot_1}$, and $\mathbb{I}_{\bot_2}$ for $\top_2, \bot_1$, and $\bot_2$, respectively.
    Note that for ($i$), $\mathbb{I}_{\bot_1} = \{1, 2, \ldots, |a| - 1\}$, $\mathbb{I}_{\top_1} = \{|a|\}$, and $\mathbb{I}_{\bot_2} = \mathbb{I}_{\top_2} = \emptyset$.
    Note that for ($ii$), $\mathbb{I}_{\bot_1}\cup \mathbb{I}_{\bot_2} = \{1, 2, \ldots, |a| - 1\}$, $\mathbb{I}_{\top_1} = \emptyset$, and $\mathbb{I}_{\top_2} = \{|a|\}$.
    
    Motivated by \cite{10.14778/3055330.3055331}, let $f_j^{(1)} (D, z) = Pr[q_j(D) + \nu_j < r/t_j + z]$ (not satisfying the condition for line~\ref{l:exceed_r} in Algorithm~\ref{algo:err_svtsum}) and $g_j^{(1)} (D, z) = Pr[q_j(D) + \nu_j \geq r/t_j + z]$ (satisfying the condition for line~\ref{l:exceed_r} in Algorithm~\ref{algo:err_svtsum}), where the superscript represents the $1$st for loop (line~\ref{l:1st_SVT} in Algorithm~\ref{algo:err_svtsum}) and $z$ represents the different values $\rho$ can take (line~\ref{l:div_budget} in Algorithm~\ref{algo:err_svtsum}). Similarly, $f_j^{(2)}$ and $g_j^{(2)}$ denote the two events, respectively, when threshold equals $(l+1)/t_j$ instead of $r/t_j$ ($2$nd for loop in line~\ref{l:2nd_SVT} in Algorithm~\ref{algo:err_svtsum}).
    
    We denote by $\Delta$ the global sensitivity of the queries used: $q_j(D) = q(D, t_j)/t_j$, for $t_j = 2^j, j = 1$ to $\log(\GSbound)$. Recall from Section~\ref{subsec:rtwotsum} that $q(D, t_j)$ is a sum query on the truncated database obtained by removing tuples in $D$ with $A_i$ value greater than $t_j$. The global sensitivity of $q(\cdot, t_j)$ (i.e., $q$ on any database truncated with threshold $t_j$) is $t_j$ as larger values are discarded. Hence, the worst case change in the output of $q_j(\cdot)$ on $D\approx D'$ is $1$ due to division by $t_j$.
    
    {\bf When $a$ satisfies ($i$): } Given that $a = \bot_1, \ldots, \bot_1, \top_1$. Here, $Pr[\text{sequence of `yes'/`no' answers in } \svtsum(q, D, D_s, \tau, \epsilon) = a]$\footnote{Since $\rho\sim Lap(\frac{1}{\epsilon/2})$ (line~\ref{l:div_budget} in Algorithm~\ref{algo:err_svtsum}) is drawn from a continuous distribution with range $(-\infty, \infty)$, we use integration instead of sum.}
    \begin{align*}
        = &\int_{-\infty}^{\infty} Pr[((q_1(D) + \nu_1 < r/t_1 + \rho) \land ... \land (q_{|a|-1}(D) + \nu_{|a|-1} < \\
        &r/t_{|a|-1} + \rho) \land (q_{|a|}(D) + \nu_{|a|} \geq r/t_{|a|} + \rho))~|~\rho = z] Pr[\rho = z] dz\\
    \end{align*}

    \begin{align}
        = \int_{-\infty}^{\infty} Pr[\rho = z] \left( \prod_{j\in \mathbb{I}_{\bot_1}} f_j^{(1)} (D, z) \right) g_{|a|}^{(1)} (D, z) dz
    \end{align}
    conditioning on $\rho$ taking value $z$, 
    and then from the conditional independence of the i.i.d. Laplace noise terms $\nu_j$ given $\rho = z$.
    
    Now, we examine the two cases for monotonic queries, i.e., a sequence of queries for which all answers that change in going from $D$ to a neighboring database $D'$ (i.e. $D\approx D'$) either all increase or all decrease (Definition~\ref{def:SVT_monotonic}):
    \begin{enumerate}
        \item $\forall j,~ q(D, t_j)\geq q(D', t_j)$, where $D\approx D'$.
        \item $\forall j,~ q(D, t_j)\leq q(D', t_j)$, where $D\approx D'$.
    \end{enumerate}
    
    Firstly, suppose $\forall j,~ q(D, t_j)\geq q(D', t_j)$, where $D\approx D'$. Thus,
    \begin{align}
        \forall j,~ q_j(D)\geq q_j(D') ~~\quad\text{(dividing both sides by $t_j$)} \label{equn:case1} 
    \end{align}
    Since the global sensitivity of $q_j$ is $\Delta = 1$, we also have
    \begin{align}
        \forall j,~ q_j(D)\leq q_j(D') + \Delta\label{equn:case1-GS-1} 
    \end{align}
    Using (\ref{equn:case1}) we get, 
    \begin{align}
        f_j^{(1)} (D, z) &= Pr[q_j(D) + \nu_j < r/t_j + z]\nonumber\\
        &\leq Pr[q_j(D') + \nu_j < r/t_j + z] = f_j^{(1)} (D', z)
    \end{align}
    Next, we prove a useful property of the Laplace distribution for $\nu_j\sim Lap(\Delta/\epsilon_2)$, namely: $\forall y,~ Pr[\nu_j = y] \leq e^{\epsilon_2} Pr[\nu_j = y + \Delta]$.
    \begin{alignat}{2}
        &|y + \Delta| &&\leq |y| + \Delta ~~\quad\text{(since $\Delta = 1$)}\nonumber\\
        \implies &e^{|y + \Delta| \frac{\epsilon_2}{\Delta}} &&\leq e^{(|y| + \Delta) \frac{\epsilon_2}{\Delta}}\nonumber\\
        \implies &e^{|y + \Delta| \frac{\epsilon_2}{\Delta}} &&\leq e^{|y|\frac{\epsilon_2}{\Delta}} \cdot e^{\epsilon_2}\nonumber\\
        \implies &e^{-|y|/b} &&\leq e^{-|y + \Delta|/b}\cdot e^{\epsilon_2} ~~\quad\text{(for $b = \Delta/\epsilon_2$)}\nonumber\\
        \implies &Pr[\nu_j = y] &&\leq e^{\epsilon_2} Pr[\nu_j = y + \Delta] ~~\quad\text{(since $\nu_j\sim Lap(\Delta/\epsilon_2)$)}\label{equn:LapDistr-Property-nuj}
    \end{alignat}
    Furthermore, 
    \begin{align}
        g_j^{(1)} (D, z) &= Pr[q_j(D) + \nu_j \geq r/t_j + z]\nonumber\\
        &\leq Pr[q_j(D') + \Delta + \nu_j \geq r/t_j + z] ~~\quad\text{(from (\ref{equn:case1-GS-1}))}\nonumber\\
        & = Pr[q_j(D')  + \nu_j \geq r/t_j + z - \Delta]\nonumber\\
        &\leq e^{\epsilon_2} Pr[q_j(D') + \nu_j \geq r/t_j + z] ~~\quad\text{(from (\ref{equn:LapDistr-Property-nuj}))}\nonumber\\
        & = e^{\epsilon_2} g_j^{(1)} (D', z)
    \end{align}
    So, $Pr[\text{sequence of `yes'/`no' answers in } \svtsum(q, D, D_s, \tau, \epsilon) = a]$
    \begin{align}
        & = \int_{-\infty}^{\infty} Pr[\rho = z] \left( \prod_{j\in \mathbb{I}_{\bot_1}} f_j^{(1)} (D, z) \right) g_{|a|}^{(1)} (D, z) dz\nonumber\\
        & \leq \int_{-\infty}^{\infty} Pr[\rho = z] \left( \prod_{j\in \mathbb{I}_{\bot_1}} f_j^{(1)} (D', z) \right) e^{\epsilon_2} g_{|a|}^{(1)} (D', z) dz\nonumber\\
        & = e^{\epsilon_2} Pr[\text{sequence of `yes'/`no' answers in }\nonumber\\
        &\hspace{12mm}\svtsum(q, D', D_s, \tau, \epsilon) = a] 
    \end{align}
    
    Now, suppose $\forall j, q(D, t_j)\leq q(D', t_j)$, where $D\approx D'$. Thus,
    \begin{align}
        \forall j,~ q_j(D)\leq q_j(D') ~~\quad\text{(dividing both sides by $t_j$)} \label{equn:case2} 
    \end{align}
    Since the global sensitivity of $q_j$ is $\Delta = 1$, we also have
    \begin{align}
        \forall j,~ q_j(D)\geq q_j(D') - \Delta\label{equn:case1-GS-2} 
    \end{align}
    Using (\ref{equn:case1-GS-2}) we get,
    \begin{align}
        f_j^{(1)} (D, z - \Delta) &= Pr[q_j(D) + \nu_j < r/t_j + z - \Delta]\nonumber\\
        &\leq Pr[q_j(D') - \Delta + \nu_j < r/t_j + z - \Delta]\nonumber\\
        & = f_j^{(1)} (D', z)
    \end{align}
    Furthermore,
    \begin{align}
        g_j^{(1)} (D, z - \Delta) &= Pr[q_j(D) + \nu_j \geq r/t_j + z - \Delta]\nonumber\\
        &\leq Pr[q_j(D') + \nu_j \geq r/t_j + z - \Delta] ~~\quad\text{(from (\ref{equn:case2}))}\nonumber\\
        &\leq e^{\epsilon_2} Pr[q_j(D') + \nu_j \geq r/t_j + z] ~~\quad\text{(from (\ref{equn:LapDistr-Property-nuj}))}\nonumber\\
        & = e^{\epsilon_2} g_j^{(1)} (D', z)
    \end{align}
    Since $\rho\sim Lap(\Delta/\epsilon_1)$, by the property of the Laplace distribution
    \begin{align}
        \forall z, ~Pr[\rho = z]\leq e^{\epsilon_1} Pr[\rho = z + \Delta]\label{equn:LapDistr-Property-rho}
    \end{align}
    So, $Pr[\text{sequence of `yes'/`no' answers in } \svtsum(q, D, D_s, \tau, \epsilon) = a]$
    \begin{align}
        & = \int_{-\infty}^{\infty} Pr[\rho = z] \left( \prod_{j\in \mathbb{I}_{\bot_1}} f_j^{(1)} (D, z) \right) g_{|a|}^{(1)} (D, z) dz\nonumber\\
        & = \int_{-\infty}^{\infty} Pr[\rho = z - \Delta] \left( \prod_{j\in \mathbb{I}_{\bot_1}} f_j^{(1)} (D, z - \Delta) \right) g_{|a|}^{(1)} (D, z - \Delta) dz\nonumber\\
        &\hspace{20mm}\text{(by change of integration variable from $z$ to $z - \Delta$)}\nonumber\\
        & \leq \int_{-\infty}^{\infty} e^{\epsilon_1} Pr[\rho = z] \left( \prod_{j\in \mathbb{I}_{\bot_1}} f_j^{(1)} (D', z) \right) e^{\epsilon_2} g_{|a|}^{(1)} (D', z) dz\nonumber\\
        & = e^{\epsilon_1 + \epsilon_2} Pr[\text{sequence of `yes'/`no' answers in }\nonumber\\
        &\hspace{16mm}\svtsum(q, D', D_s, \tau, \epsilon) = a] 
    \end{align}
    
    {\bf When $a$ satisfies (ii):} Given that $a = \bot_1, \ldots, \bot_1, \bot_2, \ldots, \bot_2, \top_2$. $Pr[\text{sequence of `yes'/`no' answers in } \svtsum(q, D, D_s, \tau, \epsilon) = a]$\footnote{Since $\rho\sim Lap(\frac{1}{\epsilon/2})$ (line~\ref{l:div_budget} in Algorithm~\ref{algo:err_svtsum}) is drawn from a continuous distribution with range $(-\infty, \infty)$, we use integration instead of sum.}
    \begin{align}
        = \int_{-\infty}^{\infty} Pr[\rho = z] \left( \prod_{j\in \mathbb{I}_{\bot_1}} f_j^{(1)} (D, z) \right) \left( \prod_{j\in \mathbb{I}_{\bot_2}} f_j^{(2)} (D, z) \right) g_{|a|}^{(2)} (D, z) dz
    \end{align}
    conditioning on $\rho$ taking value $z$, 
    and then from the conditional independence of the i.i.d. Laplace noise terms $\nu_j$ given $\rho = z$.
    % Following the analysis given above, we again get that $Pr[\text{sequence of `yes'/}$ $\text{`no' answers in } \svtsum(q, D, D_s, \tau, \epsilon) = a]\leq e^{\epsilon_1 + \epsilon_2} Pr[\text{sequence of}$ $\text{`yes'/`no' answers in } \svtsum($ $q, D', D_s,$ $\tau, \epsilon) = a]$.
    
    Again, we examine the two cases for monotonic queries.
    Firstly, suppose $\forall j,~ q(D, t_j)\geq q(D', t_j)$, where $D\approx D'$.
    Using (\ref{equn:case1}) we get,
    \begin{align}
        f_j^{(1)} (D, z) &= Pr[q_j(D) + \nu_j < r/t_j + z]\nonumber\\
        &\leq Pr[q_j(D') + \nu_j < r/t_j + z] = f_j^{(1)} (D', z)\\
        f_j^{(2)} (D, z) &= Pr[q_j(D) + \nu_j < (l+1)/t_j + z]\nonumber\\
        &\leq Pr[q_j(D') + \nu_j < (l+1)/t_j + z] = f_j^{(2)} (D', z)
    \end{align}
    % Furthermore,
    \begin{align}
        g_j^{(2)} (D, z) &= Pr[q_j(D) + \nu_j \geq (l+1)/t_j + z]\nonumber\\
        &\leq Pr[q_j(D') + \Delta + \nu_j \geq (l+1)/t_j + z] ~~\quad\text{(from (\ref{equn:case1-GS-1}))}\nonumber\\
        & = Pr[q_j(D')  + \nu_j \geq (l+1)/t_j + z - \Delta]\nonumber\\
        &\leq e^{\epsilon_2} Pr[q_j(D') + \nu_j \geq (l+1)/t_j + z] ~~\quad\text{(from (\ref{equn:LapDistr-Property-nuj}))}\nonumber\\
        & = e^{\epsilon_2} g_j^{(2)} (D', z)
    \end{align}
    So, $Pr[\text{sequence of `yes'/`no' answers in } \svtsum(q, D, D_s, \tau, \epsilon) = a]$
    \begin{align}
        & = \int_{-\infty}^{\infty} Pr[\rho = z] \left( \prod_{j\in \mathbb{I}_{\bot_1}} f_j^{(1)} (D, z) \right) \left( \prod_{j\in \mathbb{I}_{\bot_2}} f_j^{(2)} (D, z) \right) g_{|a|}^{(2)} (D, z) dz\nonumber\\
        & \leq \int_{-\infty}^{\infty} Pr[\rho = z] \left( \prod_{j\in \mathbb{I}_{\bot_1}} f_j^{(1)} (D', z) \right) \left( \prod_{j\in \mathbb{I}_{\bot_2}} f_j^{(2)} (D', z) \right)\cdot\nonumber\\
        &\hspace{12mm} e^{\epsilon_2} g_{|a|}^{(2)} (D', z) dz\nonumber\\
        & = e^{\epsilon_2} Pr[\text{sequence of `yes'/`no' answers in }\nonumber\\
        &\hspace{12mm}\svtsum(q, D', D_s, \tau, \epsilon) = a]
    \end{align}

    Now, suppose $\forall j, q(D, t_j)\leq q(D', t_j)$, where $D\approx D'$.
    Using (\ref{equn:case1-GS-2}),
    \begin{align}
        f_j^{(1)} (D, z - \Delta) &= Pr[q_j(D) + \nu_j < r/t_j + z - \Delta]\nonumber\\
        &\leq Pr[q_j(D') - \Delta + \nu_j < r/t_j + z - \Delta]\nonumber\\
        & = f_j^{(1)} (D', z)\\
        f_j^{(2)} (D, z - \Delta) &= Pr[q_j(D) + \nu_j < (l+1)/t_j + z - \Delta]\nonumber\\
        &\leq Pr[q_j(D') - \Delta + \nu_j < (l+1)/t_j + z - \Delta]\nonumber\\
        & = f_j^{(2)} (D', z)
    \end{align}
    Furthermore,
    \begin{align}
        g_j^{(2)} (D, z - \Delta) &= Pr[q_j(D) + \nu_j \geq (l+1)/t_j + z - \Delta]\nonumber\\
        &\leq Pr[q_j(D') + \nu_j \geq (l+1)/t_j + z - \Delta] ~~\text{(from (\ref{equn:case2}))}\nonumber\\
        &\leq e^{\epsilon_2} Pr[q_j(D') + \nu_j \geq (l+1)/t_j + z] ~~\text{(from (\ref{equn:LapDistr-Property-nuj}))}\nonumber\\
        & = e^{\epsilon_2} g_j^{(2)} (D', z)
    \end{align}
    Recall from (\ref{equn:LapDistr-Property-rho}) that $\forall z, ~Pr[\rho = z]\leq e^{\epsilon_1} Pr[\rho = z + \Delta]$.
    So, $Pr[\text{sequence of `yes'/`no' answers in } \svtsum(q, D, D_s, \tau, \epsilon) = a]$
    \begin{align}
        & = \int_{-\infty}^{\infty} Pr[\rho = z] \left( \prod_{j\in \mathbb{I}_{\bot_1}} f_j^{(1)} (D, z) \right) \left( \prod_{j\in \mathbb{I}_{\bot_2}} f_j^{(2)} (D, z) \right) g_{|a|}^{(2)} (D, z) dz\nonumber\\
        & = \int_{-\infty}^{\infty} Pr[\rho = z - \Delta] \left( \prod_{j\in \mathbb{I}_{\bot_1}} f_j^{(1)} (D, z - \Delta) \right) \left( \prod_{j\in \mathbb{I}_{\bot_2}} f_j^{(2)} (D, z - \Delta) \right)\cdot \nonumber\\
        &g_{|a|}^{(2)} (D, z - \Delta) dz ~~\text{(by change of integration variable from $z$ to $z - \Delta$)}\nonumber\\
        & \leq \int_{-\infty}^{\infty} e^{\epsilon_1} Pr[\rho = z] \left( \prod_{j\in \mathbb{I}_{\bot_1}} f_j^{(1)} (D', z) \right) \left( \prod_{j\in \mathbb{I}_{\bot_2}} f_j^{(2)} (D', z) \right)\cdot \nonumber\\
        &\hspace{12mm}e^{\epsilon_2} g_{|a|}^{(2)} (D', z) dz\nonumber\\
        & = e^{\epsilon_1 + \epsilon_2} Pr[\text{sequence of `yes'/`no' answers in }\nonumber\\
        &\hspace{16mm}\svtsum(q, D', D_s, \tau, \epsilon) = a]
    \end{align}
\end{proof}

The sequence of `yes'/`no' answers for Algorithm~\ref{algo:upBound_DS_fromSVT} corresponds to $a = \bot_1, \ldots, \bot_1, \top_1$ in the above proof. In Algorithm~\ref{algo:upBound_DS_fromSVT}, SVT is run with two-thirds of $\epsilon'$, its input privacy budget, and $\epsilon_1 = \epsilon_2 = \epsilon'/3$. By sequential composition (Proposition~\ref{prop:DP-comp-post}), Algorithm~\ref{algo:upBound_DS_fromSVT} is $\epsilon'$-DP.

\subsection{Improvement on \svtsum\ to lower error}\label{subsec:appendix_SVT_optimization}
Recall from (\ref{eq:interval}) that $\intvl = (l, r) = (q(D_s) - \tau, q(D_s) + \tau)$.
\svtsum\ incurs high error if $q(D)$ and $r$ (or $l$) are close because in the later iterations where $t_j$ values are large, the check is easily influenced by noise.
Note that $\forall t_k\geq \DLS, q(D, t_k) = q(D)$, where \DLS\ is the downward local sensitivity (Definition~\ref{def:DLS}).
Thus, in Algorithm~\ref{algo:upBound_DS_fromSVT}, we bound $t_j$ by a private bound for \DLS\ \cite{7837832}.
In the loops in lines \ref{l:1st_SVT} and \ref{l:2nd_SVT} in Algorithm~\ref{algo:err_svtsum}, $t_j$ is bounded by (the next largest power of $2$ from) $\GSbound$. Instead, in the improved \svtsum, we use a private upper bound of \DLS\ to reduce the noise in Algorithm~\ref{algo:upBound_DS_fromSVT}. This bound is still $\GSbound$ in the worst case, but can take lower values. 
To achieve this, we spend $\epsilon'$ equal to a third of \svtsum's privacy budget $\epsilon$ on Algorithm~\ref{algo:upBound_DS_fromSVT}. The remaining $2\epsilon/3$ of the privacy budget is used in Algorithm~\ref{algo:err_svtsum}, hence each step in \svtsum\ uses $\epsilon/3$ instead of $\epsilon/2$.

Since the number of tuples $n$ in $D$ satisfying \whereClause\ (the predicates in the WHERE clause in $q$) is private, we spend $\epsilon'/3$ of  Algorithm~\ref{algo:upBound_DS_fromSVT}'s privacy budget $\epsilon'$ to compute a DP estimate for $n$ (line~\ref{l:DS_bound_ntilde}). 
SVT (loop in line~\ref{l:DS_bound_SVT}) is run with the remaining privacy budget, where $\epsilon'/3$ is consumed by $\rho$ that is used to get noisy thresholds for all queries and $\epsilon'/3$ is used to add independent Laplace noise to query answers. 
Query $q_j^c(D)$ computes the number of tuples in $D$ that satisfy \whereClause\ and have $A_i$ value less than or equal to $2^j$, for $j = 1$ to $\lceil \log(\GSbound) \rceil$.
Note that $q_j^c(D)$s (line~\ref{l:DS_bound_query}) are monotonic (Definition~\ref{def:SVT_monotonic}), which helps save a factor of $2$ in the noise scale for $\nu_j^c$ \cite{10.14778/3055330.3055331} (discussed further in Appendix~\ref{subsec:appendix_MonoQueries_isEpsDP}).
The return value becomes the largest $j$ used in lines~\ref{l:1st_SVT} and \ref{l:2nd_SVT} in Algorithm~\ref{algo:err_svtsum}.

\IncMargin{1em}
\begin{algorithm}[!ht]
    \caption{Private bound for \DLS}
    \label{algo:upBound_DS_fromSVT}
    
    \SetKwInOut{Input}{Input}\SetKwInOut{Output}{Output}
    \LinesNumbered
    \Input{$q$ - \sumQuery\ query, $D$ - private database, $\epsilon'$ - privacy budget (= $\epsilon/3$), $\theta$ - fraction in $(0, 1)$ (default $0.95$)}
    \Output{Private bound for \DLS}
    
    \BlankLine
    
    \SetKwFunction{FMain}{PrivateBoundDS}
    \SetKwProg{Fn}{Function}{:}{}
    \Fn{\FMain{$q, D, \epsilon', \theta$}}{
        $n\gets$ number of tuples in $D$ that satisfy \whereClause\ in $q$\;
        $\nu_q\gets Lap(\frac{1}{\epsilon'/3}), \tilde{n}\gets n + \nu_q$\; \label{l:DS_bound_ntilde}
        $A_i\gets$ aggregate attribute in $q$\;
        $\rho\gets Lap(\frac{1}{\epsilon'/3})$\;
        \For{$j\in \{1, 2, 3, \ldots, \lceil \log(\GSbound) \rceil \}$}
        {
            \label{l:DS_bound_SVT}
            $\nu_j^c \sim Lap(\frac{1}{\epsilon'/3})$, $q_j^c(D) \gets$ \texttt{SELECT COUNT(*) FROM $D$ WHERE \whereClause\ AND $A_i \leq 2^j$}\; \label{l:DS_bound_query}
            \If{$q_j^c(D) + \nu_j^c \geq \theta*\tilde{n} + \rho$}
            {
                \Return $j$\;
            }
        }
        \Return $\lceil \log(\GSbound) \rceil$\;
    }
\end{algorithm}
\DecMargin{1em}

By sequential composition and post-processing (Proposition~\ref{prop:DP-comp-post}), the following holds:

\begin{observation}
    Algorithm~\ref{algo:upBound_DS_fromSVT} satisfies $\epsilon/3$-DP, and consequently the improved \svtsum\ 
    (with Algorithm~\ref{algo:upBound_DS_fromSVT}) 
    satisfies $\epsilon$-DP.
\end{observation}

\subsection{Empirical Error Analysis of \svtsum} \label{subsec:appendix_tau_min_svtsum}

Our experiments suggest that \svtsum's error may not benefit from increasing $\epsilon$ when $\tau$ is small and the query answer of the given \sumQuery\ query $q$ on $D$ is inside the interval \intvl\ (as defined in (\ref{eq:interval})). We find that as $\tau$ increases, error decreases more noticeably. 

For several queries such as $q_{13}$ to $q_{18}$, \svtsum's error was in between that of \lapsum's error and \rtwotsum's error when $\tau$ was varied.
\svtsum's error was comparable to that of \rtwotsum's error for $q_{19}$ with the third largest output on $D$ and $\DLS = 276799$ (Definition~\ref{def:DLS}).
Queries $q_{20}$ and $q_{21}$ have output greater than $685M$, and we saw that the error decreased quicker, i.e., for smaller $\tau$ values, but stayed comparable to or worse than \rtwotsum's error. The \DLS\ values for $q_{20}$ and $q_{21}$ are $403353$ and $500000$, respectively, whereas the bound \GSbound\ on the global sensitivity is $500000$. Observe that this gap is not large.

\section{Proofs from Section~\ref{sec:median}}
% \shweta{Change $n'$ to $n_{\whereClause}$}
\subsection{Error Analysis of \emmed} \label{subsec:appendix_emmed_err}
\begin{proposition} %\label{prop:emmed_err}
    Given a private database $D$, 
    synthetic database $D_s$, 
    \medianQuery\ query $q$ on aggregate attribute $A_i$, 
    distance bound $\tau$, and 
    privacy budget $\epsilon$.
    Interval $\intvl = (l, r) = (q(D_s) - \tau, q(D_s) + \tau)$ (\ref{eq:interval}).
    \emmed\ satisfies the following: 
    \begin{enumerate}
        \item If $q(D)\not\in \intvl$ but $o = 1$, then $err(\cdot)$ equals
        $$\sum\limits_{e\in dom(A_i)\cap \intvl} e^{\epsilon u(D, e)/2} / \sum\limits_{e\in dom(A_i)} e^{\epsilon u(D, e)/2}$$
        
        \item If $q(D)\in \intvl$ but $o = 0$, then $err(\cdot)$ equals
        $$\sum\limits_{e\in dom(A_i)\setminus \intvl} e^{\epsilon u(D, e)/2} / \sum\limits_{e\in dom(A_i)} e^{\epsilon u(D, e)/2}$$
    \end{enumerate}
\end{proposition}
\begin{proof}
    Recall from Definition~\ref{def:EM} that the Exponential Mechanism (EM) picks an element $e$ in the output space $\mathcal{R}$ with probability given by $c\cdot e^{\frac{\epsilon u(D, e)}{2 \Delta u}}$, where $u(D, e)$ is the utility of $e\in \mathcal{R}$ on $D$ and $c$ is a positive constant.
    \emmed\ runs \BasicDecider\ with the EM as \dpmech, with additional parameters $\mathcal{R} = dom(A_i)$ and score function $u(D, e) = -|rank_{\whereClause}(D, e) - \frac{n'}{2}|$, $\forall e\in \mathcal{R}$ and $n'$ equal to the number of tuples in $D$ satisfying $\varphi$ (\texttt{WHERE} clause) in the given \medianQuery\ query $q$.
    The sensitivity of the score function equals $1$ because rank of any $e$ either stays the same or changes in the same direction as $n'$ on $D'\approx D$.
    
    For ($1$):
    Given that $q(D)\not\in \intvl$ but $o = 1$. 
    Note that \BasicDecider\ returns $o = 1$ when the condition in line~\ref{l:basic_if} is met, i.e., when the EM (i.e. $\dpmech(\cdot)$) uses $e\in dom(A_i)\cap \intvl$ as the noisy estimate so that  $|e - q(D_s)| < \tau$.
    The error (i.e. probability of this happening) is given by
    \begin{align}
        \sum\limits_{e\in dom(A_i)\cap \intvl} c\cdot e^{\epsilon u(D, e)/2}
    \end{align}
    We know that
    \begin{align}
        \sum\limits_{e\in dom(A_i)} c\cdot e^{\epsilon u(D, e)/2} = 1
    \end{align}
    Therefore, $err(\cdot)$ equals
    \begin{align}
        \sum\limits_{e\in dom(A_i)\cap \intvl} e^{\epsilon u(D, e)/2} / \sum\limits_{e\in dom(A_i)} e^{\epsilon u(D, e)/2}
    \end{align}
    
    For ($2$): 
    Given that $q(D)\in \intvl$ but $o = 0$. 
    Note that \BasicDecider\ returns $o = 0$ when the condition in line~\ref{l:basic_if} is unmet, i.e., when the EM (i.e. $\dpmech(\cdot)$) uses $e\in dom(A_i)\setminus \intvl$ as the noisy estimate so that  $|e - q(D_s)|\geq \tau$.
    The error (i.e. probability of this happening) is given by
    \begin{align}
        \sum\limits_{e\in dom(A_i)\setminus \intvl} c\cdot e^{\epsilon u(D, e)/2}
    \end{align}
    We know that
    \begin{align}
        \sum\limits_{e\in dom(A_i)} c\cdot e^{\epsilon u(D, e)/2} = 1
    \end{align}
    Therefore, $err(\cdot)$ equals
    \begin{align}
        \sum\limits_{e\in dom(A_i)\setminus \intvl} e^{\epsilon u(D, e)/2} / \sum\limits_{e\in dom(A_i)} e^{\epsilon u(D, e)/2}
    \end{align}
\end{proof}

\subsection{Empirical Analyses and Comparison of Errors of \histmed\ \& \emmed} \label{subsec:appendix_histmed_err}

\paratitle{Empirical Error Analysis of \emmed}
Let us revisit the empirical analysis for \emmed\ from Section~\ref{sec:experiments}.
For $q_{22}$ and $q_{29}$, the support, i.e., tuples in $D$ that satisfy the predicates in the \texttt{WHERE} clause, is large. As a result, the probability distribution used to sample the noisy estimate for the median is more concentrated around the correct value. The opposite is true for $q_{23}$ and $q_{28}$. We see error decrease as $\tau$ increases.
When $\epsilon$ increases, the variance of the probability distribution used to sample the noisy estimate in \emmed\ decreases. \emmed's error has a decreasing trend except for $q_{29}$ because the scores are negatives numbers with large magnitudes due to the support being large (Definition~\ref{def:EM}).\\

\paratitle{Empirical Error Analysis of \histmed}
Recall that $n'$ is the number of tuples in $D$ that satisfy $\varphi$ (\texttt{WHERE} clause) in the given \medianQuery\ query $q$.
Our experiments suggest that when the gap between $q_1(D)$ and the majority, i.e., $\lceil \frac{n'}{2} \rceil$, is large in comparison to the noise scales (in the Laplace distributions) used to get the respective noisy estimates, \histmed\ makes fewer mistakes in line~\ref{l:q1_return} in Algorithm~\ref{algo:err_1_q_med}. Similarly for the gap between $q_2(D)$ and $\lceil \frac{n'}{2} \rceil$ in line~\ref{l:q2_return} in Algorithm~\ref{algo:err_1_q_med}.
As $\epsilon$ increases, the noise scales decrease, reducing the chance of the noise terms causing \histmed\ to make a mistake.\\

\paratitle{Empirical Error Comparison of \histmed\ \& \emmed}
In contrast to \histmed, the outcome from \emmed\ is more dependent on the ranks (in $D$) of domain values that are in the neighborhood of the true median $q(D)$. For example, \histmed's error was smaller than that of \emmed\ on queries $q_{23}$ and $q_{28}$, even at small $\tau$ and $\epsilon$.

\begin{table*}[ht]
    \centering \small
    \begin{tabular}{| c | c | c | c |}
        \rowcolor[HTML]{C0C0C0} \hline \multicolumn{2}{| c |}{{\bf Query}} & {\bf WHERE clause} & $q(D), q(D_s)$\\
        \hline \cellcolor[HTML]{D7D7D7}& $q_1$ & \begin{tabular}{@{}c@{}} \texttt{$sex$ LIKE `Female' AND $race$ LIKE `White-}\\ \texttt{American Indian-Asian' AND $workly$ LIKE `Yes'} \end{tabular} & \begin{tabular}{@{}c@{}} $34$ \\ $34$ \end{tabular}\\
        \cline{2-4} \cellcolor[HTML]{D7D7D7}& $q_2$ & \begin{tabular}{@{}c@{}} \texttt{$race$ LIKE `White-Black' AND $marst$ LIKE `Never married/}\\ \texttt{single' AND $citizen$ LIKE `Naturalized citizen'} \end{tabular} & \begin{tabular}{@{}c@{}} $54$ \\ $54$ \end{tabular}\\
        \cline{2-4} \cellcolor[HTML]{D7D7D7}& $q_3$ & \begin{tabular}{@{}c@{}} \texttt{$sex$ LIKE `Male' AND $educ$ LIKE `Doctorate}\\ \texttt{degree' AND $marst$ LIKE `Separated'} \end{tabular} & \begin{tabular}{@{}c@{}} $87$ \\ $91$ \end{tabular}\\
        \cline{2-4} \cellcolor[HTML]{D7D7D7}& $q_4$ & \begin{tabular}{@{}c@{}} \texttt{$race$ LIKE `White-American Indian' AND $marst$}\\ \texttt{LIKE `Widowed' AND $citizen$ LIKE `Born in U.S'} \end{tabular} & \begin{tabular}{@{}c@{}} $624$ \\ $624$ \end{tabular}\\
        \cline{2-4} \cellcolor[HTML]{D7D7D7}& $q_5$ & \begin{tabular}{@{}c@{}} \texttt{$sex$ LIKE `Female' AND $educ$ LIKE `Grades $5$}\\ \texttt{or $6$' AND $marst$ LIKE `Never married/single'} \end{tabular} & \begin{tabular}{@{}c@{}} $1560$ \\ $1606$ \end{tabular}\\
        \cline{2-4} \cellcolor[HTML]{D7D7D7}& $q_6$ & \begin{tabular}{@{}c@{}} \texttt{$sex$ LIKE `Female' AND $race$ LIKE}\\ \texttt{`White-Black' AND $workly$ LIKE `Yes'} \end{tabular} & \begin{tabular}{@{}c@{}} $1893$ \\ $1884$ \end{tabular}\\
        \cline{2-4} \cellcolor[HTML]{D7D7D7}& $q_7$ & \begin{tabular}{@{}c@{}} \texttt{$sex$ LIKE `Female' AND $race$ LIKE}\\ \texttt{`Asian only' AND $workly$ LIKE `No'} \end{tabular} & \begin{tabular}{@{}c@{}} $18693$ \\ $18777$ \end{tabular}\\
        \cline{2-4} \cellcolor[HTML]{D7D7D7}& $q_8$ & \begin{tabular}{@{}c@{}} \texttt{$sex$ LIKE `Female' AND $educ$ LIKE `Master's degree'}\\ \texttt{AND $marst$ LIKE `Married, spouse present'} \end{tabular} & \begin{tabular}{@{}c@{}} $36756$ \\ $36404$ \end{tabular}\\
        \cline{2-4} \cellcolor[HTML]{D7D7D7}& $q_9$ & \begin{tabular}{@{}c@{}} \texttt{$race$ LIKE `Black' AND $marst$ LIKE `Married, spouse}\\ \texttt{present' AND $citizen$ LIKE `Born in U.S'} \end{tabular} & \begin{tabular}{@{}c@{}} $41691$ \\ $42455$ \end{tabular}\\
        \cline{2-4} \cellcolor[HTML]{D7D7D7}& $q_{10}$ & \begin{tabular}{@{}c@{}} \texttt{$sex$ LIKE `Female' AND $educ$ LIKE `Bachelor's}\\ \texttt{degree' AND $marst$ LIKE `Married, spouse present'} \end{tabular} & \begin{tabular}{@{}c@{}} $79723$ \\ $79283$ \end{tabular}\\
        \cline{2-4} \cellcolor[HTML]{D7D7D7}& $q_{11}$ & \begin{tabular}{@{}c@{}} \texttt{$sex$ LIKE `Female' AND $race$ LIKE}\\ \texttt{`White' AND $workly$ LIKE `Yes'} \end{tabular} & \begin{tabular}{@{}c@{}} $320934$ \\ $320001$ \end{tabular}\\
        \cline{2-4} \multirow{-24}{*}{\cellcolor[HTML]{D7D7D7} \countQuery}& $q_{12}$ &  \begin{tabular}{@{}c@{}} \texttt{$race$ LIKE `White' AND $marst$ LIKE `Married,}\\ \texttt{spouse present' AND $citizen$ LIKE `Born in U.S'} \end{tabular} & \begin{tabular}{@{}c@{}} $471994$ \\ $470483$ \end{tabular}\\
        
        \hline \cellcolor[HTML]{D7D7D7}& $q_{13}$ & \begin{tabular}{@{}c@{}} \texttt{$sex$ LIKE `Female' AND $race$ LIKE}\\ \texttt{`White-Black' AND $workly$ LIKE `No'} \end{tabular} & \begin{tabular}{@{}c@{}} $6915340$ \\ $6942866$ \end{tabular}\\
        \cline{2-4} \cellcolor[HTML]{D7D7D7}& $q_{14}$ & \begin{tabular}{@{}c@{}} \texttt{$race$ LIKE `Asian only' AND $marst$ LIKE}\\ \texttt{`Separated' AND $citizen$ LIKE `Born in U.S'} \end{tabular} & \begin{tabular}{@{}c@{}} $7653748$ \\ $7818555$ \end{tabular}\\
        \cline{2-4} \cellcolor[HTML]{D7D7D7}& $q_{15}$ & \begin{tabular}{@{}c@{}} \texttt{$sex$ LIKE `Male' AND $educ$ LIKE `Professional school}\\ \texttt{degree' AND $marst$ LIKE `Married, spouse absent'} \end{tabular} & \begin{tabular}{@{}c@{}} $9891949$ \\ $10129562$ \end{tabular}\\
        \cline{2-4} \cellcolor[HTML]{D7D7D7}& $q_{16}$ & \begin{tabular}{@{}c@{}} \texttt{$race$ LIKE `White' AND $marst$ LIKE `Divorced'}\\ \texttt{AND $citizen$ LIKE `Not a citizen'} \end{tabular} & \begin{tabular}{@{}c@{}} $123543040$ \\ $128497757$ \end{tabular}\\
        \cline{2-4} \cellcolor[HTML]{D7D7D7}& $q_{17}$ & \begin{tabular}{@{}c@{}} \texttt{$sex$ LIKE `Female' AND $educ$ LIKE `Master's}\\ \texttt{degree' AND $marst$ LIKE `Widowed'} \end{tabular} & \begin{tabular}{@{}c@{}} $138131505$ \\ $139557517$ \end{tabular}\\
        \cline{2-4} \cellcolor[HTML]{D7D7D7}& $q_{18}$ & \begin{tabular}{@{}c@{}} \texttt{$sex$ LIKE `Male' AND $race$ LIKE}\\ \texttt{`Asian only' AND $workly$ LIKE `No'} \end{tabular} &  \begin{tabular}{@{}c@{}} $143179279$ \\ $149069654$ \end{tabular}\\
        \cline{2-4} \cellcolor[HTML]{D7D7D7}& $q_{19}$ & \begin{tabular}{@{}c@{}} \texttt{$sex$ LIKE `Male' AND $race$ LIKE}\\ \texttt{`Black' AND $workly$ LIKE `No'} \end{tabular} &  \begin{tabular}{@{}c@{}} $327036909$ \\ $336130017$ \end{tabular}\\
        \cline{2-4} \cellcolor[HTML]{D7D7D7}& $q_{20}$ & \begin{tabular}{@{}c@{}} \texttt{$sex$ LIKE `Female' AND $educ$ LIKE `High school diploma}\\ \texttt{ or equivalent' AND $marst$ LIKE `Never married/single'} \end{tabular} &  \begin{tabular}{@{}c@{}} $685635093$ \\ $690711885$ \end{tabular}\\
        \cline{2-4} \multirow{-18}{*}{\cellcolor[HTML]{D7D7D7} {\begin{tabular}{@{}c@{}} \sumQuery\\ on\\ $inctot$ \end{tabular}}}& $q_{21}$ & \begin{tabular}{@{}c@{}} \texttt{$race$ LIKE `White' AND $marst$ LIKE `Married,}\\ \texttt{spouse present' AND $citizen$ LIKE `Born in U.S'} \end{tabular} &  \begin{tabular}{@{}c@{}} $23542765109$ \\ $23434676868$ \end{tabular}\\
        
        \hline \cellcolor[HTML]{D7D7D7}& $q_{22}$ & \begin{tabular}{@{}c@{}} \texttt{$workly$ LIKE `Yes' AND $classwkr$ LIKE `Wage/salary,}\\ \texttt{private' AND $educ$ LIKE `Bachelor's degree'} \end{tabular} & \begin{tabular}{@{}c@{}} $41$ \\ $41$ \end{tabular}\\
        \cline{2-4} \cellcolor[HTML]{D7D7D7}& $q_{23}$ & \begin{tabular}{@{}c@{}} \texttt{$sex$ LIKE `Male' AND $race$ LIKE `White-}\\ \texttt{Black' AND $relate$ LIKE `Spouse'} \end{tabular} & \begin{tabular}{@{}c@{}} $40$ \\ $40$ \end{tabular}\\
        \cline{2-4} \cellcolor[HTML]{D7D7D7}& $q_{24}$ & \begin{tabular}{@{}c@{}} \texttt{$sex$ LIKE `Male' AND $race$ LIKE `Black'}\\ \texttt{AND $classwkr$ LIKE `Wage/salary, private'} \end{tabular} & \begin{tabular}{@{}c@{}} $40$ \\ $40$ \end{tabular}\\
        \cline{2-4} \cellcolor[HTML]{D7D7D7}& $q_{25}$ & \begin{tabular}{@{}c@{}} \texttt{$sex$ LIKE `Female' AND $race$ LIKE}\\ \texttt{`Black' AND $workly$ LIKE `Yes'} \end{tabular} & \begin{tabular}{@{}c@{}} $41$ \\ $41$ \end{tabular}\\
        \cline{2-4} \cellcolor[HTML]{D7D7D7}& $q_{26}$ & \begin{tabular}{@{}c@{}} \texttt{$sex$ LIKE `Female' AND $educ$ LIKE `Some college but}\\ \texttt{no degree' AND $marst$ LIKE `Married, spouse absent'} \end{tabular} & \begin{tabular}{@{}c@{}} $40$ \\ $40$ \end{tabular}\\
        \cline{2-4} \cellcolor[HTML]{D7D7D7}& $q_{27}$ & \begin{tabular}{@{}c@{}} \texttt{$sex$ LIKE `Male' AND $citizen$ LIKE}\\ \texttt{`Born in U.S' AND $workly$ LIKE `Yes'} \end{tabular} & \begin{tabular}{@{}c@{}} $41$ \\ $41$ \end{tabular}\\
        \cline{2-4} \cellcolor[HTML]{D7D7D7}& $q_{28}$ & \begin{tabular}{@{}c@{}} \texttt{$race$ LIKE `Asian only' AND $marst$ LIKE}\\ \texttt{`Separated' AND $citizen$ LIKE `Born in U.S'} \end{tabular} & \begin{tabular}{@{}c@{}} $39$ \\ $39$ \end{tabular}\\
        \cline{2-4} \cellcolor[HTML]{D7D7D7}& $q_{29}$ & \begin{tabular}{@{}c@{}} \texttt{$sex$ LIKE `Male' AND $race$ LIKE `White'}\\ \texttt{AND $classwkr$ LIKE `Wage/salary, private'} \end{tabular} & \begin{tabular}{@{}c@{}} $40$ \\ $40$ \end{tabular}\\
        \cline{2-4} \multirow{-18}{*}{\cellcolor[HTML]{D7D7D7} {\begin{tabular}{@{}c@{}} \medianQuery\\ on\\ $age$ \end{tabular}}}& $q_{30}$ & \begin{tabular}{@{}c@{}} \texttt{$workly$ LIKE `Yes' AND $classwkr$ LIKE `Armed}\\ \texttt{forces' AND $educ$ LIKE `Doctorate degree'} \end{tabular} & \begin{tabular}{@{}c@{}} $40$ \\ $40$ \end{tabular}\\
        
        \hline
    \end{tabular}
    \caption{Queries used in experiments \reva{for the database derived from the IPUMS-CPS survey data}.} \label{tbl:exp_queries}
\end{table*}

\begin{table*}[ht]
    \centering \small
    \reva{
        \begin{tabular}{| c | c | c | c |}
            \rowcolor[HTML]{C0C0C0} \hline \multicolumn{2}{| c |}{{\bf Query}} & {\bf WHERE clause} & $q(D), q(D_s)$\\
            \hline \cellcolor[HTML]{D7D7D7}& $q_{31}$ & \texttt{$passenger\_count = 5$ AND $rateCodeID = 5$ AND $payment\_type = 1$} & \begin{tabular}{@{}c@{}} $11$ \\ $11$ \end{tabular}\\
            \cline{2-4} \cellcolor[HTML]{D7D7D7}& $q_{32}$ & \texttt{$passenger\_count = 1$ AND $trip\_distance = 1$ AND $tip\_amount = 9$} & \begin{tabular}{@{}c@{}} $56$ \\ $56$ \end{tabular}\\
            \cline{2-4} \cellcolor[HTML]{D7D7D7}& $q_{33}$ & \texttt{$trip\_distance = 12$ AND $total\_amount = 53$ AND $congestion\_surcharge = 0$} & \begin{tabular}{@{}c@{}} $92$ \\ $91$ \end{tabular}\\
            \cline{2-4} \cellcolor[HTML]{D7D7D7}& $q_{34}$ & \texttt{$trip\_distance = 2$ AND $fare\_amount = 18$ AND $congestion\_surcharge = 2$} & \begin{tabular}{@{}c@{}} $751$ \\ $748$ \end{tabular}\\
            \cline{2-4} \cellcolor[HTML]{D7D7D7}& $q_{35}$ & \texttt{$passenger\_count = 2$ AND $trip\_distance = 7$ AND $tip\_amount = 5$} & \begin{tabular}{@{}c@{}} $1064$ \\ $1026$ \end{tabular}\\
            \cline{2-4} \cellcolor[HTML]{D7D7D7}& $q_{36}$ & \texttt{$trip\_distance = 17$ AND $rateCodeID = 1$ AND $payment\_type = 1$} & \begin{tabular}{@{}c@{}} $1926$ \\ $1909$ \end{tabular}\\
            \cline{2-4} \cellcolor[HTML]{D7D7D7}& $q_{37}$ & \texttt{$trip\_distance = 5$ AND $payment\_type = 2$ AND $fare\_amount \geq 16$} & \begin{tabular}{@{}c@{}} $13549$ \\ $13207$ \end{tabular}\\
            \cline{2-4} \cellcolor[HTML]{D7D7D7}& $q_{38}$ & \texttt{$vendorID = 1$ AND $store\_and\_fwd\_flag$ LIKE `Y' AND $payment\_type = 1$} & \begin{tabular}{@{}c@{}} $35190$ \\ $35280$ \end{tabular}\\
            \cline{2-4} \cellcolor[HTML]{D7D7D7}& $q_{39}$ & \texttt{$trip\_distance \geq 4$ AND $fare\_amount \geq 22$ AND $congestion\_surcharge = 0$} & \begin{tabular}{@{}c@{}} $77872$ \\ $78774$ \end{tabular}\\
            \cline{2-4} \multirow{-20}{*}{\cellcolor[HTML]{D7D7D7} \countQuery}& $q_{40}$ & \texttt{$trip\_distance \leq 4$ AND $rateCodeID = 1$ AND $payment\_type = 2$} & \begin{tabular}{@{}c@{}} $346109$ \\ $346045$ \end{tabular}\\
            
            \hline \cellcolor[HTML]{D7D7D7}& $q_{41}$ & \texttt{$passenger\_count = 1$ AND $trip\_distance = 1$ AND $tip\_amount = 9$} & \begin{tabular}{@{}c@{}} $1357$ \\ $1357$ \end{tabular}\\
            \cline{2-4} \cellcolor[HTML]{D7D7D7}& $q_{42}$ & \texttt{$trip\_distance = 12$ AND $total\_amount = 53$ AND $congestion\_surcharge = 0$} & \begin{tabular}{@{}c@{}} $4876$ \\ $4823$ \end{tabular}\\
            \cline{2-4} \cellcolor[HTML]{D7D7D7}& $q_{43}$ & \texttt{$passenger\_count = 3$ AND $rateCodeID = 3$ AND $payment\_type = 2$} & \begin{tabular}{@{}c@{}} $8239$ \\ $8133$ \end{tabular}\\
            \cline{2-4} \cellcolor[HTML]{D7D7D7}& $q_{44}$ & \texttt{$vendorID = 1$ AND $store\_and\_fwd\_flag$ LIKE `Y' AND $payment\_type = 2$} & \begin{tabular}{@{}c@{}} $161408$ \\ $161706$ \end{tabular}\\
            \cline{2-4} \cellcolor[HTML]{D7D7D7}& $q_{45}$ & \texttt{$passenger\_count < 2$ AND $trip\_distance = 8$ AND $tip\_amount \leq 2$} & \begin{tabular}{@{}c@{}} $218047$ \\ $213685$ \end{tabular}\\
            \cline{2-4} \cellcolor[HTML]{D7D7D7}& $q_{46}$ & \texttt{$trip\_distance \leq 2$ AND $total\_amount \leq 15$ AND $congestion\_surcharge = 0$} & \begin{tabular}{@{}c@{}} $324730$ \\ $480688$ \end{tabular}\\
            \cline{2-4} \cellcolor[HTML]{D7D7D7}& $q_{47}$ & \texttt{$trip\_distance \geq 8$ AND $payment\_type = 2$ AND $fare\_amount \geq 50$} & \begin{tabular}{@{}c@{}} $1133124$ \\ $1120773$ \end{tabular}\\
            \cline{2-4} \cellcolor[HTML]{D7D7D7}& $q_{48}$ & \texttt{$trip\_distance \leq 1$ AND $rateCodeID = 1$ AND $payment\_type = 2$} & \begin{tabular}{@{}c@{}} $1651975$ \\ $1685971$ \end{tabular}\\
            \cline{2-4} \cellcolor[HTML]{D7D7D7}& $q_{49}$ & \texttt{$trip\_distance \geq 8$ AND $rateCodeID = 2$ AND $payment\_type = 1$} & \begin{tabular}{@{}c@{}} $3628304$ \\ $3633347$ \end{tabular}\\
            \cline{2-4} \multirow{-20}{*}{\cellcolor[HTML]{D7D7D7} {\begin{tabular}{@{}c@{}} \sumQuery\\ on\\ $total\_amount$ \end{tabular}}}& $q_{50}$ &  \texttt{$trip\_distance \leq 1$ AND $fare\_amount \leq 6$ AND $congestion\_surcharge = 2$} & \begin{tabular}{@{}c@{}} $5680776$ \\ $5624898$ \end{tabular}\\
            
            \hline \cellcolor[HTML]{D7D7D7}& $q_{51}$ & \texttt{$trip\_distance \leq 2$ AND $total\_amount \leq 15$ AND $congestion\_surcharge = 0$} & \begin{tabular}{@{}c@{}} $1$ \\ $1$ \end{tabular}\\
            \cline{2-4} \cellcolor[HTML]{D7D7D7}& $q_{52}$ & \texttt{$trip\_distance \leq 1$ AND $rateCodeID = 1$ AND $payment\_type = 2$} & \begin{tabular}{@{}c@{}} $1$ \\ $1$ \end{tabular}\\
            \cline{2-4} \cellcolor[HTML]{D7D7D7}& $q_{53}$ & \texttt{$trip\_distance \leq 2$ AND $rateCodeID = 2$ AND $payment\_type = 1$} & \begin{tabular}{@{}c@{}} $1$ \\ $1$ \end{tabular}\\
            \cline{2-4} \cellcolor[HTML]{D7D7D7}& $q_{54}$ & \texttt{$passenger\_count \leq 3$ AND $rateCodeID = 1$ AND $payment\_type = 1$} & \begin{tabular}{@{}c@{}} $2$ \\ $2$ \end{tabular}\\
            \cline{2-4} \cellcolor[HTML]{D7D7D7}& $q_{55}$ & \texttt{$vendorID = 2$ AND $store\_and\_fwd\_flag$ LIKE `N' AND $payment\_type = 2$} & \begin{tabular}{@{}c@{}} $2$ \\ $2$ \end{tabular}\\
            \cline{2-4} \cellcolor[HTML]{D7D7D7}& $q_{56}$ & \texttt{$trip\_distance \leq 8$ AND $fare\_amount \leq 25$ AND $congestion\_surcharge = 2$} & \begin{tabular}{@{}c@{}} $2$ \\ $2$ \end{tabular}\\
            \cline{2-4} \cellcolor[HTML]{D7D7D7}& $q_{57}$ & \texttt{$trip\_distance \leq 12$ AND $total\_amount \leq 34$ AND $congestion\_surcharge = 0$} & \begin{tabular}{@{}c@{}} $3$ \\ $2$ \end{tabular}\\
            \cline{2-4} \cellcolor[HTML]{D7D7D7}& $q_{58}$ & \texttt{$passenger\_count < 2$ AND $trip\_distance = 8$ AND $tip\_amount \leq 2$} & \begin{tabular}{@{}c@{}} $8$ \\ $8$ \end{tabular}\\
            \cline{2-4} \cellcolor[HTML]{D7D7D7}& $q_{59}$ & \texttt{$passenger\_count \geq 2$ AND $trip\_distance \geq 8$ AND $tip\_amount = 5$} & \begin{tabular}{@{}c@{}} $12$ \\ $12$ \end{tabular}\\
            \cline{2-4} \multirow{-20}{*}{\cellcolor[HTML]{D7D7D7} {\begin{tabular}{@{}c@{}} \medianQuery\\ on\\ $trip\_distance$ \end{tabular}}}& $q_{60}$ &  \texttt{$trip\_distance \geq 2$ AND $payment\_type = 2$ AND $fare\_amount \geq 50$} & \begin{tabular}{@{}c@{}} $18$ \\ $18$ \end{tabular}\\
            
            \hline
        \end{tabular}
    }
    \caption{\reva{Queries used in experiments for the database derived from the NYC Yellow Taxi Trip data}.} \label{tbl:exp2_queries}
\end{table*}

\section{Additional Experiments} \label{sec:appendix_allQueryErrorPlots}

\textbf{Accuracy Analysis: Error Plots.~}
We give the complete list of queries used in the experiments in \reva{Tables~\ref{tbl:exp_queries}--\ref{tbl:exp2_queries}}. 
\reva{For the database derived from the IPUMS-CPS survey data, we} used $12$ \countQuery, $9$ \sumQuery, and $9$ \medianQuery\ queries\reva{, whereas for the database derived from the NYC Yellow Taxi Trip data, we used $10$ of each of \countQuery, \sumQuery, and \medianQuery\ queries}. Note that the query answers on the private database $D$ are private (hidden from the user) and we include them here only for reference.
The \DLS\ values for queries $q_{13}$ to $q_{21}$ are: $127764, 359892, 439420, 500000,$ $384842, 278011, 276799, 403353$ and $500000$\reva{, and for $q_{41}$ to $q_{50}$ are: $73, 53, 138, 157, 82, 15, 266, 48, 155$ and $134$.}

\reva{For the database derived from the IPUMS-CPS survey data, we}
plot the errors from \lapcount\ and \emcount\ for the $12$ \countQuery\ queries as $\tau$ varies in Figure~\ref{fig:appendix_count_vary_tau}, and as $\epsilon$ varies in Figure~\ref{fig:appendix_count_vary_eps}.
We plot the errors from \lapsum, \rtwotsum\ and \svtsum\ for the $9$ \sumQuery\ queries as $\tau$ varies in \reva{Figure~\ref{fig:appendix_sum_vary_tau}}, and as $\epsilon$ varies in \reva{Figure~\ref{fig:appendix_sum_vary_eps}}.
We plot the errors from \emmed\ and \histmed\ for the $9$ \medianQuery\ queries as $\tau$ varies in \reva{Figure~\ref{fig:appendix_median_vary_tau}}, and as $\epsilon$ varies in \reva{Figure~\ref{fig:appendix_median_vary_eps}}.

\reva{For the database derived from the NYC Yellow Taxi Trip data,
we plot the errors from \lapcount\ and \emcount\ for the $10$ \countQuery\ queries as $\tau$ varies in Figure~\ref{fig:appendix_exp2_count_vary_tau}, and as $\epsilon$ varies in Figure~\ref{fig:appendix_exp2_count_vary_eps}.
We plot the errors from \lapsum, \rtwotsum\ and \svtsum\ for the $10$ \sumQuery\ queries as $\tau$ varies in Figure~\ref{fig:appendix_exp2_sum_vary_tau}, and as $\epsilon$ varies in Figure~\ref{fig:appendix_exp2_sum_vary_eps}.
We plot the errors from \emmed\ and \histmed\ for the $10$ \medianQuery\ queries as $\tau$ varies in Figure~\ref{fig:appendix_exp2_median_vary_tau}, and as $\epsilon$ varies in Figure~\ref{fig:appendix_exp2_median_vary_eps}.}

\reva{\textbf{Pre-processing of the NYC Taxi data.~} As a data cleaning step, for numerical attributes, we round to nearest integer and set domains as $[1, 6]$ for $passenger\_count$, $[1, 100]$ for $trip\_distance$, $[4, 250]$ for $fare\_amount$, $[0, 125]$ for $tip\_amount$, and $[8, 377]$ for $total\_amount$. }

\begin{figure*}[ht]
  \includegraphics[width=0.89\textwidth]{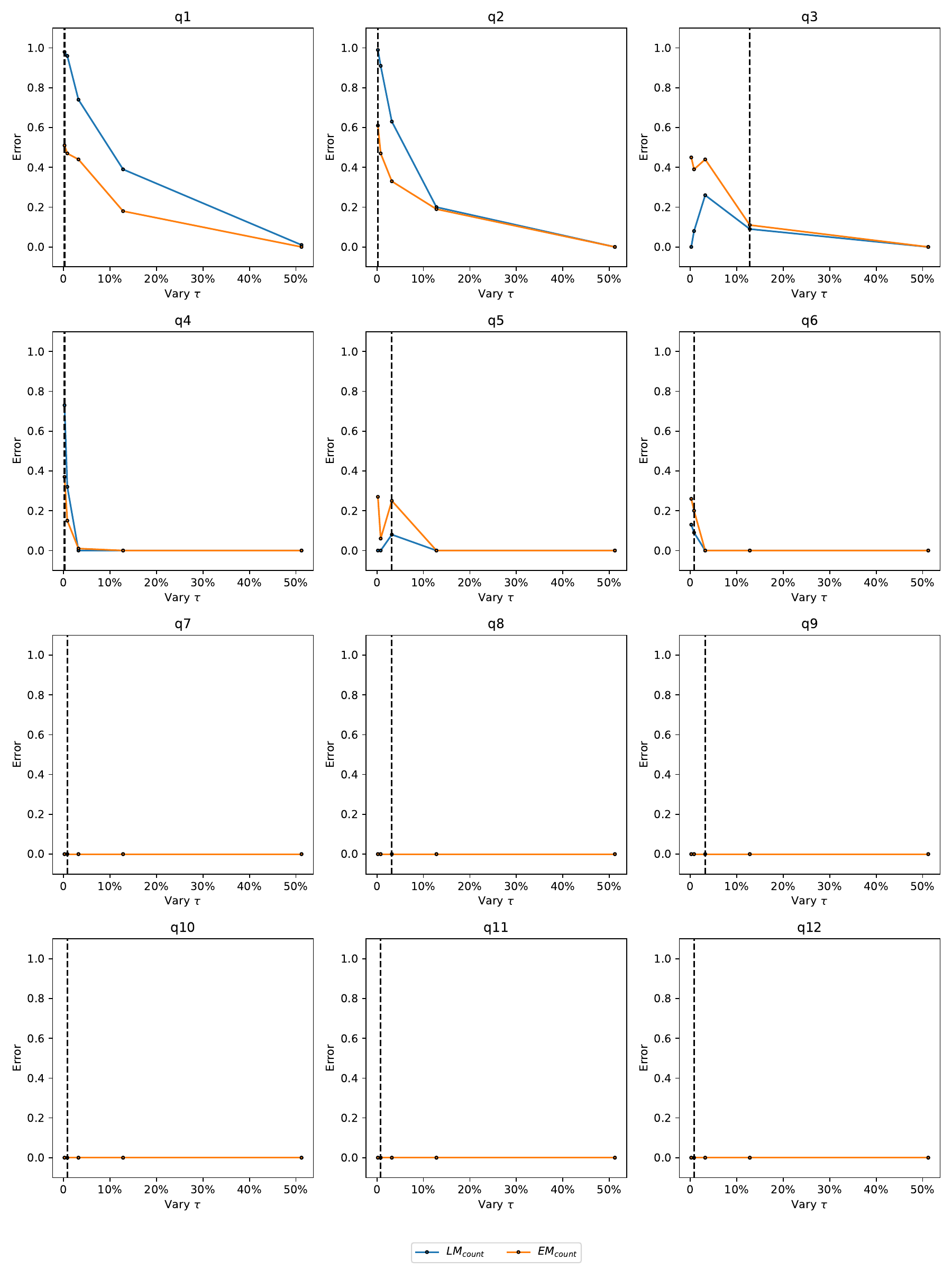}
  \caption{\reva{IPUMS-CPS:} Error for \countQuery\ queries from \lapcount\ and \emcount\ as $\tau$ varies. \revc{The dotted line marks the smallest $\tau$ value considered such that the query answer on the private data belongs in the interval \intvl.}}
  \label{fig:appendix_count_vary_tau}
\end{figure*}

\begin{figure*}[ht]
  \includegraphics[width=0.89\textwidth]{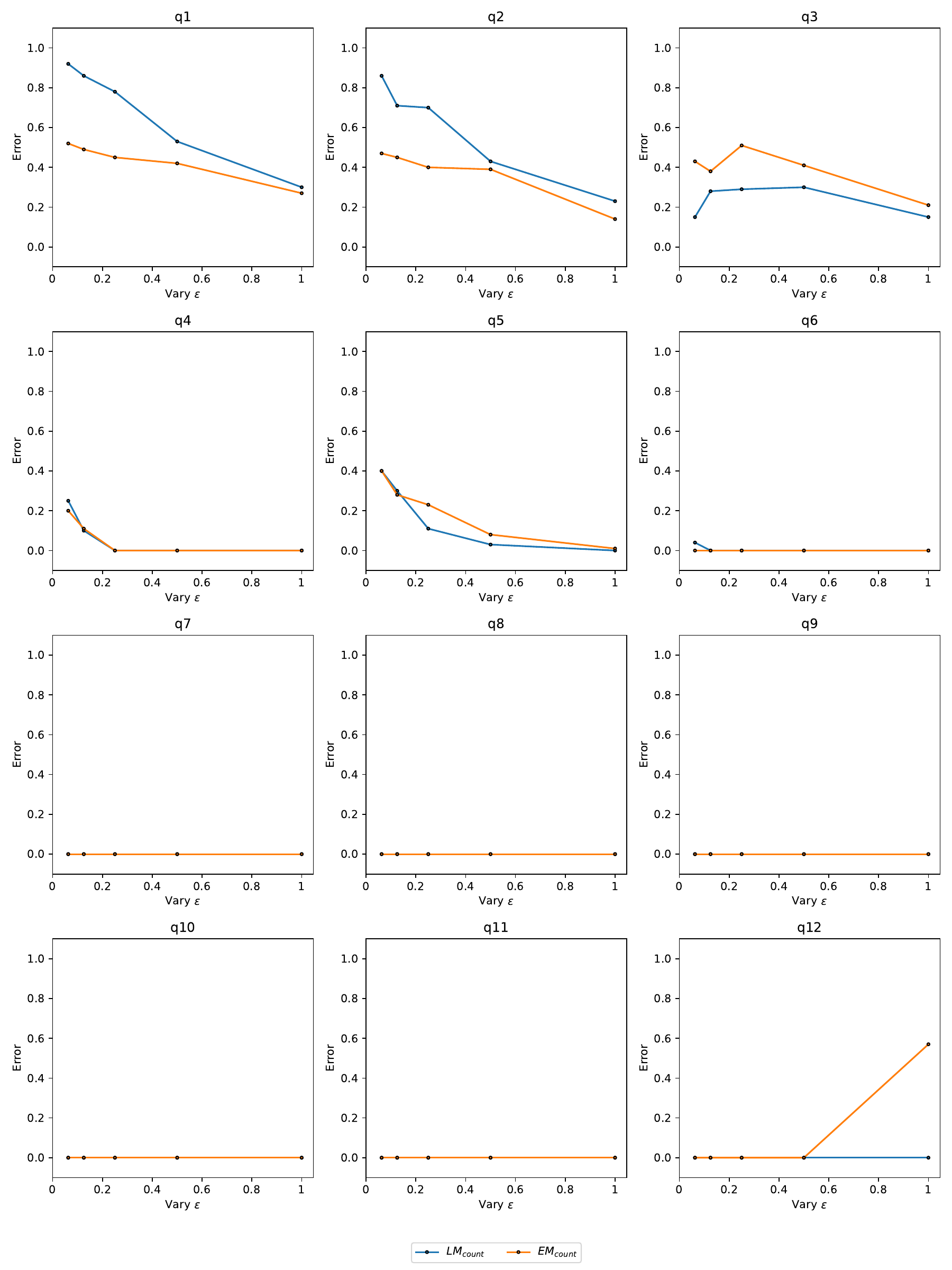}
  \caption{\reva{IPUMS-CPS:} Error for \countQuery\ queries from \lapcount\ and \emcount\ as $\epsilon$ varies.}
  \label{fig:appendix_count_vary_eps}
\end{figure*}

\begin{figure*}[ht]
  \includegraphics[width=0.89\textwidth]{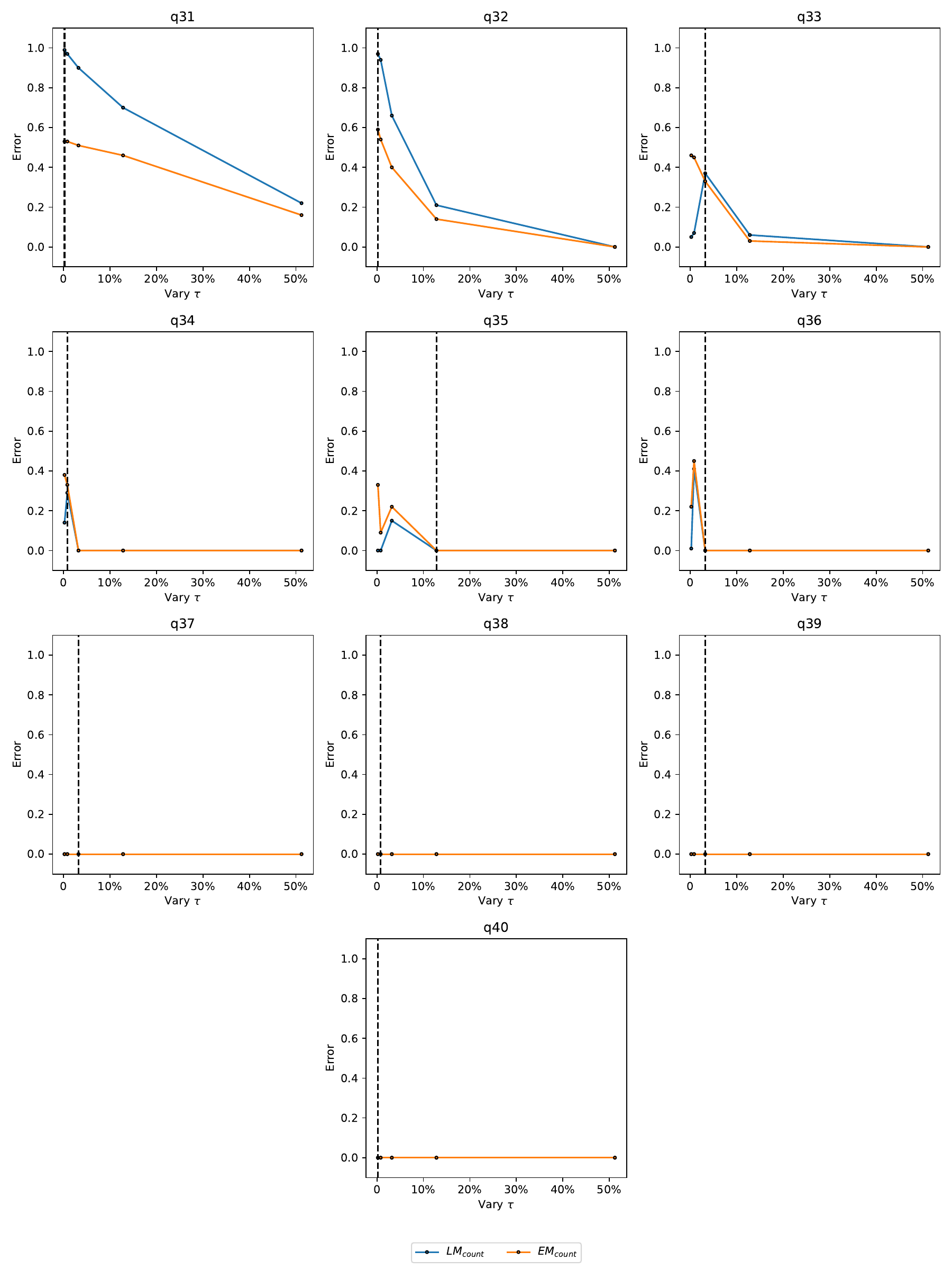}
  \caption{\reva{NYC Yellow Taxi Trip: Error for \countQuery\ queries from \lapcount\ and \emcount\ as $\tau$ varies.} \revc{The dotted line marks the smallest $\tau$ value considered such that the query answer on the private data belongs in the interval \intvl.}}
  \label{fig:appendix_exp2_count_vary_tau}
\end{figure*}

\begin{figure*}[ht]
  \includegraphics[width=0.89\textwidth]{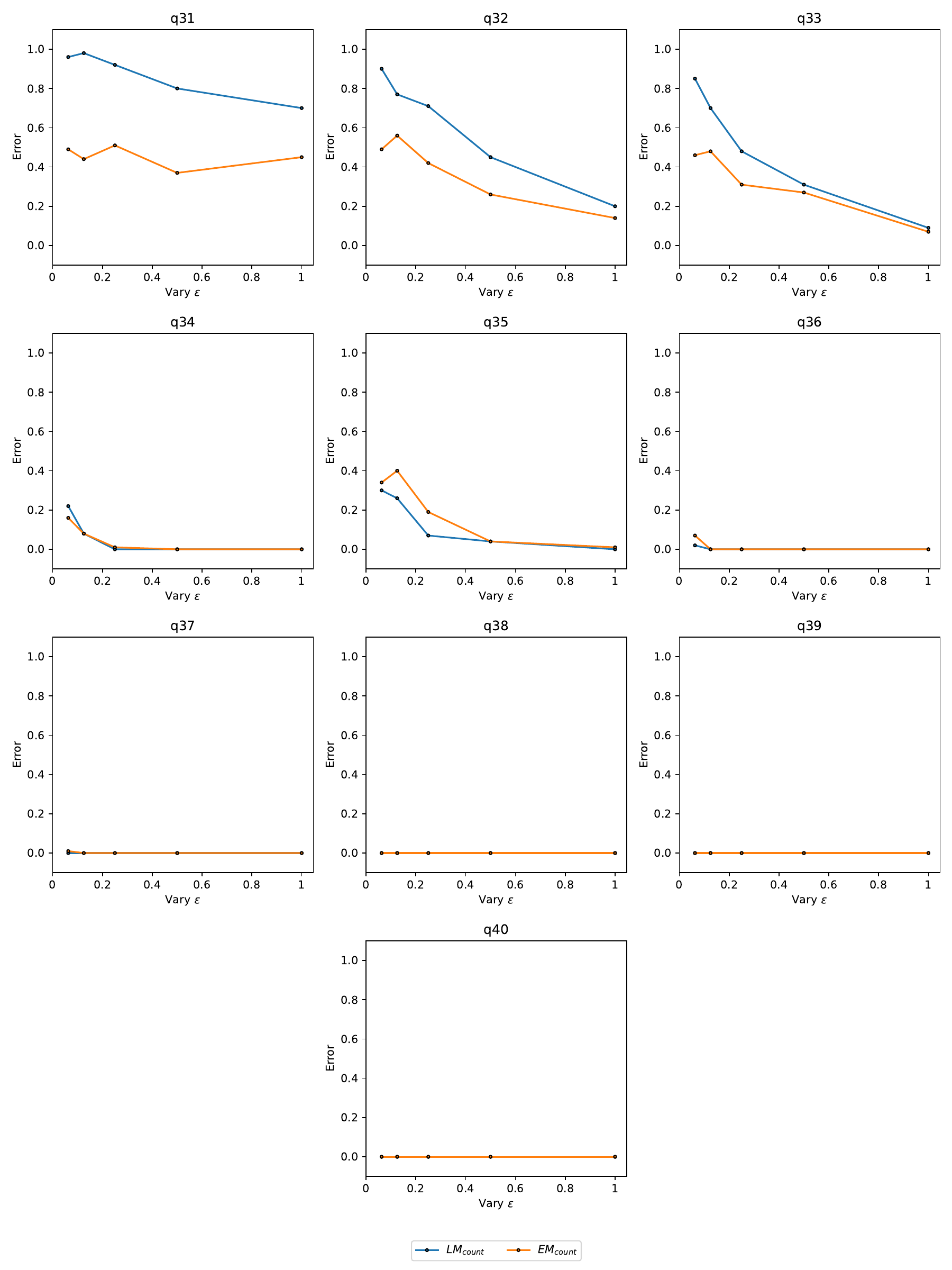}
  \caption{\reva{NYC Yellow Taxi Trip: Error for \countQuery\ queries from \lapcount\ and \emcount\ as $\epsilon$ varies.}}
  \label{fig:appendix_exp2_count_vary_eps}
\end{figure*}

\begin{figure*}[ht]
  \includegraphics[width=0.89\textwidth]{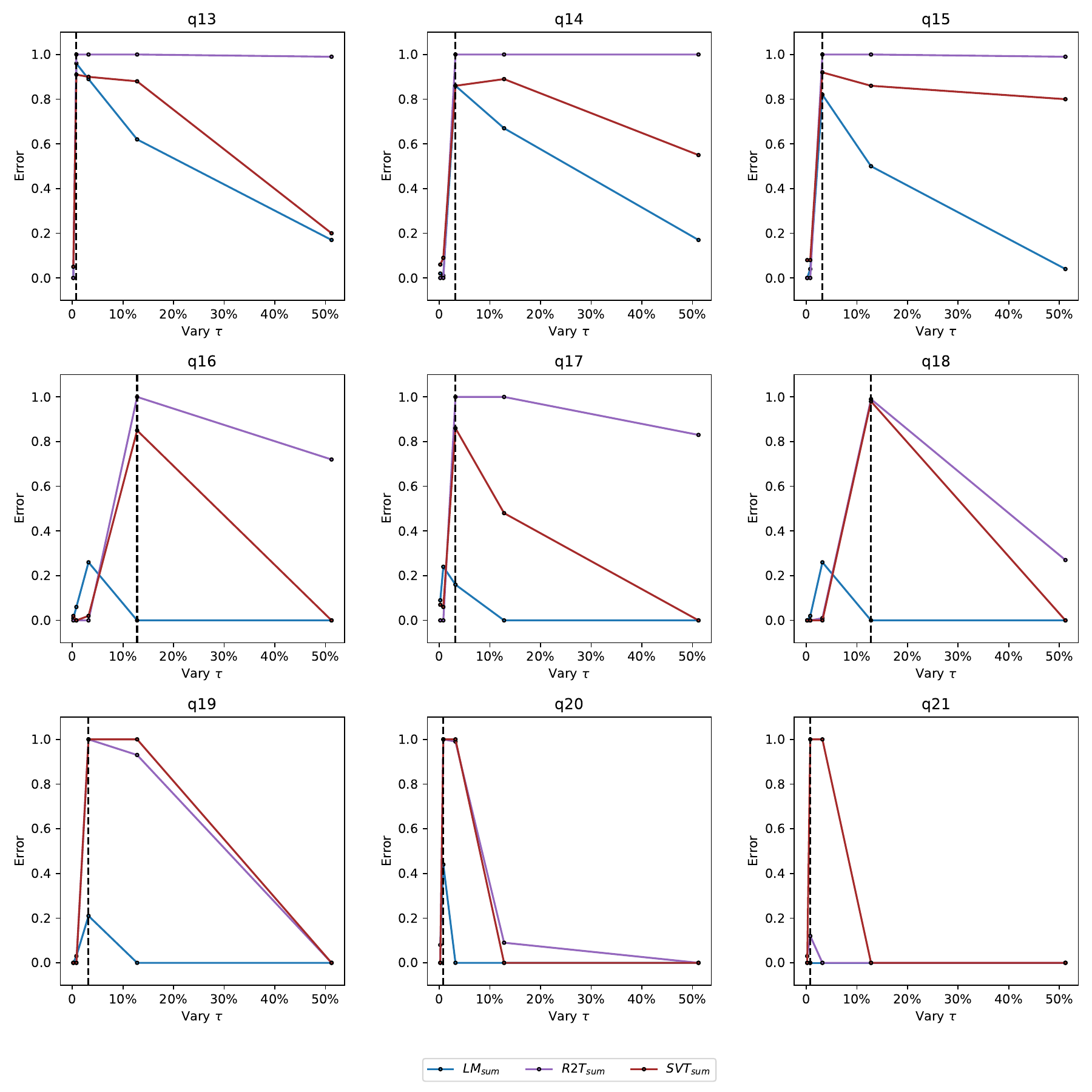}
  \caption{\reva{IPUMS-CPS:} Error for \sumQuery\ queries from \lapsum, \rtwotsum\ and \svtsum\ as $\tau$ varies. \revc{The dotted line marks the smallest $\tau$ value considered such that the query answer on the private data belongs in the interval \intvl.}}
  \label{fig:appendix_sum_vary_tau}
\end{figure*}

\begin{figure*}[ht]
  \includegraphics[width=0.89\textwidth]{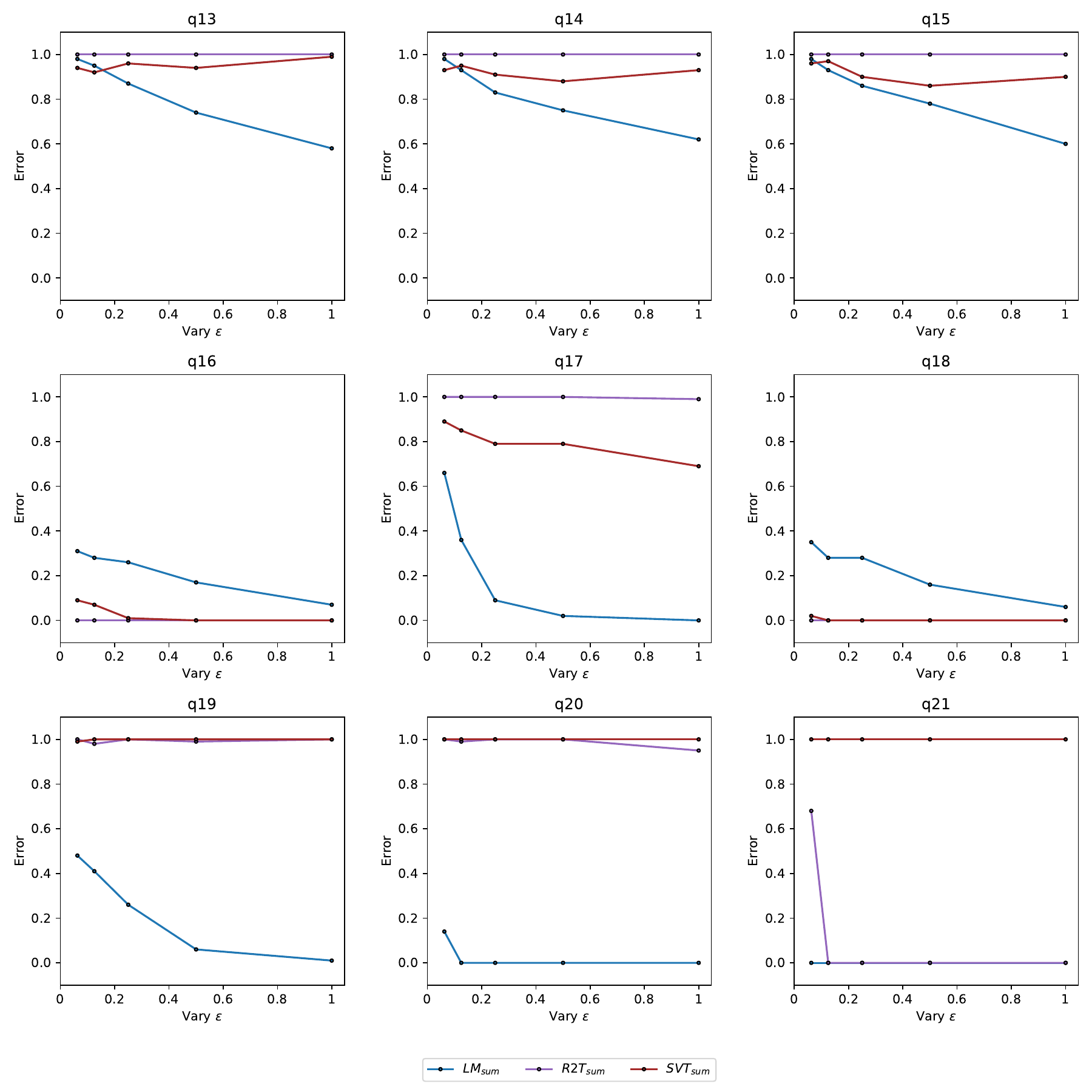}
  \caption{\reva{IPUMS-CPS:} Error for \sumQuery\ queries from \lapsum, \rtwotsum\ and \svtsum\ as $\epsilon$ varies.}
  \label{fig:appendix_sum_vary_eps}
\end{figure*}

\begin{figure*}[ht]
  \includegraphics[width=0.89\textwidth]{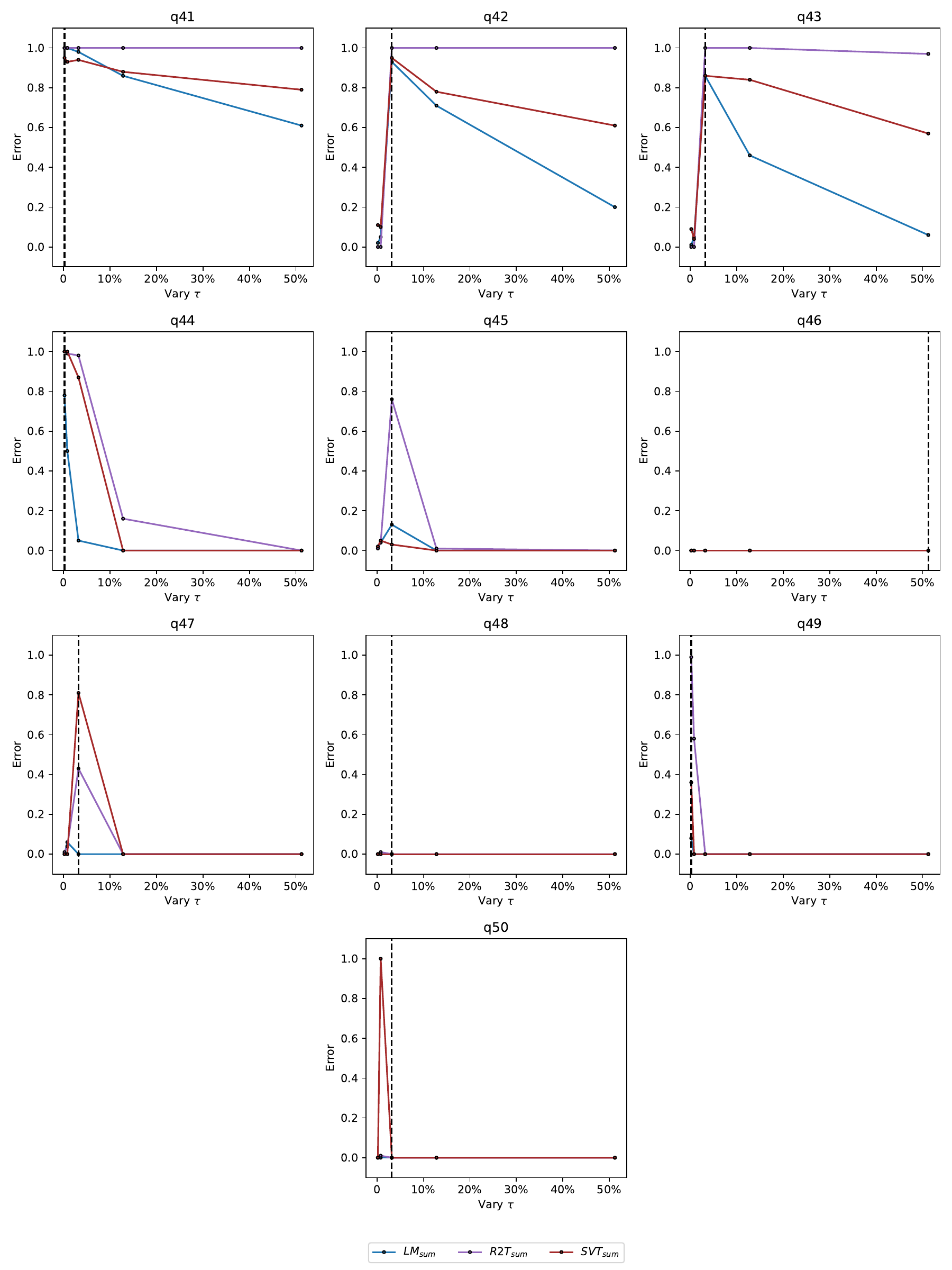}
  \caption{\reva{NYC Yellow Taxi Trip: Error for \sumQuery\ queries from \lapsum, \rtwotsum\ and \svtsum\ as $\tau$ varies.} \revc{The dotted line marks the smallest $\tau$ value considered such that the query answer on the private data belongs in the interval \intvl.}}
  \label{fig:appendix_exp2_sum_vary_tau}
\end{figure*}

\begin{figure*}[ht]
  \includegraphics[width=0.89\textwidth]{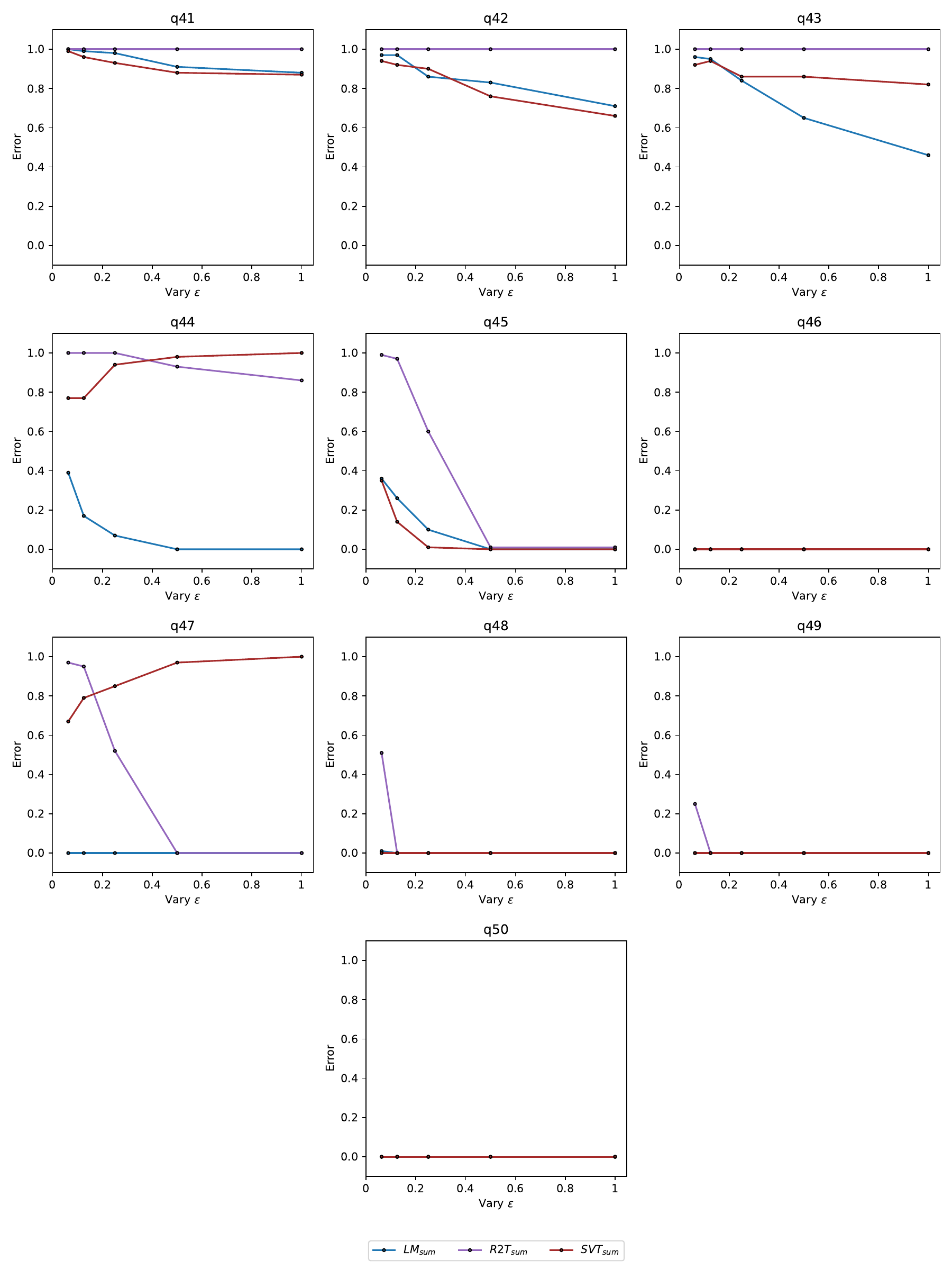}
  \caption{\reva{NYC Yellow Taxi Trip: Error for \sumQuery\ queries from \lapsum, \rtwotsum\ and \svtsum\ as $\epsilon$ varies.}}
  \label{fig:appendix_exp2_sum_vary_eps}
\end{figure*}

\begin{figure*}[ht]
  \includegraphics[width=0.89\textwidth]{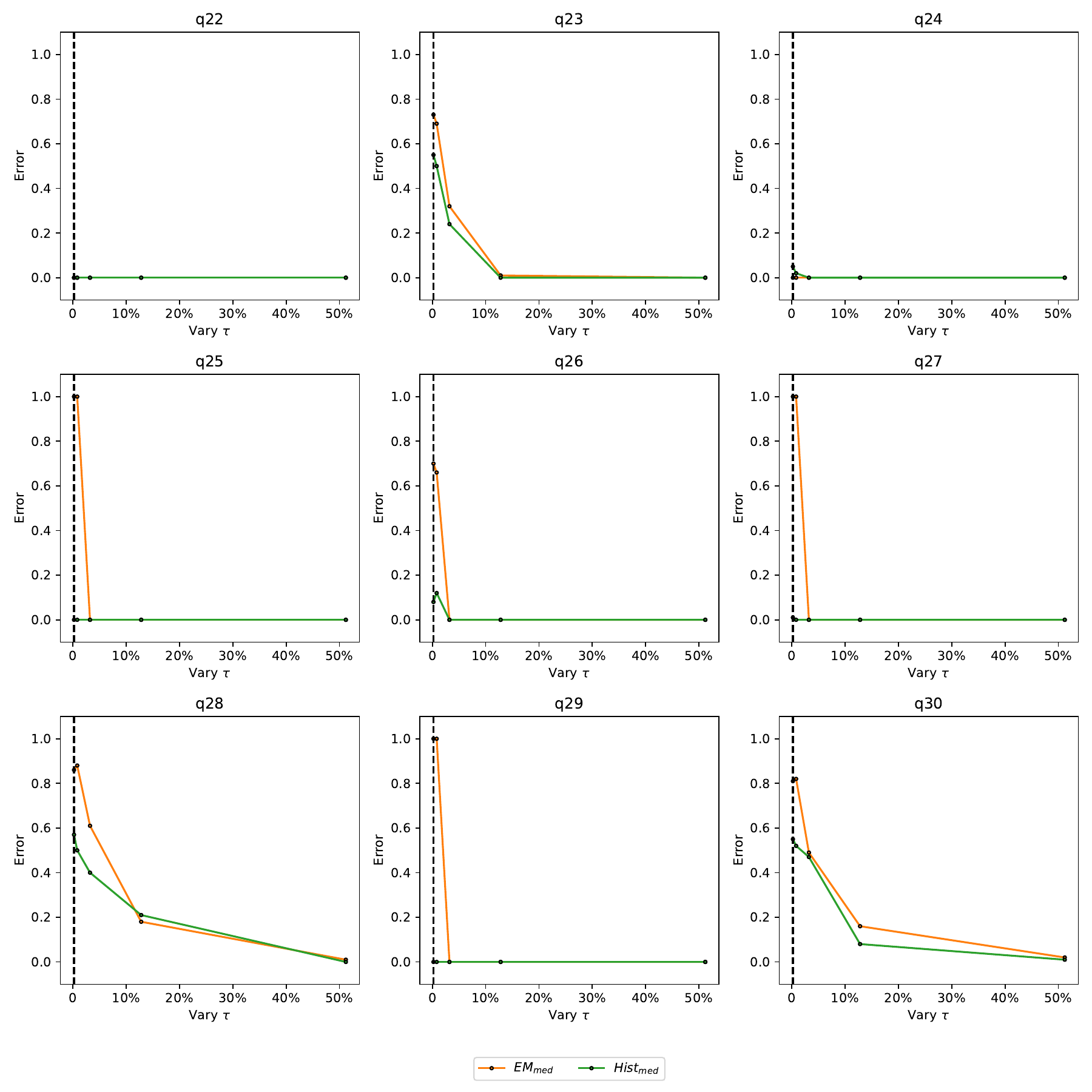}
  \caption{\reva{IPUMS-CPS:} Error for \medianQuery\ queries from \emmed\ and \histmed\ as $\tau$ varies. \revc{The dotted line marks the smallest $\tau$ value considered such that the query answer on the private data belongs in the interval \intvl.}}
  \label{fig:appendix_median_vary_tau}
\end{figure*}

\begin{figure*}[ht]
  \includegraphics[width=0.89\textwidth]{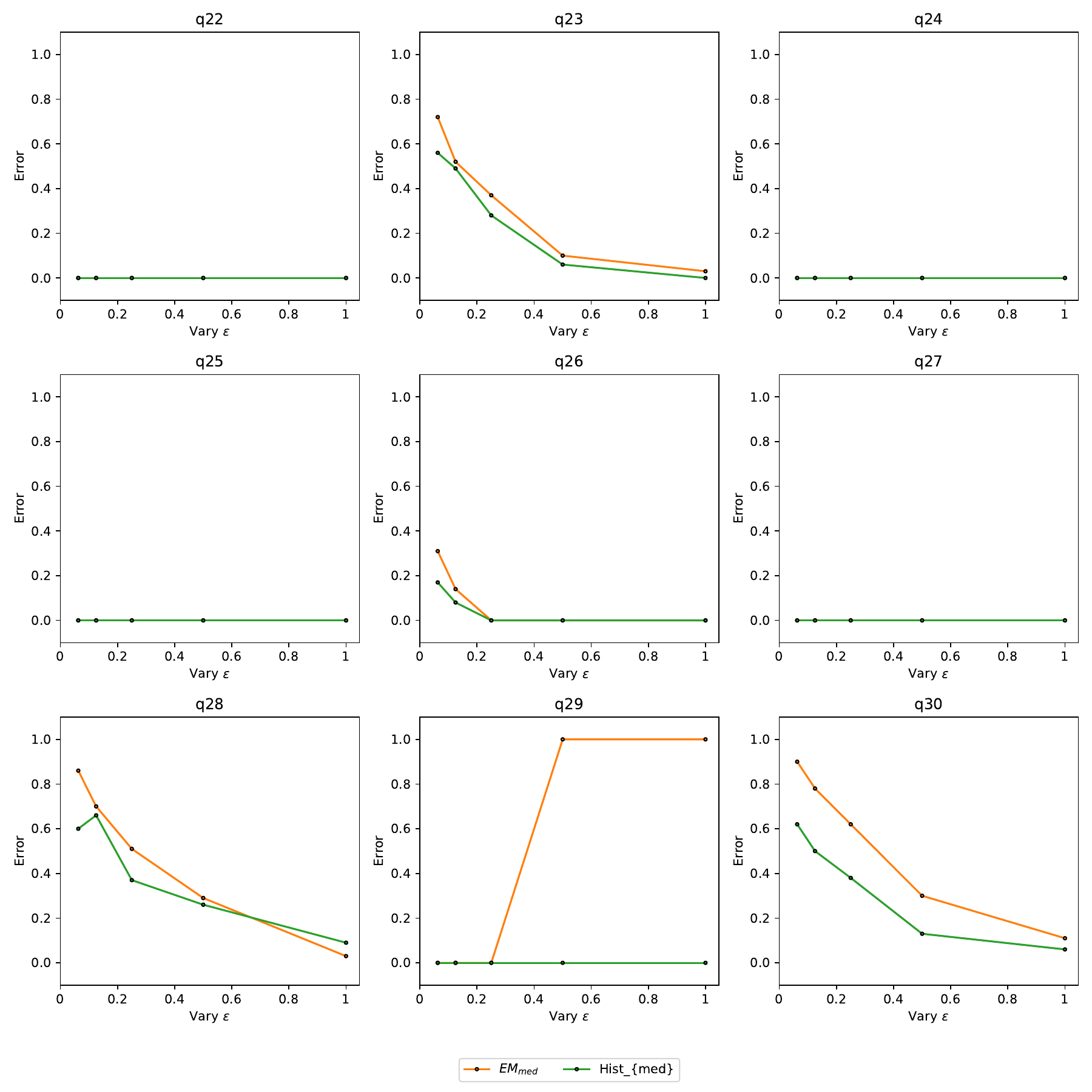}
  \caption{\reva{IPUMS-CPS:} Error for \medianQuery\ queries from \emmed\ and \histmed\ as $\epsilon$ varies.}
  \label{fig:appendix_median_vary_eps}
\end{figure*}

\begin{figure*}[ht]
  \includegraphics[width=0.89\textwidth]{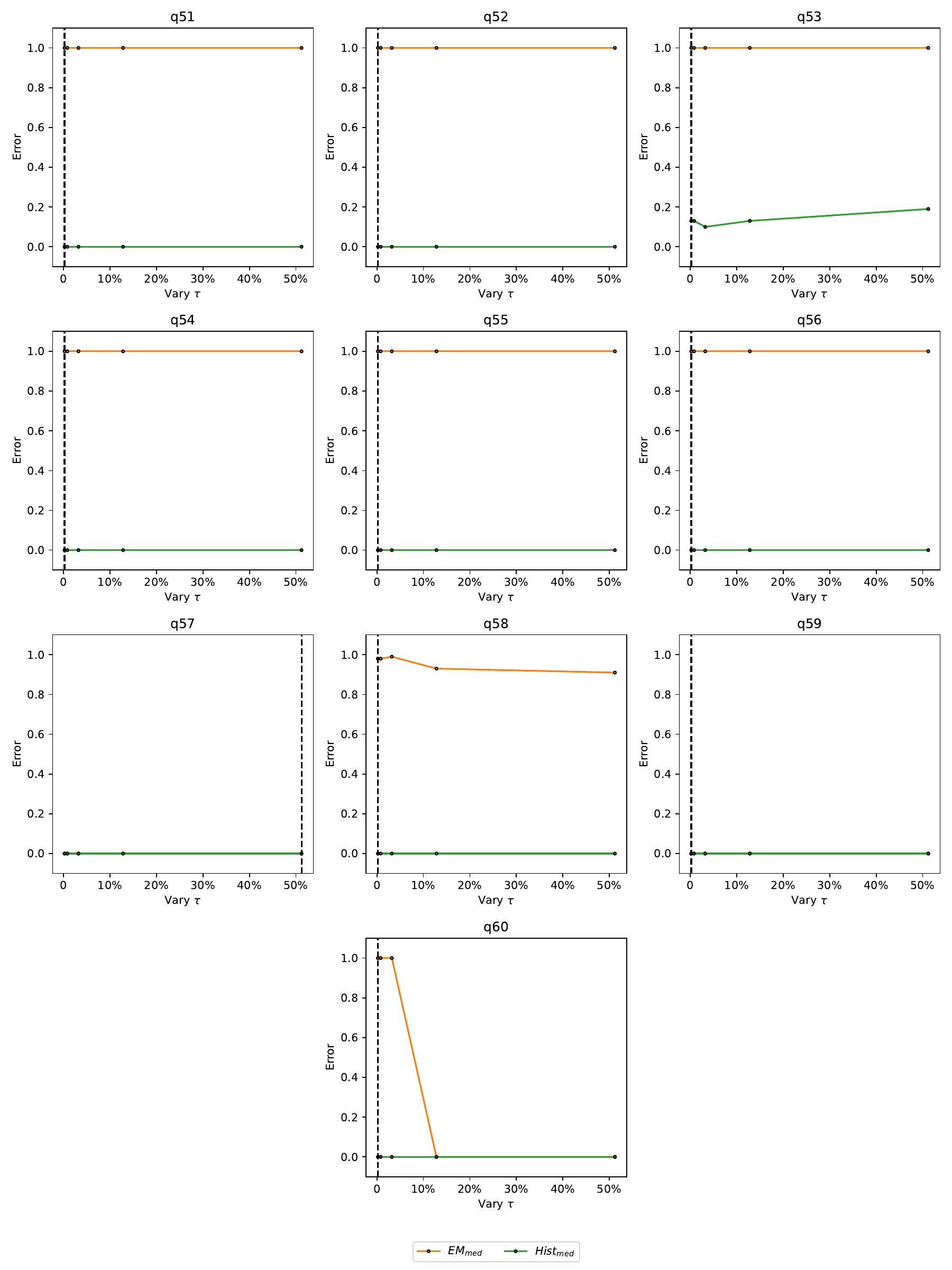}
  \caption{\reva{NYC Yellow Taxi Trip: Error for \medianQuery\ queries from \emmed\ and \histmed\ as $\tau$ varies.} \revc{The dotted line marks the smallest $\tau$ value considered such that the query answer on the private data belongs in the interval \intvl.}}
  \label{fig:appendix_exp2_median_vary_tau}
\end{figure*}

\begin{figure*}[ht]
  \includegraphics[width=0.89\textwidth]{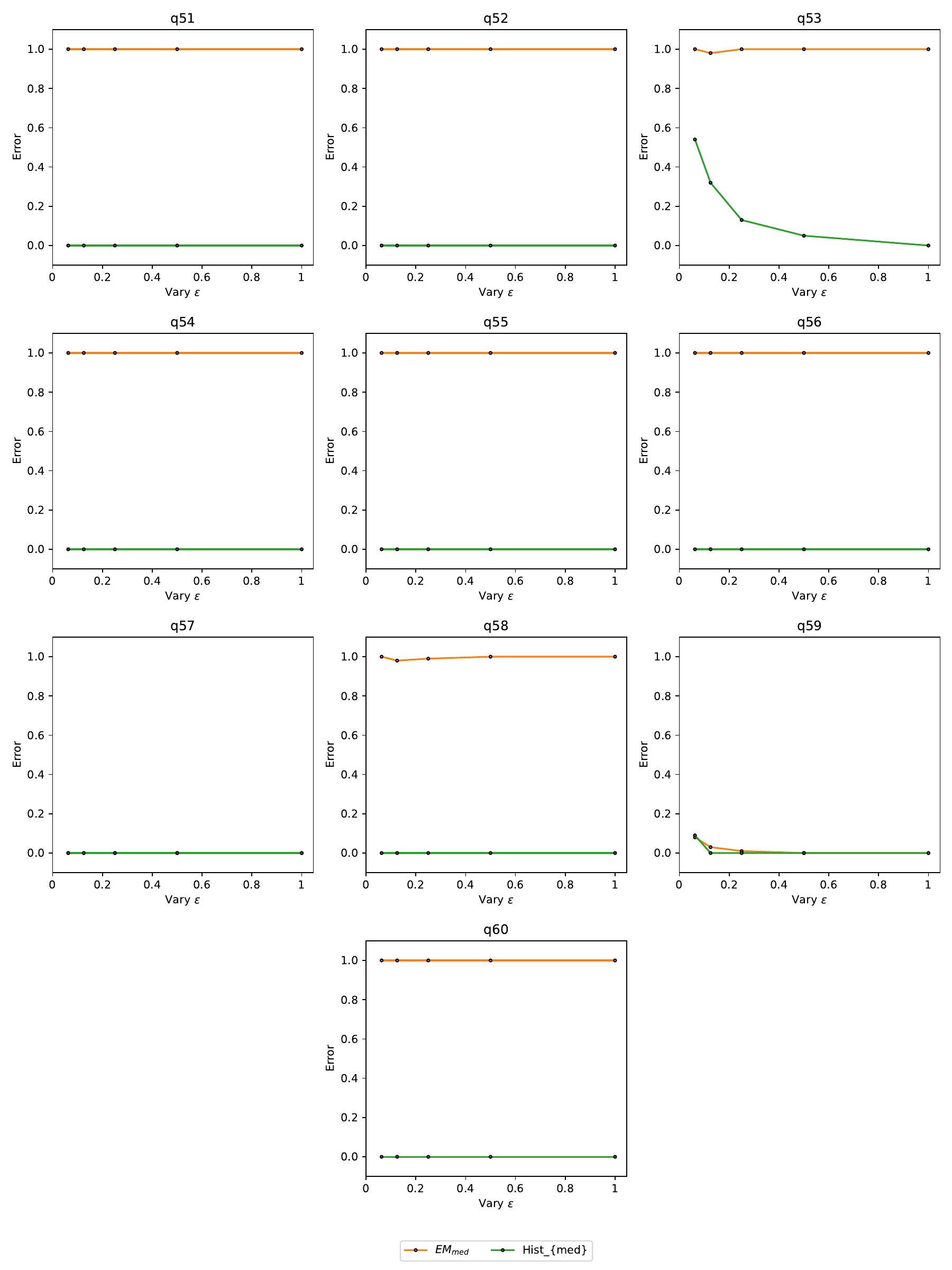}
  \caption{\reva{NYC Yellow Taxi Trip: Error for \medianQuery\ queries from \emmed\ and \histmed\ as $\epsilon$ varies.}}
  \label{fig:appendix_exp2_median_vary_eps}
\end{figure*}

\end{document}